\documentclass[twocolumn]{svjour3}          
\usepackage{graphics}
\usepackage{graphicx}
\usepackage{subfigure}

\usepackage{xspace}
\usepackage{amsmath}
\usepackage{amssymb}
\usepackage{mathrsfs}
\usepackage{array}
\usepackage{weiwAlgorithm}
\usepackage{url}
\usepackage{balance}

\usepackage[misc]{ifsym} 
\usepackage{url}
\newskip\subfigtoppskip \subfigtopskip = -1mm
\newskip\subfigcapskip \subfigcapskip = -1mm
\DeclareMathOperator*{\argmin}{arg\,min}

\newcommand*\circled[1]{\raisebox{.5pt}{\textcircled{\raisebox{-.9pt}{#1}}}}

\newcommand{\sk}{spatial-keyword\xspace}
\newcommand{\gt}{geo-textual\xspace}

\newcommand{\topk}{top-$k$\xspace}

\newcommand{\pubsub}{publish\slash subscribe\xspace}

\newcommand{\ours}{\textsf{Skype}\xspace}
\newcommand{\moduleone}{message dissemination module\xspace}
\newcommand{\moduletwo}{\topk re-evaluation module\xspace}

\newcommand{\csky}{\textsf{cSkyband}\xspace}
\newcommand{\naive}{\textsf{Skyband}\xspace}
\newcommand{\kmax}{\textsf{kmax}\xspace}

\newcommand{\ipt}{\textsf{IPT}\xspace}
\newcommand{\igpt}{\textsf{IGPT}\xspace}
\newcommand{\ciq}{\textsf{CIQ}\xspace}

\newcommand{\igptcsky}{\textsf{IGPT-cSkyband}\xspace}
\newcommand{\igptnaive}{\textsf{IGPT-Skyband}\xspace}
\newcommand{\igptkmax}{\textsf{IGPT-Kmax}\xspace}

\newcommand{\ciqkmax}{\textsf{CIQ-Kmax}\xspace}

\newcommand{\rtree}{\textsf{R-Tree}\xspace}
\newcommand{\qd}{\textsf{Quadtree}\xspace}
\newcommand{\ifile}{inverted file\xspace}

\newcommand{\skyband}{$k$-skyband\xspace}

\newcommand{\irtree}{\textsf{IR-Tree}\xspace}
\newcommand{\sti}{\textsf{S2I}\xspace}

\newcommand{\tweets}{\textit{TWEETS}\xspace}
\newcommand{\gn}{\textit{GN}\xspace}
\newcommand{\yelp}{\textit{YELP}\xspace}

\newcommand{\amp}{\textit{AMP}\xspace}
\newcommand{\emp}{\textit{EMP}\xspace}

\newcommand{\score}{\textsf{Score}}
\newcommand{\tsim}{\textsf{TSim}}
\newcommand{\ssim}{\textsf{SSim}}
\newcommand{\wt}{\textsf{wt}}
\newcommand{\maxwt}{\textsf{maxwt}}
\newcommand{\wtsum}{\textsf{wtsum}}
\newcommand{\kscore}{\textsf{kScore}}
\newcommand{\tsimlb}{\textsf{$\lambda_{T}$}}
\newcommand{\ssimub}{\textsf{SSimUB}}
\newcommand{\tsimub}{\textsf{TSimUB}}
\newcommand{\ssimlb}{\textsf{$\lambda_{S}$}}
\newcommand{\prob}{\textsf{prob}}
\newcommand{\pref}{\textsf{pref}}
\newcommand{\prefplus}{\textsf{pref$_{+}$}}

\newcommand{\zipf}{\textsf{zipf}\xspace}

\newcommand{\storm}{\textsf{Storm}\xspace}
\newcommand{\hadoop}{\textsf{Hadoop}\xspace}

\newcommand{\sparkstreaming}{\textsf{Spark Streaming}\xspace}
\newcommand{\samza}{\textsf{Samza}\xspace}

\newcommand{\oursdis}{\textsf{DSkype}\xspace}

\newcommand{\kdtree}{\textsf{KD-Tree}\xspace}

\newcommand{\dmhashing}{\textsf{hashing-based}\xspace}
\newcommand{\dmlocation}{\textsf{location-based}\xspace}
\newcommand{\dmkeyword}{\textsf{keyword-based}\xspace}
\newcommand{\dmprefix}{\textsf{prefix-based}\xspace}

\newcommand{\tornado}{\textsf{Tornado}\xspace}

\newcommand{\sfirst}{\textsf{spatial-first}\xspace}
\newcommand{\kfirst}{\textsf{keyword-first}\xspace}

\begin{document}

\title{Top-k Spatial-keyword Publish/Subscribe Over Sliding Window
}

\titlerunning{Top-k Spatial-keyword Publish/Subscribe Over Sliding Window}        

\author{Xiang Wang	\and
		Ying Zhang	\and
		Wenjie Zhang	\and
		Xuemin Lin	\and
		Zengfeng Huang
}

\authorrunning{X. Wang et al.} 

\institute{X. Wang (\Letter) \and W. Zhang \and X. Lin \and Z. Huang \at 	
						School of Computer Science and Engineering,\\
              			The University of New South Wales\\
              			Sydney, Australia\\
              			\email{xiangw@cse.unsw.edu.au}           
           \and
           Y. Zhang \at CAI, University of Technology,\\
           				Sydney, Australia\\
              			\email{ying.zhang@uts.edu.au}           
           \and
           W. Zhang \at \email{zhangw@cse.unsw.edu.au}
           \and
           X. Lin \at \email{lxue@cse.unsw.edu.au}
           \and
           Z. Huang \at \email{huangzengfeng@gmail.com}
}


\maketitle

\begin{abstract}
With the prevalence of social media and GPS-enabled devices, a massive amount of \textit{geo-textual} data has been generated in a stream fashion,
leading to a variety of applications such as location-based recommendation and information dissemination.
In this paper, we investigate a novel real-time \topk monitoring problem over sliding window of streaming data;
that is, we continuously maintain the \textit{top-k} most relevant \textit{\gt messages} (e.g., geo-tagged tweets) for a large number of \textit{\sk subscriptions} (e.g., registered users interested in \textit{local events}) simultaneously.
To provide the most recent information under controllable memory cost, sliding window model is employed on the streaming \gt data.
To the best of our knowledge, this is the first work to study \topk \sk \pubsub over sliding window.
A novel centralized system, called \ours (Top-k \textbf{S}patial-\textbf{k}e\textbf{y}word \textbf{P}ublish\slash Subscrib\textbf{e}), is proposed in this paper.
In \ours, to continuously maintain \topk results for massive subscriptions, we devise a novel indexing structure upon subscriptions such that each incoming message can be immediately delivered on its arrival.
To reduce the expensive \topk re-evaluation cost triggered by message expiration, we develop a novel \textit{cost-based \skyband} technique to reduce the number of re-evaluations in a cost-effective way.
Extensive experiments verify the great efficiency and effectiveness of our proposed techniques.
Furthermore, to support better scalability and higher throughput, we propose a distributed version of \ours, namely, \oursdis, on top of \storm, which is a popular distributed stream processing system.
With the help of fine-tuned subscription/message distribution mechanisms, \oursdis can achieve orders of magnitude speed-up than its centralized version.

\end{abstract}

\section{Introduction}
\label{sec:intro}

Recently, with the ubiquity of social media and GPS-enabled mobile devices,
large volumes of \gt data have been generated in a stream fashion, leading to the popularity of \textit{\sk \pubsub system} (e.g.,~\cite{DBLP:conf/kdd/LiWWF13,DBLP:conf/sigmod/ChenCC13,DBLP:conf/icde/WangZZLW15,DBLP:conf/icde/HuLLFT15,DBLP:conf/icde/ChenCCT15}) in a variety of applications such as location-based recommendation and social network.
In such a system, each individual user can register her interest (e.g., favorite food or sports) and location as a \textit{\sk subscription}.
A stream of \textit{\gt messages} (e.g., e-coupon promotion and tweets with location information) continuously generated by publishers (e.g., local business) are rapidly fed to the relevant users.

\begin{figure}[t]
\centering
\includegraphics[width=0.9\columnwidth]{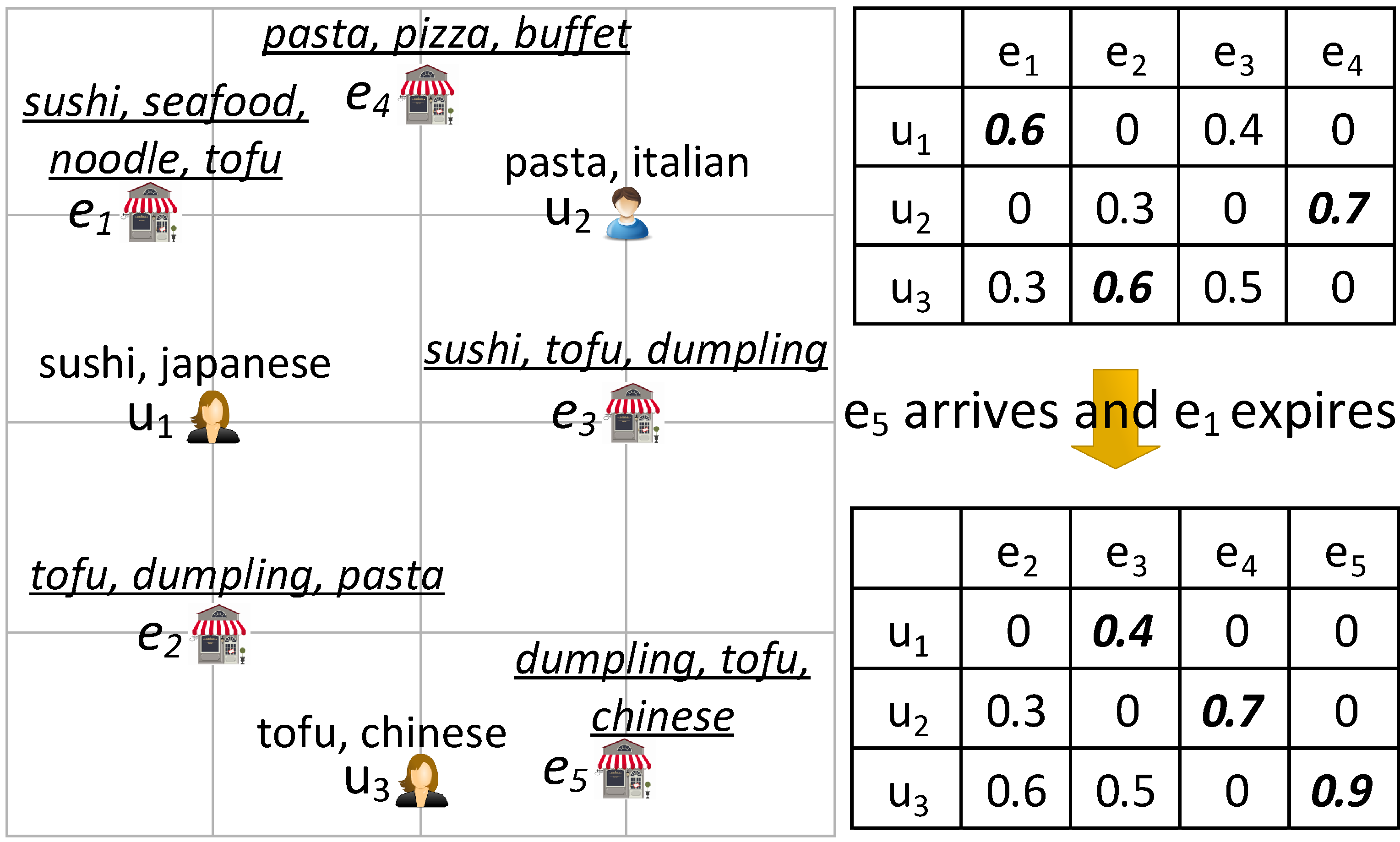}
\caption{\small{E-coupon recommendation system}}
\label{fig:intro:example_real}
\end{figure}

The \sk \pubsub system has been studied in several existing work (e.g.,~\cite{DBLP:conf/kdd/LiWWF13,DBLP:conf/sigmod/ChenCC13,DBLP:conf/icde/WangZZLW15}).
Most of them are geared towards boolean matching, thus making the size of messages received by users unpredictable.
This motivates us to study the problem of \topk \sk \pubsub such that only the \topk most relevant messages are presented to users.
Moreover, we adopt the popular sliding window model~\cite{DBLP:conf/pods/BabcockBDMW02} on \gt stream to provide the fresh information under controllable memory usage.
In particular,
for each subscription, we score a message based on their \gt similarity,
and the \topk messages are continuously maintained against the update of the sliding window (i.e., message arrival and expiration).
Below is a motivating example.

%


\begin{example}
Figure~\ref{fig:intro:example_real} shows an example of location-aware e-coupon recommendation system.
Three users interested in nearby restaurants are registered with their locations and favorite food,
intending to keep an eye on the most relevant e-coupon issued recently.
We assume the system only stores the most recent four e-coupons.
An e-coupon $e$ will be delivered to a user $u$ if $e$ has the highest score w.r.t. $u$ according to their spatial and textual similarity.
Initially, we have four e-coupons, and the top-1 answer of each user is shown in bold in the upper-right table,
where the relevance score between user and e-coupon is depicted.
When a new e-coupon $e_5$ arrives and the old e-coupon $e_1$ expires, the updated results are shown in bottom-right table.
Particularly, the top-1 answer of $u_1$ is replaced by $e_3$ since $e_1$ is discarded from the system,
while the answer of $u_3$ is replaced by $e_5$, as $e_5$ is the most relevant to $u_3$.
The top-1 answer of $u_2$ remains unchanged.
\end{example}

\noindent \textbf{Challenges.}
Besides the existing challenges in spatial-keyword query processing~\cite{DBLP:conf/icde/FelipeHR08,DBLP:journals/pvldb/CongJW09,DBLP:conf/ssd/RochaGJN11,DBLP:conf/sigir/ZhangCT14},
our problem presents two new challenges.

The first challenge is to devise an efficient indexing structure for a huge number of subscriptions, such that each message from the high-speed stream can be disseminated immediately on its arrival.
The only work that supports \topk \sk \pubsub is proposed by Chen~\emph{et al.}~\cite{DBLP:conf/icde/ChenCCT15}.
In a nutshell, they first deduce a textual bound for each subscription and then employ DAAT (Document-at-a-time~\cite{DBLP:conf/cikm/BroderCHSZ03}) paradigm to traverse the \ifile built in each spatial node.
However, we observe that the continuous \topk monitoring problem is essentially a \textit{threshold-based similarity search problem} from the perspective of message; that is, a new message will be delivered to a subscription if and only if its score is not less
than the current threshold score (e.g., $k$-th highest score) of the subscription.
Consequently, although the DAAT paradigm has been widely used for \topk search (e.g.,~\cite{DBLP:conf/sigir/DingS11}), it is not suitable to our problem
because the advanced threshold-based pruning techniques cannot be naturally integrated under DAAT paradigm.

%

The second challenge is the \topk re-evaluation problem triggered by frequent message expiration from the sliding window.
For example, in Figure~\ref{fig:intro:example_real}, the expiration of $e_1$ invalidates the current top-$1$ answer (i.e., $e_1$) of $u_1$,
and thus the system has to re-compute the new result for $u_1$ over the sliding window.
It is cost-prohibitive to re-evaluate all the affected subscriptions from scratch when a message expires.
Some techniques have been proposed to solve this problem (e.g.,~\cite{DBLP:conf/icde/YiYYXC03,DBLP:conf/sigmod/MouratidisBP06,DBLP:conf/icde/BohmOPY07,DBLP:journals/tods/PripuzicZA15}).
Yi~\emph{et al.}~\cite{DBLP:conf/icde/YiYYXC03} introduce a \kmax strategy, trying to maintain top-$k'$ results, with $k'$ being a value between $k$ and \kmax, rather than buffering the exact \topk results.
Later, Mouratidis~\emph{et al.}~\cite{DBLP:conf/sigmod/MouratidisBP06} notice that \kmax ignores the dominance relationship between messages, and propose a novel idea to convert \topk maintenance into \textit{partial \skyband} maintenance to reduce the number of re-evaluations.
Nevertheless, they simply use the $k$-th score of a continuous query (i.e., subscription in our paper) as the threshold of its \skyband without theoretical underpinnings, which may result in poor performance in practice.




On the other hand, the limited computational resources (e.g., CPU, memory) in a single machine often become the bottleneck when we increase the scale of real-life applications, where millions of active users need to be maintained simultaneously.
To alleviate this issue, we extend \ours on top of \storm\footnote{Apache storm project. http://storm.apache.org/}, an open-source distributed real-time in-memory processing system, to leverage parallel processing such that high throughput can be achieved.
\storm itself is intrinsically designed to solve real-time stream processing tasks, which therefore best suits our \topk \pubsub problem.
The main challenge here lies in how to partition and distribute subscriptions and messages such that workload balance and high throughput can be achieved at a small communication cost.


In this paper, we propose a novel centralized system, i.e.,~\ours,
to efficiently support \topk \textbf{S}patial-\textbf{k}e\textbf{y}word \textbf{P}ublish/Subscrib\textbf{e} over sliding window.
Two key modules, \textit{\moduleone} and \textit{\moduletwo}, are designed to address the above challenges.
Specifically, the \moduleone aims to rapidly deliver each arriving message to its affected subscriptions on its arrival.
We devise efficient subscription indexing techniques which carefully integrate both spatial and textual information.
Following the TAAT (Term-at-a-time~\cite{DBLP:conf/sigir/BuckleyL85}) paradigm, we significantly reduce the number of non-promising subscriptions
for the incoming message by utilizing a variety of spatial and textual pruning techniques.
On the other hand, the \moduletwo is designed to refill the \topk results of subscriptions when their results expire.
To alleviate frequent re-evaluations, we develop a novel \textit{cost-based \skyband} technique which carefully selects the messages to be buffered based on a threshold value determined by a cost model, considering both \topk re-evaluation cost and \skyband maintenance cost.
In addition, to speed-up real-time processing, we follow most of the existing \pubsub systems (e.g.,~\cite{DBLP:conf/kdd/LiWWF13,DBLP:conf/icde/WangZZLW15,DBLP:conf/icde/HuLLFT15,DBLP:conf/icde/ChenCCT15}) to implement all our indexes in main memory.

To support better scalability beyond \ours, we pioneer a novel distributed real-time processing system, namely,~\oursdis, which is a distributed version of \ours deployed on top of \storm.
We propose four different distribution mechanisms, i.e., hashing-based, location-based, keyword-based and prefix-based, to distribute subscriptions and messages to relevant components.
Among them, prefix-based technique yields the best overall performance in terms of both throughput and communication cost.
For example, it can process nearly $1300$ messages per second over $5M$ subscriptions on a small-size cluster.

\noindent \textbf{Contributions.} Our principal contributions are summarized as follows:
\begin{itemize}
\item We propose a novel framework, called \ours, which continuously maintains \topk \gt messages for a large number of subscriptions over sliding window model.
To the best of our knowledge, this is the first work to integrate sliding window model into \sk \pubsub system. (Section~\ref{sec:framework})
\item For \moduleone, we propose both \textit{individual pruning technique} and \textit{group pruning technique} to significantly improve the dissemination efficiency following the TAAT paradigm. (Section~\ref{sec:disseminate})
\item For \moduletwo,
a novel \textit{cost-based \skyband} method is developed to determine the best threshold value with in-depth theoretical analysis.
It is worth mentioning that our technique is a general approach which can be applied to other continuous \topk problems over sliding window. (Section~\ref{sec:refill})
\item We extend \ours on top of \storm, a distributed real-time processing environment.
By introducing to \storm a distribution layer which employs several efficient distribution mechanisms, the distributed version can achieve high throughput with better scalability.
As far as we know, this is the first work which extends \topk \pubsub system on top of \storm. (Section~\ref{sec:distributed})
\item We conduct extensive experiments to verify the efficiency and effectiveness of both \ours and its distributed version \oursdis.
It turns out that \ours usually achieves up to orders of magnitude improvement compared to its competitors, while \oursdis achieves further improvement over \ours with better scalability and large margin. (Section~\ref{sec:exp})
\end{itemize}


\section{Related Work}
\label{sec:related}

\subsection{Spatial-keyword Search}
Spatial-keyword search has been widely studied in literatures.
It aims to retrieve a set of \gt objects based on boolean matching
(e.g.,~\cite{DBLP:conf/cikm/ZhouXWGM05,DBLP:conf/ssdbm/HariharanHLM07,DBLP:conf/icde/FelipeHR08})
or score function (e.g.,~\cite{DBLP:journals/pvldb/CongJW09,DBLP:conf/ssd/RochaGJN11,DBLP:conf/cikm/ChristoforakiHDMS11,DBLP:conf/sigir/ZhangCT14}) by combining both spatial index (e.g.,~\rtree, \qd) and textual index (e.g.,~\ifile).
A nice summary of \sk query processing is available in~\cite{DBLP:journals/pvldb/ChenCJW13}.
Several extensions based on \sk processing have also been investigated, such as moving \sk query~\cite{DBLP:conf/sigmod/GuoZLTB15}, collective \sk query~\cite{DBLP:conf/sigmod/GuoCC15} and reverse \sk query~\cite{DBLP:conf/sigmod/LuLC11}.
Note that a spatial-keyword search is an ad-hoc/snapshot query (i.e., user-initiated model) while our problem focuses on continuous query (i.e., server-initiated model).

\subsection{Publish\slash Subscribe System}
Users register their interest as long-running queries in a \pubsub system, and streaming publications are delivered to relevant users whose interests are satisfied.
Nevertheless, most of the existing \pubsub systems (e.g.,~\cite{DBLP:journals/pvldb/WhangBSVVYG09,DBLP:conf/sigmod/SadoghiJ11,DBLP:journals/pvldb/ZhangCT14,DBLP:journals/pvldb/ShraerGFJ13}) do not consider spatial information.
Recently, \sk \pubsub system has been studied in a line of work (e.g.,~\cite{DBLP:conf/kdd/LiWWF13,DBLP:conf/sigmod/ChenCC13,DBLP:conf/icde/WangZZLW15,DBLP:conf/icde/HuLLFT15,DBLP:conf/icde/ChenCCT15}).
Among them, \cite{DBLP:conf/kdd/LiWWF13,DBLP:conf/sigmod/ChenCC13,DBLP:conf/icde/WangZZLW15} study the boolean matching problem
while \cite{DBLP:conf/icde/HuLLFT15} studies the similarity search problem, where each subscription has a \textit{pre-given} threshold.
These work are inherently different from ours, and it is non-trivial to extend their techniques to support \topk monitoring.

The \ciq index proposed by Chen~\emph{et al.}~\cite{DBLP:conf/icde/ChenCCT15} is the only close work that supports \topk \sk \pubsub
(shown in Figure~\ref{fig:related:example_ciq}).
In \ciq, a \qd is used to partition the whole space.
Each subscription is assigned to a number of covering cells, forming a disjoint partition of the entire space.
In Figure~\ref{fig:related:example_ciq}, we assume all the subscriptions have the same cell covering, i.e., from $c_1$ to $c_7$.
A textual bound (e.g., \textsf{MinT}) is precomputed for each subscription w.r.t. each assigned cell, as shown in the tables where the textual bounds w.r.t. $c_2$ and $c_7$ are displayed.
An \ifile ordered by subscription id is built to organize the subscriptions assigned to each cell.
For a new message (e.g., $m_1$),
\ciq traverses all the inverted files with corresponding cells penetrated by message location (e.g., $c_2$) in DAAT paradigm,
and finds all the subscriptions with textual similarity higher than the precomputed bound as candidates, which are then verified to get final results.
However, we notice that DAAT paradigm employed in \ciq cannot integrate some advanced techniques for threshold-based similarity search, given that the nature of our problem is a threshold-based search problem.
Contrary to \ciq, our indexing structure is designed for the TAAT paradigm, combined with advanced techniques for threshold-based pruning, thus enabling us to exclude a significant number of subscriptions.
Moreover, \ciq indexes each subscription into multiple cells, taking advantage of precomputed spatial bound.
However, the gain is limited since the number of covering cells for each subscription cannot be too large; otherwise, it would lead to extremely high memory cost.
Thus, we turn to an \textit{on-the-fly} spatial bound computation strategy, where each subscription is assigned to a single cell with finer spatial granularity.
Finally, we remark that \ciq integrates a time decay function rather than a sliding window, which, in the worst case, may overwhelm the limited memory.

\begin{figure}[t]
\centering
\includegraphics[width=0.95\columnwidth]{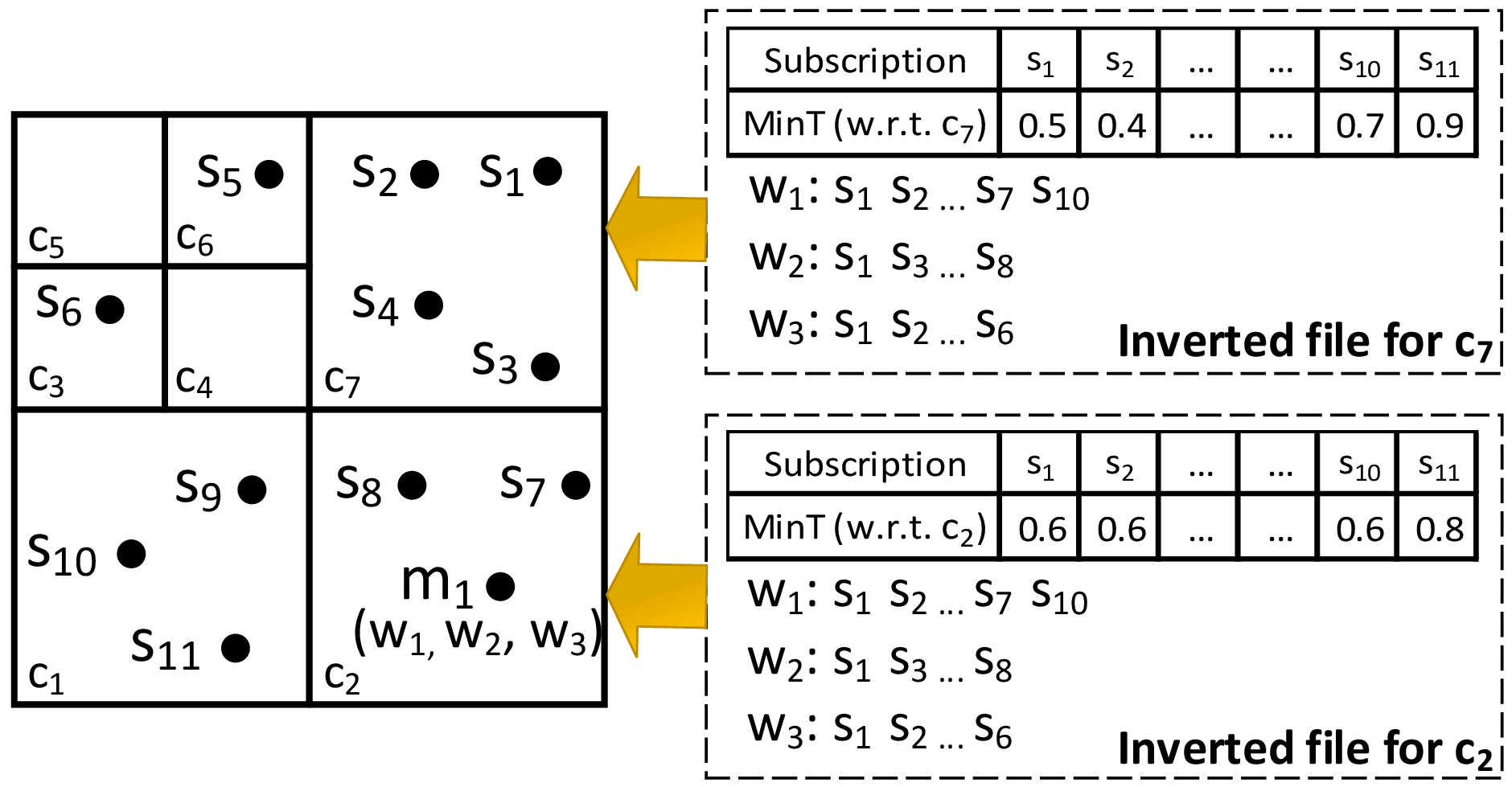}
\caption{\small{Example of \ciq index}}
\label{fig:related:example_ciq}
\end{figure}

\subsection{Top-k Maintenance Over Sliding Window}
One critical problem for \topk maintenance over sliding window is that,
when an old element (i.e., message in this paper) expires, we have to recompute the \topk results for the affected
continuous queries (i.e., subscriptions in this paper), which is cost-expensive if we simply re-evaluate from scratch.
On the flip side, it is also infeasible to buffer all elements and their scores for each individual query to avoid \topk re-evaluation.
Several techniques are proposed aiming to identify a trade-off between the number of re-evaluations and the buffer size.
In~\cite{DBLP:conf/icde/YiYYXC03}, Yi~\emph{et al.} introduce a \kmax approach.
Rather than maintain exact \topk results, they continuously maintain top-$k'$ results where $k'$ is between $k$ and a parameter \kmax.
However, followed by observation from Mouratidis~\emph{et al.}~\cite{DBLP:conf/sigmod/MouratidisBP06},
\kmax may contain redundant elements due to the overlook of dominance relationship.
Thus, Mouratidis~\emph{et al.} propose a \skyband based algorithm to remove redundancy.
Since it is very expensive to maintain the \textit{full \skyband} for each individual query,
they only keep elements with scores not lower than the $k$-th highest score determined by the most recent \topk re-evaluation.
We observe that this setting is rather ad-hoc and thus may result in unsatisfactory performance in practice.
B{\"{o}}hm~\emph{et al.}~\cite{DBLP:conf/icde/BohmOPY07} utilize a delay buffer to avoid inserting the newly-arriving objects with low scores into the \skyband.
However, since each object has to probe query index twice during its life time, their method is not suitable to our problem given the large number of registered queries (i.e., subscriptions).
Pripuzic~\emph{et al.}~\cite{DBLP:journals/tods/PripuzicZA15} propose a probabilistic \skyband method to drop the data which is unlikely to become \topk results in order to save space and improve efficiency.
However, their technique may discard some \topk elements due to its probabilistic nature.
In this paper, we propose a novel \textit{cost-based \skyband} technique to carefully determine the size of \skyband buffer based on a cost model.

\subsection{Distributed Spatial Query Processing}
There are a bunch of work studying spatial query processing by utilizing distributed system.
Nishimura~\emph{et al.}~\cite{DBLP:conf/mdm/NishimuraDAA11} extend HBase\footnote{Apache HBase project. https://hbase.apache.org/} to support multi-dimensional index.
Aji~\emph{et al.}~\cite{DBLP:journals/pvldb/AjiWVLL0S13} propose Hadoop-GIS, a distributed data warehouse infrastructure built on top of \hadoop, which provides functionality of spatial data analytics.
Later, Eldawy~\emph{et al.}~\cite{DBLP:conf/icde/EldawyM15} develop SpatailHadoop, a full-fledged system which supports various spatial queries by integrating spatial-awareness in each \hadoop layer.
Aly~\emph{et al.} present an adaptive mechanism on top of \hadoop to partition large-scale spatial data for efficient query processing~\cite{DBLP:journals/pvldb/AlyMHAOEQ15}.
Xie~\emph{et al.}~\cite{xiesimba} introduce a system called Simba to provide efficient in-memory spatial analytics by extending Spark SQL engine.
All the work above focus on some fundamental spatial queries, such as range query and kNN query, which is inherently different from our \topk \sk \pubsub problem.
A very relevant work, called \tornado, which also supports \sk stream processing on \storm, appears in~\cite{DBLP:journals/pvldb/MahmoodAQRDMAHA15}.
\tornado is a general spatial-keyword stream processing system to support both snapshot and continuous queries.
However, their main focus is not on the index construction over subscription queries, which nevertheless is the main contribution of our paper.
Besides, they cannot support the \topk \sk subscription queries as ours.

On the other hand, many stream processing systems, such as \sparkstreaming\footnote{Apache spark project. http://spark.apache.org/streaming/}, \samza\footnote{Apache samza project. http://samza.apache.org/} and \storm, have been developed to support efficient processing of real-time data.
Most of them are featured with open-source, low-latency, distributed, scalable and fault-tolerant characteristics.
A nice comparison between different stream processing systems can be found in~\cite{DBLP:journals/cloudcomp/Ranjan14}.
We choose \storm here mainly because of its simplicity, efficiency, well-documented APIs and very active community\footnote{https://github.com/apache/storm}.
To the best of our knowledge, our work is the first one to support \topk \sk \pubsub in a distributed environment.

\section{Preliminary}
\label{sec:preliminary}
In this section, we formally present some concepts which are used throughout this paper.

\begin{definition}[Geo-textual Message]
A \gt message is defined as $m=(\psi, \rho, t)$, where $m.\psi$ is a collection of keywords from a vocabulary $\mathcal{V}$, $m.\rho$ is a point location, and $m.t$ is the arrival time.
\end{definition}

\begin{definition}[Spatial-keyword Subscription]
A \sk subscription is denoted as $s=(\psi, \rho, k, \alpha)$, where $s.\psi$ is a set of keywords, $s.\rho$ is a point location, $s.k$ is the number of messages that $s$ is willing to receive and $s.\alpha$ is the preference parameter used in the score function.
\end{definition}

To buffer the most recent data from \gt stream, we adopt a count-based sliding window defined as follows.

\begin{definition}[Sliding Window]
Given a stream of geo-textual messages arriving in time order, the sliding window $\mathcal{W}$ over the stream with size $|\mathcal{W}|$ consists of most recent $|\mathcal{W}|$ \gt messages.
\end{definition}


In the following of the paper, we abbreviate \textit{\gt message} and \textit{\sk subscription} as \textit{message} (denoted as $m$) and \textit{subscription} (denoted as $s$) respectively if there is no ambiguity.
We assume that the keywords in vocabulary $\mathcal{V}$, as well as the keywords in subscription and message, are sorted in increasing order of their term frequencies.
Note that sorting keywords in increasing order of frequency is a widely-adopted heuristic to speed up similarity search~\cite{DBLP:conf/icde/ChaudhuriGK06,DBLP:conf/www/BayardoMS07,DBLP:journals/tods/XiaoWLYW11}.
The $i$-th keyword in $s$ is denoted as $s.\psi[i]$,
and we use $s.\psi[i:j]$ to denote a subset of $s.\psi$, i.e., $\cup_{i \leq k \leq j}{\{s.\psi[k]\}}$.
Particularly, $s.\psi[i:]$ denotes $\cup_{i \leq k \leq |s.\psi|}{\{s.\psi[k]\}}$.
Message $m$ follows the similar notations.

\noindent \textbf{Score function.}
To measure the \textit{relevance} between a subscription $s$ and a message $m$, we employ a score function defined as follows:

\begin{equation}
\begin{split}
\label{equ:score}
\score(s,m) & = s.\alpha \cdot \ssim(s.\rho, m.\rho) \\
& + (1 - s.\alpha) \cdot \tsim(s.\psi, m.\psi)
\end{split}
\end{equation}
where $\ssim(s.\rho,m.\rho)$ is the spatial proximity and $\tsim(s.\psi,m.\psi)$ is the textual relevance between $s$ and $m$.
Thus, a subscriber can receive messages which are not only close to her location but also fulfil her interest.
Meanwhile, the parameter $\alpha$ can be adjusted by subscribers to best satisfy their diverse preferences.

To compute spatial proximity, we utilize Euclidean distance as $\ssim(s.\rho, m.\rho) = 1 - \frac{\textsf{EDist}(s.\rho, m.\rho)}{\textsf{MaxDist}}$,
where $\textsf{EDist}(s.\rho,m.\rho)$ is the Euclidean distance between $s$ and $m$, and $\textsf{MaxDist}$ is maximum distance in the space.

For textual similarity, we employ the well-known \textit{cosine similarity} \cite{manning2008introduction} as $\tsim(s.\psi, m.\psi) = \sum_{w \in s.\psi \cap m.\psi}{\wt(s.w) \cdot \wt(m.w)}$,
where $\wt(s.w)$ and $\wt(m.w)$ are \textit{tf-idf} weights of keyword $w$ in $s$ and $m$ respectively.
Note that the weighting vectors of both $s$ and $m$ are normalized to unit length.
Also, same as~\cite{DBLP:conf/icde/ChenCCT15}, to guarantee the \topk results are \textit{textual-relevant}, a message must contain at least one common keyword with a subscription to become its \topk results.

\noindent \textbf{Problem statement.}
Given a massive number of \sk subscriptions and a \gt stream, we aim to continuously monitor \topk results for all the subscriptions against the stream over a sliding window $\mathcal{W}$ in real time.

\section{Framework}
\label{sec:framework}

\begin{figure}[t]
\centering
\includegraphics[width=0.9\columnwidth]{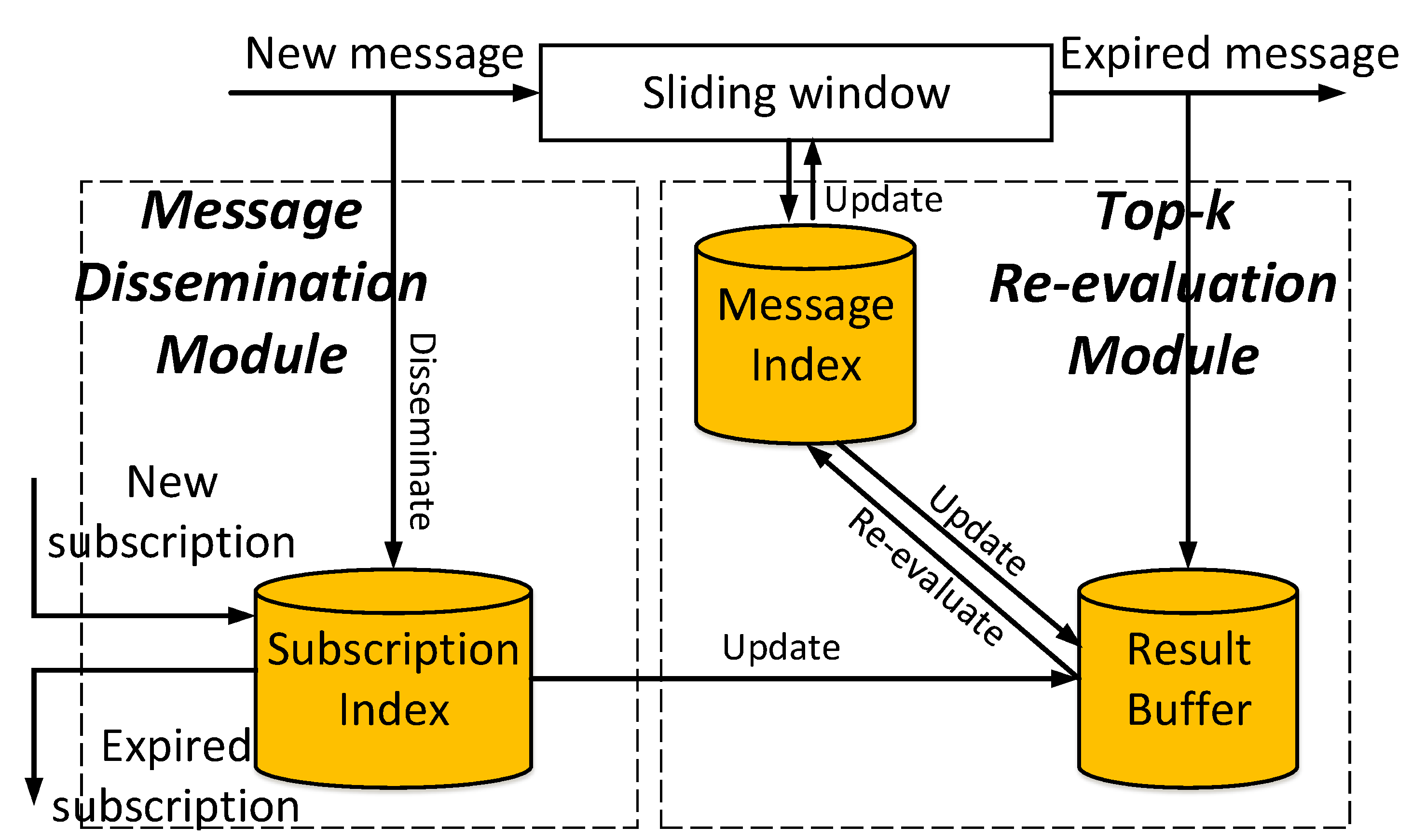}
\caption{\small{Framework of \ours}}
\label{fig:framework}
\end{figure}

Figure~\ref{fig:framework} shows the framework of \ours (Top-$k$ \textbf{S}patial-\textbf{k}e\textbf{y}word \textbf{P}ublish/Subscrib\textbf{e}).
We assume our system already has some registered subscriptions.
An arriving message will be processed by \textbf{\moduleone}, where a \textit{subscription index} is built to find all the affected subscriptions and update their \topk results.
An expired message will be processed by \textbf{\moduletwo}.
Specifically, it will check against a \textit{result buffer}, which maintains the \topk results (possibly including some non-\topk results) of all the subscriptions.
For the subscriptions that cannot be refilled through result buffer, their \topk results will be re-evaluated from scratch against a \textit{message index} containing all the messages over the sliding window.
Note that the message index can be implemented with any existing \sk index, such as \irtree~\cite{DBLP:journals/pvldb/CongJW09} and \sti~\cite{DBLP:conf/ssd/RochaGJN11}.
\ours can also support subscription update efficiently.
A new subscription will be inserted into subscription index, with its \topk results being initialized against message index, while an unregistered subscription will be deleted from both subscription index and result buffer.
Note that the subscription index and message index serve different purposes and cannot be trivially combined together.


\section{Message Dissemination}
\label{sec:disseminate}

In this section, we introduce a novel subscription index, which groups similar subscriptions, to support real-time dissemination against message stream.
Specifically, two key techniques, i.e., individual pruning and group pruning, are proposed in Section~\ref{sec:disseminate:ipt} and Section~\ref{sec:disseminate:gpt} respectively,
followed by the detailed indexing structure in Section~\ref{sec:disseminate:data_structure}.
Finally, we introduce dissemination algorithm in Section~\ref{sec:disseminate:algorithm} and index maintenance in Section~\ref{sec:disseminate:index_maintenance}.

\subsection{Individual Pruning Technique}
\label{sec:disseminate:ipt}

For each incoming message $m$, the key challenge is to determine all the subscriptions whose \topk results are affected.
Specifically, we denote the $k$-th highest score of a subscription $s$ as $\kscore(s)$.
Then the \topk results of $s$ need to be updated if $\kscore(s) \leq \score(s,m)$.
In this section, we propose a novel \textit{location-aware prefix filtering} technique to skip an individual subscription efficiently.

\subsubsection{Location-aware Prefix Filtering}
For ease of exposition, we denote a spatial similarity upper bound between a subscription $s$ and a message $m$ as $\ssimub(s.\rho,m.\rho)$, which will be discussed in detail in Section~\ref{sec:disseminate:ipt:spatial_bound}.
Based on Equation~\ref{equ:score}, we can derive a textual similarity threshold for pruning purpose accordingly:
\begin{equation}
\label{equ:tsimlb}
\tsimlb(s.\psi,m.\psi) = \frac{\kscore(s)}{1 - s.\alpha} - \frac{s.\alpha}{1 - s.\alpha} \cdot \ssimub(s.\rho,m.\rho)
\end{equation}

Then the following lemma claims that if the textual similarity between $s$ and $m$ is less than $\tsimlb(s.\psi,m.\psi)$, we can safely skip $s$.
\begin{lemma}
\label{lemma:tsimlb}
A message $m$ cannot affect \topk results of a subscription $s$ if $\tsim(s.\psi,m.\psi) < \tsimlb(s.\psi,m.\psi)$.
\end{lemma}
\begin{proof}
It is immediate from Equation~\ref{equ:score} and Equation~\ref{equ:tsimlb}.
\end{proof}


To utilize Lemma~\ref{lemma:tsimlb}, we employ prefix filtering technique, which is widely adopted in textual similarity join problems (e.g.,~\cite{DBLP:conf/icde/ChaudhuriGK06,DBLP:conf/www/BayardoMS07,DBLP:journals/tods/XiaoWLYW11}).
Prefix filtering is based on the fact that \tsim~is essentially a vector product; therefore, we can determine the similarity upper bound between two objects by only comparing their prefixes.
Before we introduce prefix filtering technique, we first introduce a threshold value for each keyword in $s$:
\begin{equation}
\wtsum(s.\psi[i]) = \sum_{ i \leq j \leq |s.\psi|}{\wt(s.\psi[j])}
\end{equation}
Then we define a \textit{location-aware prefix} as follows.
\begin{definition}[Location-aware Prefix]
\label{def:loc_aware_prefix}
Given a subscription $s$, a message $m$ and a textual similarity threshold $\tsimlb(s.\psi,m.\psi)$,
we use $\pref(s|m) = s.\psi[1:p]$ to denote the location-aware prefix of $s$ w.r.t. $m$,
where $p = \argmin_{i}{ \{ \wtsum(s.\psi[i+1]) < \tsimlb(s.\psi,m.\psi) \} }$.
\end{definition}

The following lemma claims that location-aware prefix is sufficient to decide whether a message can be \topk result of a subscription.
\begin{lemma}
\label{lemma:prune_s}
Given a subscription $s$ and a message $m$, $\pref(s|m) \cap m.\psi = \emptyset$ is sufficient to skip $s$ regarding $m$.
\end{lemma}
\begin{proof}
Since $\pref(s|m) \cap m.\psi = \emptyset$, $\tsim(s.\psi,m.\psi) \leq \sum_{ p+1 \leq i \leq |s.\psi| }{  \wt(s.\psi[i]) \cdot 1.0 } < \tsimlb(s.\psi,m.\psi)$, where $p$ is defined in Definition~\ref{def:loc_aware_prefix} and $1.0$ is the maximum weight for keyword in $m$.
Then the lemma holds immediately based on Lemma~\ref{lemma:tsimlb}.
\end{proof}
\begin{example}
Figure~\ref{fig:disseminate:example_spatial_prefix} shows an example of location-aware prefix, with 3 registered subscriptions and 3 incoming messages.
The underlined value to the right of each keyword corresponds to its weight, and we do not normalize the keyword weight for simplicity.
Assuming $\ssimub(s_1.\rho, m_1.\rho) = 0.98$, then $\tsimlb(s_1.\psi, m_1.\psi) = \frac{0.7}{1 - 0.6} - \frac{0.6}{1 - 0.6} \cdot 0.98 = 0.28$.
Thus, $\pref(s_1| m_1) = \{w_1, w_2, w_3\}$.
Since $m_1.\psi \cap \pref(s_1| m_1) = \emptyset$, we can skip $s_1$ w.r.t. $m_1$.
\end{example}

\begin{figure}[t]
	\centering
	\includegraphics[width=\columnwidth]{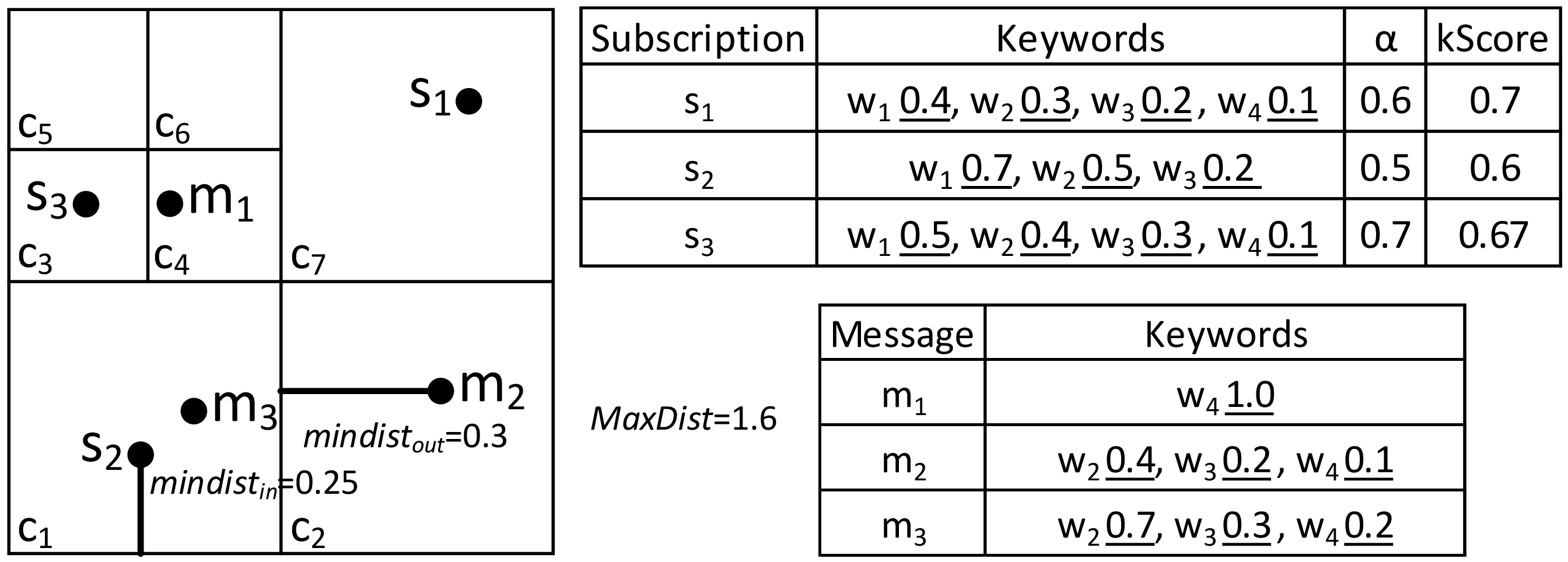}
	\caption{\small{Example of location-aware prefix}}
	\label{fig:disseminate:example_spatial_prefix}	
\end{figure}

It is noticed that different from conventional prefix technique (e.g.,~\cite{DBLP:conf/www/BayardoMS07,DBLP:journals/tods/XiaoWLYW11}) where only the prefix of a data entry needs to be indexed,
our location-aware prefix is dependent on the spatial location of messages,
and different locations may lead to different prefixes.
Thus, it is impossible to pre-compute and index the prefix of subscriptions.
To address this issue, we utilize the threshold value $\wtsum$ for each keyword in $s.\psi$ to indicate whether this keyword should occur in the prefix regarding a message $m$, which is stated formally in the following lemma.
\begin{lemma}
\label{lemma:single:prefix_filter}
Given a subscription $s$, a message $m$ and $\tsimlb(s.\psi,m.\psi)$,
if $m.\psi$ does not contain any keyword $s.w \in s.\psi$ that satisfies $\wtsum(s.w) \geq \tsimlb(s.\psi,m.\psi)$,
we can safely skip $s$ regarding m.
\end{lemma}
\begin{proof}
It is immediate from Definition~\ref{def:loc_aware_prefix} and Lemma~\ref{lemma:prune_s}.
\end{proof}

In this way, we can dynamically determine the location-aware prefix of a subscription w.r.t. an arriving message.
Also, since $\wtsum(s.w)$ is irrelevant to incoming messages, it can be materialized for each subscription.
\begin{example}
\label{example:spatial_prefix}
Following the same example in Figure~\ref{fig:disseminate:example_spatial_prefix},
since $\wtsum(s_1.w_4) = 0.1 < \tsimlb(s_1.\psi, m_1.\psi) = 0.28$, and $w_4$ is the only keyword in $m_1$, Lemma~\ref{lemma:single:prefix_filter} holds and we can skip $s_1$ w.r.t. $m_1$.
\end{example}

\noindent \textbf{Max-weight refinement.}
We notice that for a specific message $m$, we can compute a better location-aware prefix for $s$ by considering the maximum weight for the keywords in $m$.
We first define $\maxwt(m.\psi[i])$ as:
\begin{equation}
\maxwt(m.\psi[i]) = \max_{i \leq j \leq |m.\psi|}{\wt(m.\psi[j])}
\end{equation}

Then we define a \textit{refined location-aware prefix}:
\begin{definition}[Refined Location-aware Prefix]
\label{def:loc_aware_prefix_refined}
Given a subscription $s$, a message $m$ and $\tsimlb(s.\psi,m.\psi)$,
we use $\prefplus(s|m) = s.\psi[1:p]$ to denote the refined location-aware prefix of $s$ w.r.t. $m$,
where $p = \argmin_{i} \{  \maxwt(m.\psi[j]) \times \wtsum(s.\psi[i+1]) < \tsimlb(s.\psi,m.\psi) \} $ with $m.\psi[j] = s.\psi[i+1]$.
\end{definition}

The following theorem claims that $\prefplus(s|m)$ is sufficient to decide whether a message can be \topk result of a subscription.
\begin{theorem}
\label{theorem:single:prefix_filter_refined}
Given a subscription $s$ and a message $m$,
$\prefplus(s|m) \cap m.\psi = \emptyset$ is sufficient to skip $s$ regarding $m$.
\end{theorem}
\begin{proof}
Since $\prefplus(s|m) \cap m.\psi = \emptyset$, $\tsim(s.\psi,m.\psi) \leq 0 + \wtsum(s.\psi[p+1]) \cdot \maxwt(m.\psi[j]) < \tsimlb(s.\psi,m.\psi)$ holds, where $p$ is defined in Definition~\ref{def:loc_aware_prefix_refined} and $m.\psi[j] = s.\psi[p+1]$.
Then the theorem holds immediately based on Lemma~\ref{lemma:tsimlb}.
\end{proof}
\begin{example}
Assuming $\ssimub(s_1.\rho,m_2.\rho) = 0.99$ in Figure~\ref{fig:disseminate:example_spatial_prefix}, then $\tsimlb(s_1.\psi,m_2.\psi) = 0.26$.
Based on Lemma~\ref{lemma:single:prefix_filter}, $\pref(s_1|m_2) = \{w_1,w_2,w_3\}$, and thus $s_1$ cannot be skipped w.r.t. $m_2$.
However, if we consider \maxwt, then $\prefplus(s_1|m_2) = \{w_1\}$.
Thus, $s_1$ can be skipped w.r.t. $m_2$.
\end{example}

\subsubsection{Spatial Bound Estimation}
\label{sec:disseminate:ipt:spatial_bound}
In this section, we discuss the computation of $\ssimub(s.\rho,m.\rho)$ between a subscription $s$ and a message $m$ in order to get a better threshold $\tsimlb(s.\psi,m.\psi)$ for efficient location-aware prefix filtering.
To this end, we employ a spatial index to group subscriptions with similar locations, such that the spatial upper bound for a group of subscriptions can be computed simultaneously.
Due to the easy implementation and well-adaptiveness to skewed spatial distributions, we choose \qd to index subscriptions.
Specifically, each subscription $s$ is assigned into a leaf cell $c$ with range $c.r$ based on its location $s.\rho$.
Then the following two types of spatial bounds can be defined and utilized.

\begin{definition}[Inner Spatial Bound]
\label{def:inner_spatial_bound}
Given a subscription $s$ and its residing cell $c$,
inner spatial bound, denoted as $\ssimub_{in}(s.\rho,c.r)$ is computed as $1.0 - \frac{mindist(s.\rho,c.r)}{MaxDist}$, where $mindist(s.\rho,c.r)$ is the $mindist$ from $s$ to the nearest boundary of cell range $c.r$.
\end{definition}

It is obvious that for any $s \in c$ and $m \notin c$, we have $\ssimub_{in}(s.\rho,c.r) \geq \ssim(s.\rho,m.\rho)$.
An example is shown in Figure~\ref{fig:disseminate:example_spatial_prefix}.
Since the $mindist$ from $s_2$ to $c_1$ is $0.25$, $\ssimub_{in}(s_2.\rho,c_1.r) = 1 - \frac{0.25}{1.6} = 0.84$ if we assume the $\textsf{MaxDist}$ in the space is $1.6$.

\begin{definition}[Outer Spatial Bound]
\label{def:outer_spatial_bound}
Given a message $m$ and an outer cell $c$,
outer spatial bound, denoted as $\ssimub_{out}(m.\rho,c.r)$ is computed as $1.0 - \frac{mindist(m.\rho,c.r)}{MaxDist}$, where $mindist(m.\rho,c.r)$ is the $mindist$ from $m$ to $c$.
\end{definition}

For any $s \in c$ and $m \notin c$, we have $\ssimub_{out}(m.\rho,c.r) \geq \ssim(s.\rho,m.\rho)$.
An example is also shown in Figure~\ref{fig:disseminate:example_spatial_prefix}.
The $mindist$ from $m_2$ to $c_1$ is $0.3$, and thus $\ssimub_{out}(m_2.\rho,c_1.r) = 1 - \frac{0.3}{1.6} = 0.81$.

\begin{definition}[Spatial Upper Bound]
\label{def:spatial_upper_bound}
Given a subscription $s$ inside a cell $c$ and a message $m$ outside $c$,
the spatial upper bound, denoted as $\ssimub(s.\rho,m.\rho)$ is computed as $1.0 - \frac{ mindist(s.\rho,c.r) + mindist(m.\rho,c.r)}{MaxDist}$, where $mindist(s.\rho,c.r)$ and $mindist(m.\rho,c.r)$ are the same as Definition~\ref{def:inner_spatial_bound} and~\ref{def:outer_spatial_bound}.
\end{definition}

Following the example in Figure~\ref{fig:disseminate:example_spatial_prefix}, by combining both inner and outer distance, we can get a tighter spatial upper bound between $s_2$ and $m_2$ as $\ssimub(s_2.\rho,m_2.\rho) = 1 - \frac{0.25 + 0.3}{1.6} = 0.65$.

Note that the inner spatial bound can be precomputed and materialized, while the outer spatial bound has to be computed on-the-fly as it is relevant to the location of an arriving message.
However, the computation cost of $\ssimub_{out}(m.\rho,c.r)$ is not expensive since we only need to compute this value against each leaf cell.
Finally, we remark that when $s$ and $m$ are within the same cell, both $\ssimub_{in}(s.\rho,c.r)$ and $\ssimub_{out}(m.\rho,c.r)$ are always $1.0$.
\begin{example}
An example is shown in Figure~\ref{fig:disseminate:example_spatial_prefix}.
If we assume the $\ssimub(s_2.\rho,m_2.\rho) = 1.0$, we have $\tsimlb(s_2,m_2) = \frac{0.6}{1-0.5}-\frac{0.5}{1-0.5} \cdot 1.0 = 0.20$, and $\prefplus(s_2|m_2) = \{w_1, w_2\}$.
Thus, $s_2$ cannot be skipped w.r.t. $m_2$.
However, if we utilize the inner spatial bound and outer spatial bound together, we have $\ssimub(s_2.\rho,m_2.\rho) = 1 - \frac{0.25+0.3}{1.6} = 0.65$, $\tsimlb(s_2,m_2) = 0.55$, and $\prefplus(s_2|m_2) = \{w_1\}$.
In this case, we can safely skip $s_2$ w.r.t. $m_2$.
\end{example}

\subsubsection{Bound Estimation for Unseen Keywords}
Since we employ TAAT paradigm to visit \ifile, we can estimate a textual upper bound for unseen keywords.
If this upper bound plus the textual similarity that has already been computed is still less than the required threshold, we can safely skip $s$.
The textual upper bound between the unseen keywords of $s$ and $m$ can be computed as follows:
\begin{align}
& \tsimub(s.\psi[i:],m.\psi[j:]) =  min \Big\{ \wtsum(s.\psi[i]) \nonumber \\
& \times \maxwt(m.\psi[j]), \wtsum(m.\psi[j]) \times \maxwt(s.\psi[i]) \Big\}
\end{align}
where $i$ and $j$ are starting positions of unseen keywords.
Then the following theorem claims we can skip a subscription by utilizing the textual upper bound.
\begin{theorem}
\label{theorem:single:positional_filter}
Given a subscription $s$, a message $m$ and their textual similarity threshold $\tsimlb(s.\psi,m.\psi)$,
assuming we have already computed the partial similarity between $s.\psi[1:i]$ and $m.\psi[1:j]$, denoted as $\tsim(s.\psi[1:i],m.\psi[1:j])$,
then $\tsim(s.\psi[1:i],m.\psi[1:j]) + \tsimub(s.\psi[i+1:],m.\psi[j+1:]) < \tsimlb(s.\psi,m.\psi)$ is sufficient to skip $s$.
\end{theorem}
\begin{proof}
As $\tsim(s.\psi,m.\psi) \leq \tsim(s.\psi[1:i],m.\psi[1:j]) + \tsimub(s.\psi[i+1:],m.\psi[j+1:])$, we have $\tsim(s.\psi,m.\psi) < \tsimlb(s.\psi,m.\psi)$.
The theorem holds immediately from Lemma~\ref{lemma:tsimlb}.
\end{proof}
\begin{example}
In Figure~\ref{fig:disseminate:example_spatial_prefix},
consider that we are currently disseminating $m_3$.
Based on the dissemination algorithm to be discussed later in Section~\ref{sec:disseminate:algorithm},
we need to traverse the inverted lists in cell $c_3$ (where $s_3$ resides) for all the keywords in $m_3.\psi$ one by one.
We first check the inverted list of $w_2$ since $w_2$ is the 1st keyword of $m_3$.
Assuming $\tsimlb(s_3.\psi,m_3.\psi) = 0.48$, we cannot skip $s_3$ since $w_2 \in \prefplus(s_3|m_3) = \{w_1, w_2\}$.
However, since $w_2$ is the 2nd keyword in $s_3.\psi$, we can compute $\tsim(s_3.\psi[1:2],m_3.\psi[1:1]) = 0.4 \cdot 0.7 = 0.28$,
and $\tsimub(s_3.\psi[3:],m_3.\psi[2:]) = \min(0.4 \cdot 0.3, 0.5 \cdot 0.3) = 0.12$.
Because $0.28 + 0.12 < \tsimlb(s_3.\psi,m_3.\psi) = 0.48$, we can immediately skip $s_3$.
\end{example}

\subsection{Group Pruning Technique}
\label{sec:disseminate:gpt}

After applying individual pruning technique, many subscriptions can be skipped without the need to compute their exact similarity w.r.t. a message.
To further enhance the performance, we propose a novel \textit{Group Pruning Technique} such that we can skip a group of subscriptions without the need to visit them individually.
To begin with, we first define \textit{subscription-dependent prefix} for a message.
\begin{definition}[Refined Sub-dependent Prefix]
\label{def:refined_sub_aware_prefix}
Given a message $m$, a subscription $s$ and $\tsimlb(s.\psi,m.\psi)$,
we use $\prefplus(m|s) = m.\psi[1:p]$ to denote the refined subscription-dependent prefix of $m$ w.r.t. $s$,
where $p = \argmin_{i} \{ \maxwt(s.\psi[j]) \times \wtsum(m.\psi[i+1]) $
$< \tsimlb(s.\psi,m.\psi) \} $ with $s.\psi[j] = m.\psi[i+1]$.
\end{definition}
The following lemma claims the refined subscription-dependent prefix is sufficient to determine whether a message could be \topk result of a subscription.
\begin{lemma}
\label{lemma:group:prune_s}
Given a subscription $s$ and a message $m$, $\prefplus(m|s) \cap s.\psi = \emptyset$ is sufficient to skip $s$ regarding $m$.
\end{lemma}
\begin{proof}
Since $\prefplus(m|s) \cap s.\psi = \emptyset$, $\tsim(s.\psi,m.\psi) \leq 0 + \wtsum(m.\psi[p+1]) \cdot \maxwt(s.\psi[j]) < \tsimlb(s.\psi,m.\psi)$, where $p$ is defined in Definition~\ref{def:refined_sub_aware_prefix} and $s.\psi[j] = m.\psi[p+1]$.
Then the lemma holds immediately based on Lemma~\ref{lemma:tsimlb}.
\end{proof}

Let us denote the posting list of keyword $w$ in cell $c$ as $plist(c,w)$, which contains all the subscriptions having $w$ and residing in $c$.
Then based on Lemma~\ref{lemma:group:prune_s},
for a subscription $s$ in $plist(c,w)$,
if $s.w \notin \prefplus(m|s)$, we can safely skip $s$.
Further, if this holds for a group of subscriptions on $plist(c,w)$, we can safely skip the whole group as follows.
\begin{lemma}
\label{lemma:group:skip}
Given a message $m$, a keyword $w \in m.\psi$, a posting list $plist(c,w)$ and a group of subscriptions $\mathcal{G}$ inside $plist(c,w)$,
$\max_{s \in \mathcal{G}}{ \left\{ \maxwt(s.w) \right\} \cdot \wtsum(m.w)  } < \min_{s \in \mathcal{G}}{ \left\{ \tsimlb(s.\psi,m.\psi) \right\} }$ is sufficient to skip the whole group $\mathcal{G}$.
\end{lemma}
\begin{proof}
For each $s \in \mathcal{G}$, $\maxwt(s.w) \cdot \wtsum(m.w) < \tsimlb(s.\psi,m.\psi)$ holds which indicates $s.w \notin \prefplus(m|s)$ according to Definition~\ref{def:refined_sub_aware_prefix}.
Thus, $s$ can be skipped based on Lemma~\ref{lemma:group:prune_s}, and therefore $\mathcal{G}$ can be skipped immediately.
\end{proof}

The left side of the inequality in Lemma~\ref{lemma:group:skip} can be computed in $O(1)$ time since we can materialize $\max_{s \in \mathcal{G}}{ \left\{ \maxwt(s.w) \right\} }$ for each group.
However, for the right side, it would be quite inefficient if we compute it on the fly for each new message.
To avoid this, we propose a lower bound for $ \min_{s \in \mathcal{G}}{ \left\{ \tsimlb(s.\psi,m.\psi) \right\} }$ which can be computed in constant time.
In the following, we first present the subscription grouping strategy and then introduce the details of the lower bound deduction.

\subsubsection{$\alpha$-Partition Scheme}
\label{sec:disseminate:gpt:alpha_partition}
Intuitively, we should group subscriptions with similar $\tsimlb(s.\psi,m.\psi)$ such that we can get a tighter textual threshold for the group.
We first let $\ssimub(s.\rho,m.\rho) = \ssimub_{out}(m.\rho, c.r)$.
Note that we compute $\ssimub(s.\rho,m.\rho)$ by utilizing $\ssimub_{out}(m.\rho, c.r)$ only in order to make $\ssimub(s.\rho,m.\rho)$ independent of a specific subscription.
Therefore, it is observed from Equation~\ref{equ:tsimlb} that, for the computation of $\tsimlb(s.\psi,m.\psi)$, only $\frac{\kscore(s)}{1-s.\alpha}$ and $\frac{s.\alpha}{1-s.\alpha}$ are dependent on $s$ while $\ssimub(s.\rho,m.\rho)$ is irrelevant to $s$.
For simplicity, we denote $\frac{\kscore(s)}{1-s.\alpha}$ as $\kscore^{*}(s)$ and $\frac{s.\alpha}{1-s.\alpha}$ as $s.\alpha^{*}$ respectively.
Then, we partition subscriptions into groups based on their $\alpha^{*}$ values, such that the subscriptions inside a group have similar $\alpha^{*}$ values.
We employ a quantile-based method to partition the domain of $\alpha^{*}$ to ensure that each group has similar number of subscriptions.
Then, we can skip the whole group $\mathcal{G}$ as stated in the following theorem.
\begin{theorem}
\label{theorem:group:lower_bound}
Given a group $\mathcal{G}$ generated by $\alpha$-partition in a posting list $plist(c,w)$, we denote $\min_{s \in \mathcal{G}}{ \left\{ \kscore^{*}(s) \right\} }$ as $\kscore^{*}(\mathcal{G})$ and $\max_{s \in \mathcal{G}}{ \left\{ s.\alpha^{*} \right\} }$ as $\mathcal{G}.\alpha^{*}$.
then $\max_{s \in \mathcal{G}}{ \left\{ \maxwt(s.w) \right\} \cdot \wtsum(m.w) } <  \kscore^{*}(\mathcal{G}) - \mathcal{G}.\alpha^{*} \cdot \ssimub_{out}(m.\rho, c.r)$ is sufficient to skip the whole group $\mathcal{G}$.
\end{theorem}
\begin{proof}
It is obvious that for any subscription $s$ in $\mathcal{G}$, we have
$ \max_{s' \in \mathcal{G}}{ \left\{ \maxwt(s'.w) \right\} \cdot \wtsum(m.w) } < \kscore^{*}(\mathcal{G}) - \mathcal{G}.\alpha^{*} \cdot \ssimub(c.r,m.\rho)  \leq
\kscore^{*}(s) - s.\alpha^{*} \cdot \ssimub(c.r,m.\rho)
= \tsimlb(s.\psi,m.\psi)$.
Thus, we have $\max_{s' \in \mathcal{G}}{ \left\{ \maxwt(s'.w) \right\} \cdot \wtsum(m.w) } < \min_{s' \in \mathcal{G}}{ \left\{ \tsimlb(s'.\psi,m.\psi) \right\} }$.
Combined with Lemma~\ref{lemma:group:skip}, the theorem holds immediately.
\end{proof}

\noindent \textbf{Time complexity.} The condition checking in Theorem~\ref{theorem:group:lower_bound} takes $O(1)$ time, since we can precompute the values of $\kscore^{*}(\mathcal{G})$ and $\mathcal{G}.\alpha^{*}$.

\vspace{-2mm}
\subsubsection{Early Termination Within Group}
When a group $\mathcal{G}$ cannot be skipped given a message, we have to check each subscription in it.
To avoid this, we propose an \textit{early termination technique} to early stop within a group when the group cannot be skipped totally.
To enable early termination,
for each group $\mathcal{G}$ in $plist(c,w)$,
we sort the subscriptions in $\mathcal{G}$ by their $\kscore^{*}$ values increasingly.
For each subscription $s$ in $\mathcal{G}$, we denote the subscriptions with $\kscore^{*}$ not less than $\kscore^{*}(s)$ as $\mathcal{G}[s] = \left\{ s' \in \mathcal{G} | \kscore^{*}(s') \geq \kscore^{*}(s) \right\}$,
and maintain two statistics $\maxwt(\mathcal{G}[s])$ and $\mathcal{G}[s].\alpha^{*}$ w.r.t. keyword $w$ as follows:
\begin{equation}
\maxwt(\mathcal{G}[s]) = \max_{s' \in \mathcal{G}[s] }{ \left\{ \maxwt(s'.w) \right\} }
\end{equation}
\vspace{-4mm}
\begin{equation}
\mathcal{G}[s].\alpha^{*} = \max_{s' \in \mathcal{G}[s] }{ \left\{ s'.\alpha^{*} \right\} }
\end{equation}
Then we can employ early termination as follows.
\begin{theorem}
\label{theorem:group:early_terminate}
Given a group $\mathcal{G}$ inside a posting list $plist(c,w)$,
and assuming $\hat{s}$ is the subscription with smallest position in $\mathcal{G}$ such that the following inequality holds:
$\maxwt(\mathcal{G}[\hat{s}]) \cdot \wtsum(m.w) <
\kscore^{*}(\hat{s}) - \mathcal{G}[\hat{s}].\alpha^{*} \cdot \ssimub_{out}(m.\rho,c.r)$,
then there is no need to check the subscriptions after $\hat{s}$ (including $\hat{s}$ itself).
\end{theorem}
\begin{proof}
For any subscription $s'$ after $\hat{s}$, the following inequalities hold:
$\kscore^{*}(\hat{s}) \leq \kscore^{*}(s')$,
$\maxwt(\mathcal{G}[\hat{s}]) \geq \maxwt(\mathcal{G}[s'])$,
$\mathcal{G}[\hat{s}].\alpha^{*} \geq \mathcal{G}[s'].\alpha^{*}$.
Thus, $ \maxwt(s'.w) \cdot \wtsum(m.w)  \leq
\maxwt(\mathcal{G}[s']) \cdot \wtsum(m.w)  \leq
\maxwt(\mathcal{G}[\hat{s}]) \cdot \wtsum(m.w)  <
\kscore^{*}(\hat{s}) - \mathcal{G}[\hat{s}].\alpha^{*} \cdot \ssimub_{out}(m.\rho,c.r) \leq
\kscore^{*}(s') - \mathcal{G}[s'].\alpha^{*} \cdot \ssimub_{out}(m.\rho,c.r) \leq
\kscore^{*}(s') - s'.\alpha^{*} \cdot \ssimub_{out}(m.\rho,c.r) = \tsimlb(s'.\psi,m.\psi)$.
Based on Definition~\ref{def:refined_sub_aware_prefix},
we know that $s'.w \notin \prefplus(m|s')$.
Thus $s'$ can be skipped based on Lemma~\ref{lemma:group:prune_s}.
Thus, the theorem holds immediately.
\end{proof}
\noindent \textbf{Time complexity.} To speed-up the real-time processing, we precompute $\maxwt(\mathcal{G}[s])$ and $\mathcal{G}[s].\alpha^{*}$  and store them with each subscription in the group $\mathcal{G}$.
The condition checking in Theorem~\ref{theorem:group:early_terminate} can be efficiently computed in $O(log |\mathcal{G}|)$ time with a binary search method.

\begin{figure*}[t]
	\centering
	\includegraphics[width=\linewidth]{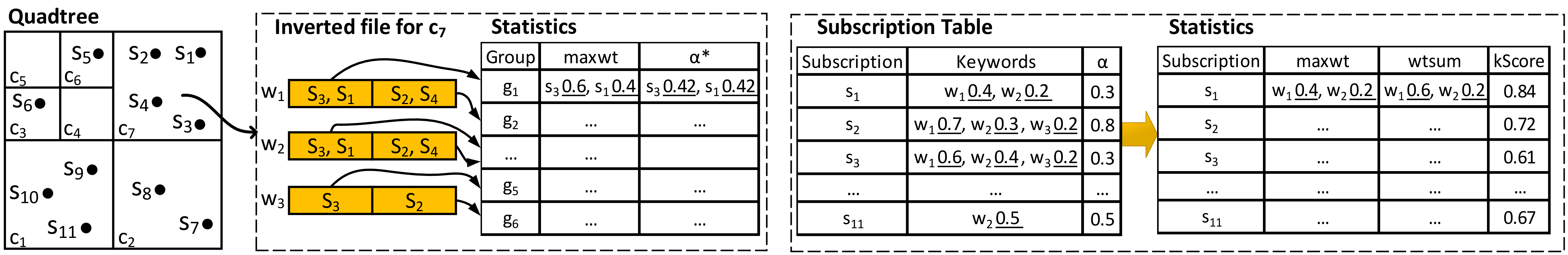}
	\caption{\small{Subscription index}}
	\label{fig:disseminate:our_data_structure}	
\end{figure*}
\vspace{-6mm}

\subsubsection{Cell-based Pruning}
Besides the above group pruning technique, we notice that for some cells which are far away from the location of an arriving message, we can safely skip the whole cell.
Specifically, for each subscription $s$ within a cell $c$, we can derive a spatial similarity threshold as follows:
\begin{equation}
\ssimlb(s.\rho) = \frac{\kscore(s)}{s.\alpha} - \frac{1 - s.\alpha}{s.\alpha}
\end{equation}
where we assume the textual similarity achieves the largest value, i.e., $1$.
Then we can reach the following lemma.
\begin{lemma}
\label{lemma:other:cell_prune}
Given a cell $c$, if $ \min_{s \in c}{ \ssimlb(s.\rho) } > \ssimub_{out}(m.\rho,c.r)$,
we can safely skip all the subscriptions in cell $c$.
\end{lemma}
\begin{proof}
For $\forall s \in c$, we have $\ssimlb(s.\rho) = \frac{\kscore(s)}{s.\alpha} - \frac{1 - s.\alpha}{s.\alpha} > \ssimub_{out}(m.\rho,c.r)$. Thus, $\kscore(s) > (1 - s.\alpha) + s.\alpha \cdot \ssimub_{out}(m.\rho,c.r) >= \score(s,m)$.
Thus, $m$ cannot be \topk results of any $s$ in $c$.
\end{proof}
\vspace{-6mm}

\subsection{Subscription Index}
\label{sec:disseminate:data_structure}

Relying on all the techniques discussed above, our subscription index is essentially a \qd structure integrated with \ifile in each leaf cell, as shown in Figure~\ref{fig:disseminate:our_data_structure}.
For each registered subscription, we store its detailed information and relevant statistics in a subscription table, and insert it into a leaf cell of \qd based on its spatial location.
Note that in \qd, we only store the subscription id referring to its detailed information in subscription table.
Within each leaf cell, an \ifile is built upon all the subscriptions inside the cell.
Then each posting list in \ifile is further partitioned into groups based on the subscription preference $\alpha^{*}$ to enable group pruning.
Each group is also associated with some statistics mentioned above.
Finally, to facilitate early termination, the subscriptions within each group are ordered based on their $\kscore^{*}$.

\subsection{Dissemination Algorithm}
\label{sec:disseminate:algorithm}
Algorithm~\ref{alg:disseminate} shows our message dissemination algorithm.
We follow a \textit{filtering-and-verification} paradigm, where we first generate a set of candidate subscriptions (Lines~\ref{alg:disseminate:emptymap}-\ref{alg:disseminate:set_infty}),
and then compute the exact scores to determine the truly affected ones, with the updated results being disseminated accordingly (Line~\ref{alg:disseminate:verify}).
Specifically, we first initialize an empty map $\mathcal{R}$ to store candidates with their scores (Line~\ref{alg:disseminate:emptymap}).
Then the $\maxwt$ and $\wtsum$ values for all the keywords in the arriving message $m$ are computed for later use (Line~\ref{alg:disseminate:comp_maxwt}).
For each leaf cell $c$ surviving from cell pruning (Line~\ref{alg:disseminate:cell_pruning}),
we first compute $\ssimub_{out}(m.\rho,c.r)$ and then traverse the \ifile in cell $c$ following a TAAT manner.
For each group $\mathcal{G}$ encountered in $plist(c,w)$ (Line~\ref{alg:disseminate:each_group}),
we skip $\mathcal{G}$ if group pruning can be applied (Line~\ref{alg:disseminate:group_pruning});
otherwise, we identify $\hat{s}$ for early termination based on Theorem~\ref{theorem:group:early_terminate} (Line~\ref{alg:disseminate:identify_hat_s} and Line~\ref{alg:disseminate:early_termination}).
For each surviving subscription $s$, we employ location-aware prefix filtering (Line~\ref{alg:disseminate:prefix_filter}) and bound estimation for unseen keywords (Line~\ref{alg:disseminate:positional_filter}) to skip it as early as possible.
For the surviving subscriptions, we store the accumulated textual similarity so far w.r.t. $m$ in $\mathcal{R}$, while for the skipped subscriptions, we set $\mathcal{R}[s]$ to negative infinity (Line~\ref{alg:disseminate:set_infty}).
Finally, for each subscription in $\mathcal{R}$ with $\mathcal{R}[s] > 0$, we verify it and update its \topk results if needed (Line~\ref{alg:disseminate:verify}).
Note that when verifying a candidate $s$, we only need to compute the exact spatial similarity to get the final score because the textual similarity, i.e., $\mathcal{R}[s]$, has already been computed.
The statistics relevant to pruning techniques are also updated in Line~\ref{alg:disseminate:verify}.

\begin{algorithm}[tb]
\SetVline 
\SetFuncSty{textsf}
\SetArgSty{textsf}
\small
\caption{\textsf{MessageDissemination}($m$)}
\label{alg:disseminate}
\Input
{
    $m:$ a new incoming message \\
}
\State{$\mathcal{R} := \emptyset$ \tcc*[f]{A candidate map}}\label{alg:disseminate:emptymap}
\For{$1 \leq i \leq |m.\psi|$}
{
	\State{Compute $\maxwt(m.\psi[i])$ and $\wtsum(m.\psi[i])$} \label{alg:disseminate:comp_maxwt}
}\label{alg:disseminate:comp_maxwt_stop}
\For{each leaf cell $c$ in the \qd}
{
	\If(\tcc*[f]{Cell pruning}){ $c$ is skipped by Lemma~\ref{lemma:other:cell_prune}}
	{
		\State{\textbf{Continue}}
	}\label{alg:disseminate:cell_pruning}
	\State{Compute outer spatial bound $\ssimub_{out}(m.\rho,c.r)$}
	\For{ $1 \leq i \leq |m.\psi|$ }
	{
		\State{$w := m.\psi[i]$ }
		\For{ each group $\mathcal{G}$ in $plist(c,w)$}
		{
			\If(\tcc*[f]{Group pruning based on Theorem~\ref{theorem:group:lower_bound}} ){ $\max_{s \in \mathcal{G}}{ \left\{ \maxwt(s.w) \right\} \cdot \wtsum(m.w) } <  \kscore^{*}(\mathcal{G}) - \mathcal{G}.\alpha^{*} \cdot \ssimub_{out}(m.\rho,c.r)$ }
			{
				\State{\textbf{Continue}}
			}\label{alg:disseminate:group_pruning}
			\State{Identify $\hat{s}$ based on Theorem~\ref{theorem:group:early_terminate} } \label{alg:disseminate:identify_hat_s}
			\For{ each subscription $s$ in group $\mathcal{G}$}
			{
				\If(\tcc*[f]{Early termination based on Theorem~\ref{theorem:group:early_terminate}} ){ $s == \hat{s}$}
				{
					\State{\textbf{Break}}
				}\label{alg:disseminate:early_termination}
				\If{ $ s \in \mathcal{R}~\&\&~\mathcal{R}[s] == -\infty$}
				{
					\State{\textbf{Continue}}
				}
				\State{Compute $\tsimlb(s.\psi,m.\psi)$ based on both inner and outer spatial bounds}
				\If (\tcc*[f]{Theorem~\ref{theorem:single:prefix_filter_refined}} ){$\maxwt(m.w) \cdot \wtsum(s.w) < \tsimlb(s.\psi,m.\psi)$}
				{
					\State{\textbf{Continue}}
				} \label{alg:disseminate:prefix_filter}
				\If{ $s \in \mathcal{R}$}
				{
					\State{$\mathcal{R}[s] +:= \wt(s.w) \cdot \wt(m.w)$}
				}
				\Else
				{
					\State{$\mathcal{R} := \mathcal{R} \cup s$; $\mathcal{R}[s] := \wt(s.w) \cdot \wt(m.w)$}
				}
				\State{$pos :=$ the position of keyword $w$ in $s.\psi$}
				\If (\tcc*[f]{Theorem~\ref{theorem:single:positional_filter}}){$\mathcal{R}[s] + \tsimub(s.\psi[pos+1:],m.\psi[i+1:]) < \tsimlb(s.\psi,m.\psi)$}
				{
					\State{$\mathcal{R}[s] := -\infty$} \label{alg:disseminate:set_infty}
				} \label{alg:disseminate:positional_filter}
				
			}\label{alg:disseminate:each_sub}
				
		}\label{alg:disseminate:each_group}

	}
}\label{alg:disseminate:each_cell}
\State{\textsf{VerifyAndUpdate}($\mathcal{R}$) and disseminate updated results to corresponding subscriptions}\label{alg:disseminate:verify}
\end{algorithm}

\subsection{Index Maintenance}
\label{sec:disseminate:index_maintenance}
Our indexing structure can also support subscription update efficiently.
For a new subscription $s$, we first find the leaf cell containing its location, and then insert it into the \ifile with $O(|s.\psi| \cdot \log{|\mathcal{G}|})$ cost.
Note that the statistics mentioned above need to be updated accordingly.
For an expired subscription, we simply delete it from index and update the statistics if necessary.

\section{Top-k Re-evaluation}
\label{sec:refill}
In this section, we present the details of \moduletwo.
We first introduce some background knowledge for \skyband in Section~\ref{sec:refill:naive_kskyband}.
Then we present our \textit{cost-based skyband} technique in detail in Section~\ref{sec:refill:cost_kskyband}.
In the following of this paper, we denote the \skyband buffer (either fully or partially) of a subscription $s$ as $s.\mathcal{A}$ for simplicity, and the exact \topk results are denoted as $s.\mathcal{A}_k$.
Meanwhile, $s.k$ is denoted as $k$ if it is clear from context.

\subsection{K-Skyband}
\label{sec:refill:naive_kskyband}
The idea of utilizing \skyband to reduce the number of re-evaluations for \topk queries over a sliding window is first proposed in~\cite{DBLP:conf/sigmod/MouratidisBP06}.
In particular, for a given subscription $s$, only the messages in its corresponding \skyband can appear in its \topk results over the sliding window, thus being maintained.
Following are formal definitions of \textit{dominance} and \textit{\skyband}.
\begin{definition}[Dominance]
A message $m_1$ dominates another message $m_2$ w.r.t. a subscription $s$ if both $\score(s,m_1) \geq \score(s,m_2)$ and $m_1.t > m_2.t$ hold.
\end{definition}
\begin{definition}[\skyband]
The \skyband of a subscription $s$, denoted as $s.\mathcal{A}$, contains a set of messages which are dominated by less than $k$ other messages.
\end{definition}
Instead of keeping \skyband over all the messages in the sliding window,
which is cost-prohibitive, Mouratidis~\emph{et al.}~\cite{DBLP:conf/sigmod/MouratidisBP06} maintain a \textit{partial \skyband}.
Specifically, they only maintain the messages with score not lower than a threshold $s.\theta$, where $s.\theta$ is the $\kscore(s)$ after the most recent \topk re-evaluation for $s$ and remains unchanged until next re-evaluation is triggered.
However, as our experiments suggest, the method in~\cite{DBLP:conf/sigmod/MouratidisBP06} may result in expensive computational cost due to the improper selection of $s.\theta$.

To alleviate the above problem, we propose a novel \textit{cost-based \skyband} technique, which judiciously selects a best threshold $s.\theta$ for the \skyband maintenance of each subscription.
To start with, we present an overview of our \topk re-evaluation algorithm in Algorithm~\ref{alg:topk_refill}.
For each subscription $s$ containing the expired message $m$, if the size of $s.\mathcal{A}$ after deleting $m$ is less than $k$, we need to re-evaluate its \topk results from scratch.
Specifically, we first compute a proper threshold $s.\theta$ based on our cost model (Line~\ref{alg:topk_refill:compute_threshold}),
and then re-compute \skyband buffer $s.\mathcal{A}$ based on $B$, which contains all the messages with score at least $s.\theta$ (Line~\ref{alg:topk_refill:refill_from_index} and Line~\ref{alg:topk_refill:compute_skyband})
\footnote{The same technique in~\cite{DBLP:conf/sigmod/MouratidisBP06} is used to compute \skyband.}.
Note that $B$ can be computed by utilizing message index.
Finally, we extract \topk results from $s.\mathcal{A}$ (Line~\ref{alg:topk_refill:extract_topk}).
The key challenge here is to estimate a best threshold $s.\theta$, which will be discussed in the following in detail.
We remark that we use the term \textit{re-evaluation} to refer in particular to the \topk re-computation against message index.

%

\begin{algorithm}[tb]
\SetVline 
\SetFuncSty{textsf}
\SetArgSty{textsf}
\small
\caption{\textsf{TopkRe-evaluation}($m$)}
\label{alg:topk_refill}
\Input
{
    $m:$ an expired message \\
}
\For{each subscription $s$ whose \skyband buffer contains $m$}
{
	\State{Delete $m$ from $s.\mathcal{A}$}
	\If{$|s.\mathcal{A}| < k$}
	{
		\State{Compute the best $s.\theta$ based on our cost model}\label{alg:topk_refill:compute_threshold}
		\State{Retrieve $B := \{ m | m \in \mathcal{W} ~\&\&~ \score(s,m) \geq s.\theta \}$}\label{alg:topk_refill:refill_from_index}
		\State{$s.\mathcal{A} :=$ \skyband of $B$}\label{alg:topk_refill:compute_skyband}
	}
	\State{Extract $s.\mathcal{A}_k$ from $s.\mathcal{A}$}\label{alg:topk_refill:extract_topk}
}
\end{algorithm}

\subsection{Cost-based K-Skyband}
\label{sec:refill:cost_kskyband}
The general idea of our cost-based \skyband model is to select a best threshold $s.\theta$ for each subscription such that the overall cost defined in the cost model can be minimized.
The following theorem guarantees that, as long as we maintain a partial \skyband over all the messages with score not lower than $s.\theta$, we can extract \topk results from partial \skyband safely when some message expires.
\begin{theorem}
\label{theorem:cost_sky_correctness}
Given a subscription $s$, let $\kscore_{last}(s)$ be the $\kscore(s)$ after the most recent \topk re-evaluation for $s$.
We always have $s.\mathcal{A}_k \subseteq s.\mathcal{A}$ if the following conditions hold: (1) $|s.\mathcal{A}| \geq k$; (2) $s.\mathcal{A}$ is a partial \skyband which is built over all the messages with score at least $s.\theta$ in the sliding window, where $ 0 \leq s.\theta \leq \kscore_{last}(s)$.
\end{theorem}
\begin{proof}
We prove it by contradiction.
Assuming there exists a message $m \in s.\mathcal{A}_k$ while $m \notin s.\mathcal{A}$,
then we discuss two possible cases:
(1) $\score(s,m) \geq s.\theta$; (2) $\score(s,m) < s.\theta$.
For the first case, since $m \notin s.\mathcal{A}$, $m$ must be dominated by more than $k$ messages in $s.\mathcal{A}$, which indicates it cannot be \topk results, i.e., $m \notin s.\mathcal{A}_k$.
For the second case, at least $k$ messages in $s.\mathcal{A}$ must have a higher score than $m$ because $|s.\mathcal{A}| \geq k$ and all the messages in $s.\mathcal{A}$ have score at least $s.\theta$. Thus, $m$ still cannot be \topk results.
Thus, the original assumption does not hold, which immediately indicates $s.\mathcal{A}_k \subseteq s.\mathcal{A}$.
\end{proof}

Thus, based on Theorem~\ref{theorem:cost_sky_correctness}, we can safely extract \topk results from \skyband buffer $s.\mathcal{A}$ when $|s.\mathcal{A}| \geq k$;
when $|s.\mathcal{A}| < k$, we have to re-evaluate from message index.

Our cost-based \skyband model, based on Theorem~\ref{theorem:cost_sky_correctness}, aims to find the best $s.\theta$ such that the overall cost can be minimized for each subscription.
We mainly consider two costs.
The first one is \textit{\skyband maintenance cost}, denoted as $\mathcal{C}_{sm}(s)$, which is triggered upon message arrival and expiration.
The second one is \textit{\topk re-evaluation cost}, denoted as $\mathcal{C}_{re}(s)$, which is triggered when some message expires and the \topk results can no longer be retrieved from \skyband buffer.
We aim to estimate the expected overall cost w.r.t. each message update, i.e., message arrival and message expiration, each of which we assume occurs with probability $\frac{1}{2}$ as the window slides.
To simplify the presentation, we denote as $\prob(s.\theta)$ the probability that the score between a random message and a subscription $s$ is at least $s.\theta$.
We may immediately derive $\prob(s.\theta)$ for a given $s.\theta$ from historical data, assuming the score follows previous distribution.
The details of these two costs are presented in the following respectively.

\subsubsection{$K$-Skyband Maintenance Cost}
The maintenance of \skyband is triggered when the following two types of updates happen, both with probability $\frac{1}{2} \cdot \prob(s.\theta)$,
where $\frac{1}{2}$ is the probability of message arrival or message expiration due to the count-based sliding window,
and $\prob(s.\theta)$ is the probability that the score between a random message and $s$ is at least $s.\theta$.
Please note that if the independence assumption does not hold for messages, the above probabilities cannot be estimated accurately, and we may resort to utilizing historical data for the estimation.

The first type of update is triggered when a message $m$ with score at least $s.\theta$ arrives.
Apart from the insertion of $m$ into $s.\mathcal{A}$,
the dominance counters of all the messages in $s.\mathcal{A}$ with score not higher than $\score(s,m)$ will increase by 1, and the messages with dominance counter equal to $k$ will be evicted.
Since we implement our \skyband buffer with a linked list sorted by $\score(s,m)$.
The above operations can be processed in $O(|s.\mathcal{A}|)$ time with a linear scan.
The next challenge is to estimate $|s.\mathcal{A}|$.
Based on the independence assumption between score dimension and time dimension, the expected number, i.e.,~$|s.\mathcal{A}|$, of messages in the partial \skyband is $ k \cdot \ln(\frac{|\mathcal{W}| \cdot \prob(s.\theta)}{k})$~\cite{DBLP:journals/tkde/ZhangLYKZY10}, where $|\mathcal{W}|$ is the size of sliding window.
Please note that if the independence assumption does not hold, the worst case space complexity will be $|\mathcal{W}| \cdot \prob(s.\theta)$.

The second type of update occurs when an old message $m$ among the \skyband buffer of $s$ expires.
In this case, we only need to delete $m$ from $s.\mathcal{A}$ in $O\left(|s.\mathcal{A}| \right)$ time.
Note that $m$ does not dominate any remaining messages and therefore the dominance counters of the remaining messages are not affected.
Finally, we get the total cost of \skyband maintenance as follows:
\begin{eqnarray}
\label{equ:csky:maintenance_cost}
\mathcal{C}_{sm}(s) &=& \frac{1}{2} \cdot \prob(s.\theta) \cdot |s.\mathcal{A}|  + \frac{1}{2} \cdot \prob(s.\theta) \cdot |s.\mathcal{A}|  \nonumber \\
&=& \prob(s.\theta) \cdot k \cdot \ln(\frac{|\mathcal{W}| \cdot \prob(s.\theta)}{k})
\end{eqnarray}

\subsubsection{Top-$k$ Re-evaluation Cost}
The \topk re-evaluation cost can be formalized as:
\begin{equation}
\label{equ:csky:topk_reevaluation_cost}
\mathcal{C}_{re}(s) = \mathcal{C}_{topk}(s) \cdot \frac{1}{\mathbb{Z}(s)}
\end{equation}
where $\mathcal{C}_{topk}(s)$ is the average \topk computation cost over message index for subscription $s$, and $\mathbb{Z}(s)$ is the expected number of message updates that is required to trigger \topk re-evaluation, i.e., leading to $|s.\mathcal{A}| < k$.
The value of $\mathcal{C}_{topk}(s)$ can be estimated by the average of previous \topk computation cost against message index.
The remaining issue is how to estimate $\mathbb{Z}(s)$, which is non-trivial.

To solve this problem, we model the streaming updating process as a simple \textit{random walk}.
A random walk is a stochastic sequence ${RW_n}$, with $RW_0$ being the starting position, defined by
$RW_n = \sum_{i=1}^{n}{X_i}$
where ${X_i}$ are independent and identically distributed random variables (i.e., i.i.d.).
The random walk is \textit{simple} if $\prob(X_i = 1) = p$, $\prob(X_i = -1) = q$ and $\prob(X_i = 0) = r$, where $p+q+r=1$.
We map the estimation of $\mathbb{Z}(s)$ into a simple random walk as follows.
We model the change of \skyband buffer $s.\mathcal{A}$ w.r.t. each message update as an i.i.d. variable $X_i$.
$X_i$ is set to $1$ when the size of $s.\mathcal{A}$ is increased by 1 at $i$-th step,
while $X_i$ is set to $-1$ when the size of $s.\mathcal{A}$ is decreased by 1.
When the size of $s.\mathcal{A}$ does not change, $X_i$ is set to 0.
Unfortunately, it is difficult to estimate the probability of $\prob(X_i = 1)$ and $\prob(X_i = -1)$ for each message update, due to the eviction of messages by dominance relationship.
For example, for a new message, the size of $s.\mathcal{A}$ may decrease rather than increase due to the eviction of messages with dominance counter reaching $k$.
To address this problem, rather than estimating $\mathbb{Z}(s)$ for $s.\mathcal{A}$ maintenance, we estimate $\mathbb{Z}'(s)$ for $s.\mathcal{A}'$, which contains all the messages with score not lower than $s.\theta$.
Specifically, when we maintain $s.\mathcal{A}'$, we do not consider the dominance relationship between messages for each message update, and thus the messages dominated by $k$ (or more) messages are not evicted.
Clearly, $s.\mathcal{A}'$ is a superset of $s.\mathcal{A}$, i.e., $s.\mathcal{A} \subseteq s.\mathcal{A}'$.
The following theorem guarantees that $\mathbb{Z}(s)$ is equal to $\mathbb{Z}'(s)$.
\begin{theorem}
\label{theorem:csky:equivalent_reduction}
The expected number of message updates that is required to trigger \topk re-evaluation for $s.\mathcal{A}$ maintenance is the same as that for $s.\mathcal{A}'$ maintenance, i.e., $\mathbb{Z}(s) = \mathbb{Z}'(s)$.
\end{theorem}
\begin{proof}
To show $\mathbb{Z}(s)$ is equal to $\mathbb{Z}'(s)$, it is sufficient to prove that $|s.\mathcal{A}| < k$ \textit{if and only if} $|s.\mathcal{A}'| < k$ at any point.
Initially, we have $s.\mathcal{A} \cup s.\mathcal{A}_{evict} = s.\mathcal{A}'$ and $s.\mathcal{A} \cap s.\mathcal{A}_{evict} = \emptyset$, where $s.\mathcal{A}_{evict}$ is a set of messages evicted by messages in $s.\mathcal{A}$.
We prove that after each message update, $s.\mathcal{A} \cup s.\mathcal{A}_{evict} = s.\mathcal{A}'$ always holds.
As to the arrival of a new message $m$,
if $m$ is inserted into $s.\mathcal{A}'$, it will also be inserted into $s.\mathcal{A}$ since $m$ will not be dominated by any existing message due to its freshness;
and vice versa.
Note that some messages may be evicted from $s.\mathcal{A}$ to $s.\mathcal{A}_{evict}$ due to dominance relationship.
Thus,  $s.\mathcal{A} \cup s.\mathcal{A}_{evict} = s.\mathcal{A}'$ holds.
As to the expiration of an old message $m$,
(1) If $m \in s.\mathcal{A}$, $m$ will also expire from $s.\mathcal{A}'$, while $s.\mathcal{A}_{evict}$ does not change.
(2) If $m \notin s.\mathcal{A}$, $m$ will expire from both $s.\mathcal{A}_{evict}$ and $s.\mathcal{A}'$.
Thus, $s.\mathcal{A} \cup s.\mathcal{A}_{evict} = s.\mathcal{A}'$ still holds.
Therefore, when $|s.\mathcal{A}| < k$ occurs,
$s.\mathcal{A}_{evict}$ must be empty because there is no $k$ messages in $s.\mathcal{A}$ that can dominate any message in $s.\mathcal{A}_{evict}$.
Thus, $|s.\mathcal{A}'| = |s.\mathcal{A}| < k$ holds.
Contrarily, when $|s.\mathcal{A}'| < k$ occurs, it is immediate that $|s.\mathcal{A}| < k$ because $s.\mathcal{A}$ is always a subset of $s.\mathcal{A}'$.
Therefore, the theorem holds.
\end{proof}

Based on the above theorem, we turn to estimate $\mathbb{Z}'(s)$, which is much easier.
Now the probability distribution of $X_i$ can be estimated as:
\begin{equation}
\label{equ:csky:X_i}
  \prob(X_i) =
  \begin{cases}
    \frac{1}{2} \cdot \prob(s.\theta) & \text{if } X_i = 1,
    \\
    \frac{1}{2} \cdot \prob(s.\theta) & \text{if } X_i = -1,
    \\
    1 - \prob(s.\theta) & \text{if } X_i = 0.
  \end{cases}
\end{equation}

We denote the initial size of $s.\mathcal{A}'$ as $|s.\mathcal{A}'_{init}|$.
Now, the estimation of $\mathbb{Z}'(s)$ is equivalent to a well-known random walk problem, namely \textit{Monkey at the cliff with reflecting barriers}~\cite{feller2008introduction}.
Specifically, we set the starting position $RW_0$ as $|s.\mathcal{A}'_{init}|$ and the destination position as $k-1$; the i.i.d. variable $X_i$ is defined as Equation~\ref{equ:csky:X_i}; and the reflecting barrier is set as $2 \cdot RW_0$.
By applying some mathematical reduction based on the property of random walk~\cite{feller2008introduction}, we get the following result.
\begin{eqnarray}
\mathbb{Z}'(s) &=& \frac{ 2 \cdot \left( |s.\mathcal{A'}_{init}| - k + 1  \right) \cdot |s.\mathcal{A'}_{init}| }{ \prob(s.\theta)}  \nonumber \\
&+& \frac{ \left( |s.\mathcal{A'}_{init}| - k + 1  \right) \cdot \left( |s.\mathcal{A'}_{init}| - k + 2 \right) }{ \prob(s.\theta)}
\end{eqnarray}
where $|s.\mathcal{A'}_{init}|$ can be estimated as $\prob(s.\theta) \cdot |\mathcal{W}|$.
Thus, the \topk re-evaluation cost in Equation~\ref{equ:csky:topk_reevaluation_cost} can be estimated by replacing $\mathbb{Z}(s)$ with $\mathbb{Z}'(s)$.
Based on Equation~\ref{equ:csky:maintenance_cost} and Equation~\ref{equ:csky:topk_reevaluation_cost}, we get our final cost model:
\begin{equation}
\label{equ:csky:overall_cost}
\mathcal{C}(s) = \mathcal{C}_{sm}(s) + \mathcal{C}_{re}(s)
\end{equation}

To minimize Equation~\ref{equ:csky:overall_cost} where the only variable is $s.\theta$, we employ an incremental estimation algorithm similar to gradient descent~\cite{avriel2003nonlinear} to compute the best value of $s.\theta$.

\noindent \textbf{Remark.}
To accommodate our cost-based skyband model with the message dissemination algorithm, we need to replace $\kscore(s)$ in Section~\ref{sec:disseminate} with $s.\theta$ such that any message with score not lower than $s.\theta$ will be considered to possibly affect the \topk results of $s$.
Moreover, since our dominance definition simply depends on the 2-dimensional score-time space while is irrelevant to the exact score function, our technique can be easily applied to other \topk monitoring problems with different score functions.

%
%

\subsection{Discussions}
\label{sec:refill:discussions}
\noindent \textbf{Initialization of incoming subscriptions.}
The initialization of a new subscription $s$ can be processed in a similar way to Algorithm~\ref{alg:topk_refill}, where we regard the initial size of $s.\mathcal{A}$ as 0 and execute Lines~\ref{alg:topk_refill:compute_threshold}-\ref{alg:topk_refill:compute_skyband} in Algorithm~\ref{alg:topk_refill} sequentially.

\noindent \textbf{Time-based sliding window model.}
Our techniques discussed above can also be extended to support time-based sliding window model, where only the messages within a recent time period are maintained.
Unlike count-based sliding window whose size is constant, the size of time-based sliding window, i.e., $|\mathcal{W}|$, can change at any time due to the volatile message workload.
To estimate $|\mathcal{W}|$, we assume that the message workload does not change significantly in the near future.
Then we can estimate $|\mathcal{W}|$ by the historical message workload from a recent period.
Another difference is that the probability of message arrival (resp.  expiration) cannot be regarded as $\frac{1}{2}$ trivially as indicated in Equation~\ref{equ:csky:maintenance_cost} and Equation~\ref{equ:csky:X_i}, because the number of message arrival and the number of message expiration are possibly rather different in each timestamp.
To alleviate this issue, we resort to estimating the above probabilities based on the relative proportion of message arrival and expiration within a recent time period.
Then the probabilities (e.g., $\frac{1}{2}$) in Equation~\ref{equ:csky:maintenance_cost} and Equation~\ref{equ:csky:X_i} are updated accordingly.
We also conduct experiments to verify the efficiency of our techniques under time-based sliding window in Section~\ref{sec:exp}.

\section{Distributed Processing}
\label{sec:distributed}

In this section, we introduce \oursdis, a distributed \topk \sk \pubsub system built on top of \storm.
We first touch some background knowledge about \storm in Section~\ref{sec:distributed:storm_background}, followed by the detailed system framework in Section~\ref{sec:distributed:framework}.
Four novel distribution mechanisms are discussed in Section~\ref{sec:distributed:dis_mechanism}, which manage to partition the subscriptions and messages to multiple bolt instances for parallel processing.
The maintenance issue is finally discussed in Section~\ref{sec:distributed:maintenance}.
To the best of our knowledge, this is the first work to extend \topk \sk \pubsub system to a distributed environment.

\subsection{Storm Background}
\label{sec:distributed:storm_background}
\storm is a distributed, fault-tolerant and general-purpose stream processing system.
Unlike \hadoop\footnote{Apache Hadoop project. https://hadoop.apache.org/} which is mainly designed to process batch tasks, \storm is designed to process streaming data continuously and endlessly.
There are three key abstractions in \storm: \textbf{spout}, \textbf{bolt} and \textbf{topology}.
A spout is a source of streams, which reads input stream from external resources, such as Twitter API\footnote{https://dev.twitter.com/rest/public}.
A bolt is a processing unit responsible for data processing, which handles any number of input streams and produces any number of new output streams.
A topology is a network of spouts and bolts, with each directed edge in the network representing a bolt subscribing to the output stream of some other spout or bolt.
Essentially, topology defines the working flow of a real-time computation task, which is similar to a MapReduce job~\cite{DBLP:journals/cacm/DeanG08}.
\storm employs various \textbf{stream groupings} techniques\footnote{http://storm.apache.org/releases/0.10.0/Concepts.html}, such as shuffle grouping and fields grouping, to specify for each bolt instance which streams it should receive as input.
Figure~\ref{fig:distributed:storm_flow} depicts a simple \storm working flow where there are two spouts and five bolts connected by directed edges.
Note that each spout/bolt can have many parallel-running instances/tasks.
\begin{figure}[t]
\centering
\includegraphics[width=0.8\columnwidth]{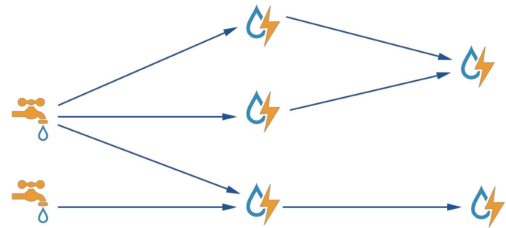}
\caption{\small{A simple storm flow}}
\label{fig:distributed:storm_flow}
\end{figure}

\subsection{Framework}
\label{sec:distributed:framework}
Figure~\ref{fig:distributed:topology} shows the topology of our \oursdis, which contains five main components:
\begin{itemize}
\item \textbf{Subscription/Message spouts.} Subscription spouts receive new subscription request while message spouts collect message stream from external source, e.g., Twitter API.
The incoming streams are then subscribed by other components of the topology.

\item \textbf{Distribution bolts.} Distribution bolts receive streams from spouts and navigate them downwards to the subscription bolts according to some carefully designed distribution mechanisms (Section~\ref{sec:distributed:dis_mechanism}), aiming to achieve good workload balance.
This component is critical to the overall communication cost and throughput of our system.
Note that the distribution bolts will also route new messages to the message bolts to ensure that the sliding window is always up-to-date.

\item \textbf{Subscription bolts.} Subscription bolts partition the subscription index and result buffer among multiple parallel-running tasks.
A new subscription from a distribution bolt is inserted into one or more subscription bolts, and a new message is processed simultaneously among multiple bolts.
Note that distribution bolts and subscription bolts together correspond to \moduleone in the centralized version discussed in Section~\ref{sec:framework}.

\item \textbf{Message bolts.} Message bolts maintain the sliding window in a distributed manner, and each bolt contains part of the sliding window.
A message index (e.g.,~\irtree,~\sti) is built over the messages residing in each bolt.
The \topk re-evaluation request for a subscription $s$ issued by subscription bolts will be processed concurrently among all message bolts, each generating a partial message buffer consisting of all the messages with score at least $s.\theta$.
Note that each message is stored in only one message bolt.

\item \textbf{Aggregation bolts.} Aggregation bolts are introduced to aggregate the partial message buffer generated by message bolts.
Then the final \skyband buffer is computed and forwarded back to the subscription bolts where $s$ resides.
Note that message bolts and aggregation bolts together form the counterpart to the \moduletwo in its centralized version.
\end{itemize}

All the stream groupings in the topology are summarized in Table~\ref{tb:stream_groupings}.
We remark that the result buffer can be easily swapped to any persistent state, such as Memcached\footnote{http://storm.apache.org/releases/0.10.0/Trident-tutorial.html} and HDFS~\cite{DBLP:conf/mss/ShvachkoKRC10}, to support various applications.

\begin{figure}[t]
\centering
\includegraphics[width=1.0\columnwidth]{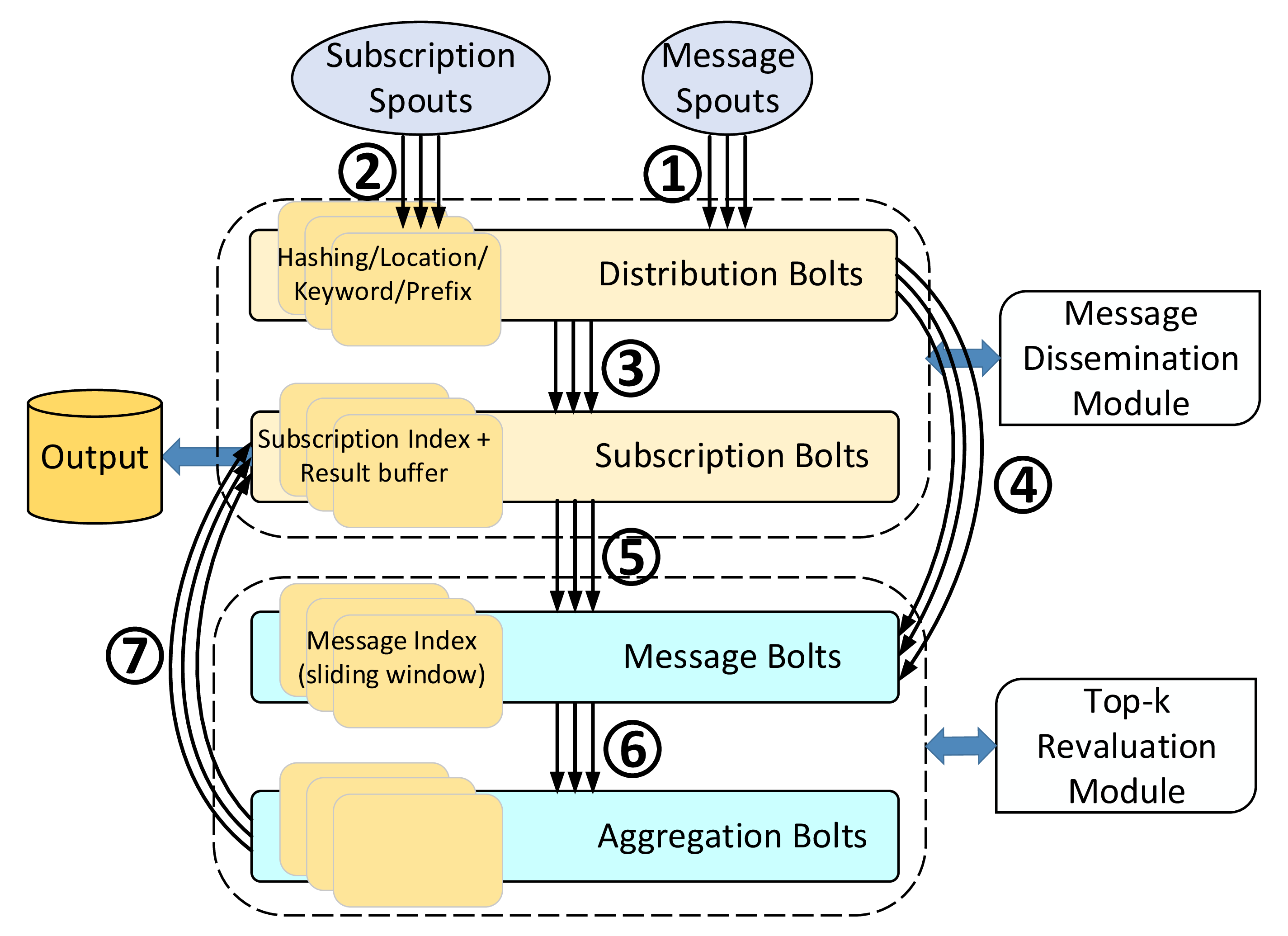}
\caption{\small{\oursdis topology (solid arrows indicate stream flow between components)}}
\label{fig:distributed:topology}
\end{figure}

\begin{table*}[t]
  	\caption{Stream grouping methods. ID in first column corresponds to ID in Figure~\ref{fig:distributed:topology}.}\label{tb:stream_groupings}
  	\centering
    \begin{tabular}{|c|p{1.6cm}|p{1.8cm}|p{2.2cm}|p{9.5cm}|}
      \hline
      \textbf{ID} & \textbf{Source} & \textbf{Destination} & \textbf{Grouping} & \textbf{Description} \\ \hline \hline
      \circled{1} & message spouts & distribution bolts & shuffle grouping &  Each message is distributed to only one distribution bolt randomly. \\ \hline
      \circled{2} & subscription spouts & distribution bolts & shuffle grouping &  Each subscription is distributed to only one distribution bolt randomly. \\ \hline
      \circled{3} &	distribution bolts & subscription bolts & all/direct grouping & Subscriptions/messages are distributed to subscription bolts with either all grouping or direct grouping (depending on the distribution mechanism employed). \\ \hline
      \circled{4} &	distribution bolts & message bolts & direct grouping & Each message is distributed to only one message bolt based on its unique id to keep the sliding window up-to-date. \\ \hline	
      \circled{5} &	subscription bolts & message bolts & all grouping & Each \topk re-evaluation request is distributed to all message bolts for parallel processing. \\ \hline
      \circled{6} &	message bolts & aggregation bolts & fileds grouping & The partial result buffers from the same subscription are distributed to the same aggregation bolt. \\ \hline
      \circled{7} &	aggregation bolts & subscription bolts & direct grouping & The aggregated \skyband buffer is forwarded back to one or more subscription bolts where this subscription resides. \\ \hline
    \end{tabular}
\end{table*}

\noindent \textbf{Working flow.}
When a new message $m$ is digested by a message spout, it will be delivered to a distribution bolt (stream \circled{1}, see Figure~\ref{fig:distributed:topology}).
The distribution bolt then will navigate $m$ to some of subscription bolts (stream \circled{3}) such that it can be processed against local subscription index in a parallel manner.
The distribution bolt will also disseminate $m$ to a message bolt (stream \circled{4}) to keep the sliding window therein up-to-date.
When $m$ expires, the \topk re-evaluations triggered by $m$ in the subscription bolts will be emitted to message bolts (stream \circled{5}) and then to aggregation bolts (stream \circled{6}), where the \skyband buffers of all affected subscriptions will be re-computed and forwarded back to the subscription bolts (stream \circled{7}).
Similarly, a new subscription $s$ will be firstly delivered from a subscription spout to a distribution bolt (stream \circled{2}), and then routed to one or more subscription bolts (stream \circled{3}) for indexing based on the distribution mechanism.
The initial result of $s$ will be computed from message bolts (stream \circled{5}) and aggregation bolts (stream \circled{6}) accordingly and forwarded back to subscription bolt (stream \circled{7}) where $s$ is indexed.
An unregistered subscription will be simply deleted from all its residing subscription bolts.

\noindent \textbf{Challenges.}
As the number of subscriptions increases, the subscription bolts become the main bottleneck of our system.
Meanwhile, the communication cost between distribution bolts and subscription bolts dominates all the other communication cost as we increase the number of subscription bolts.
Thus, the key challenge in \oursdis is to develop an efficient distribution mechanism to assign subscriptions and messages only to some inevitable subscription bolts, such that both small communication cost and high throughput can be realized while still guaranteeing the correctness of our algorithms.
Furthermore, the distribution mechanism should be able to handle workload balance, since both the subscription workload and message workload in real life are extremely biased regarding keywords and locations.
At last, the distribution mechanism should be light-weighted without consuming many CPU and memory resources.

\subsection{Distribution Mechanism}
\label{sec:distributed:dis_mechanism}
In this section, we present several novel, efficient and light-weighted distribution mechanisms, which can be integrated into distribution bolts.
For the ease of exposition, we assume we already have a set of existing subscriptions $\mathcal{S}$ and a random message $m$ sampled from message stream. We denote the number of subscription bolts as $N_{sb}$, with each bolt identified by a partition index ranging from $0$ to $N_{sb} - 1$.
A distribution mechanism aims to partition $\mathcal{S}$ into $N_{sb}$ subscription bolts and navigate the message $m$ to relevant subscription bolts for \topk dissemination.
In the following, we propose four different distribution methods, namely \textit{hashing-based}, \textit{location-based}, \textit{keyword-based} and \textit{prefix-based}, respectively.

\subsubsection{Hashing-based Method}
Hashing-based method partitions the subscriptions based on a uniform hashing function defined as follows:
\begin{equation}
h(s) = s.id \mod N_{sb}
\end{equation}
where $s.id$ is a unique id assigned to each subscription, and $h(s)$ is the bolt index where $s$ should be allocated.

\noindent \textbf{Analysis.}
Since each subscription is allocated to only one bolt, the \textit{replication ratio} of subscriptions is 1.
The replication ratio here indicates the number of times a subscription has been stored in the system.
Note that we ignore the replication ratio of messages, because it is always 1 regardless of the distribution mechanisms.
Meanwhile, for any new message, it needs to be distributed to all the subscription bolts to ensure the correctness.
Thus, the average \textit{communication cost} of each message is $N_{sb}$.
Note that we only consider the communication cost between distribution bolts and subscription bolts w.r.t. each message since it is dominant.
For example, on a cluster with 32 subscription bolts, the communication cost between distribution bolts and subscription bolts account for more than $90\%$ of total communication cost.

\subsubsection{Location-based Method}
\label{sec:distributed:dis_mechanism:location-based}
Hashing-based method is simple and can achieve very good workload balance, because the number of subscriptions in each subscription bolt is nearly the same by the nature of uniform hashing.
However, it does not take the location factor into consideration.
Intuitively, distributing subscriptions with high spatial similarity into the same bolt can lead to lower \amp cost and thus higher throughput, since we can acquire better spatial bounds as discussed in Section~\ref{sec:disseminate:ipt:spatial_bound}.
On the other hand, it is also pivotal to balance subscription and message workloads among the subscription bolts.
To this end, we propose a cost-based spatial partition schema using \kdtree~\cite{DBLP:journals/cacm/Bentley75}, where the leaf nodes of \kdtree form a disjoint partition of the whole space.
For each leaf node, we allocate a subscription bolt to process all the subscriptions whose locations are inside the leaf node.
Formally, given any node, denoted as $nd$, in \kdtree, we estimate its cost by:
\begin{equation}
C(nd) = N(nd) \times p(nd)
\end{equation}
where $N(nd)$ is the number of subscriptions whose locations are inside $nd$ and $p(nd)$ is the probability that a random incoming message falls inside $nd$.
Note that $p(nd)$ can be easily estimated from historical message workloads.
The cost-based spatial partition algorithm is depicted in Algorithm~\ref{alg:kdtree_partition}, which is very similar to the original \kdtree construction algorithm.
The key difference is that, unlike \kdtree which selects a line halving all the points along x-axis or y-axis alternately, our algorithm tries to find a splitting line which minimizes the cost difference between two children (Lines~\ref{alg:kdtree_partition:find_x_split} and~\ref{alg:kdtree_partition:find_y_split}) in order to achieve workload balance.
We limit the total number of partitions (i.e., leaf nodes) by setting a $maxdepth$ value (Line~\ref{alg:kdtree_partition:maxdepth}).

\begin{algorithm}[tb]
\SetVline 
\SetFuncSty{textsf}
\SetArgSty{textsf}
\small
\caption{\textsf{SpatialPartition}($nd$, $depth$)}
\label{alg:kdtree_partition}
\Input
{
    $nd:$ a node in \kdtree\\
    $depth: $ the depth of $nd$
}
\Output{A spatial partition consisting of all the leaf nodes in \kdtree}

\If{$depth == maxdepth$}
{
	\State{\Return}
}\label{alg:kdtree_partition:maxdepth}

\If(\tcc*[f]{split by x-coordinate}){$depth$ is even}
{
	\State{Find a value $x$ from the x-coordinates of all the subscriptions in $nd$, which leads to minimum $|C(nd_1) - C(nd_2)|$, where $nd_1$ and $nd_2$ are two child nodes split based on $x$}\label{alg:kdtree_partition:find_x_split}
}
\Else(\tcc*[f]{split by y-coordinate})
{
		\State{Find a value $y$ from the y-coordinates of all the subscriptions in $nd$, which leads to minimum $|C(nd_1) - C(nd_2)|$, where $nd_1$ and $nd_2$ are two child nodes split based on $y$}\label{alg:kdtree_partition:find_y_split}
}
\State{Assign subscriptions in $nd$ to $nd_1$ and $nd_2$ based on the splitting value}
\State{\textsf{SpatialPartition}($nd_1$, $depth + 1$)}
\State{\textsf{SpatialPartition}($nd_2$, $depth + 1$)}
\end{algorithm}

\noindent \textbf{Analysis.}
The time complexity of Algorithm~\ref{alg:kdtree_partition} is bounded by $O(|\mathcal{S}| \times \log{|\mathcal{S}|})$ as we need to sort subscriptions in $\mathcal{S}$ by x-coordinate and y-coordinate respectively beforehand and then do a divide-and-conquer partition.
Each subscription will be assigned to only one subscription bolt containing its location.
Thus, the replication ratio is 1.
On the other hand, each incoming message needs to be delivered to all subscription bolts to guarantee the algorithm correctness.
Thus, the average communication cost is $N_{sb}$.

\subsubsection{Keyword-based Method}
Both hashing-based and location-based methods have to send each message to all the subscription bolts, which results in high communication overhead especially when we increase the number of subscription bolts.
To alleviate this issue, we propose a novel keyword-based partition mechanism, which can reduce the communication cost significantly at the cost of small subscription replications.
The general idea is that, each subscription bolt only accounts for a subset of keywords;
thus, each subscription $s$ (resp. message $m$) will be distributed only to the subscription bolts whose keyword sets overlap with $s.\psi$ (resp. $m.\psi$).
To start with, similar to Section~\ref{sec:distributed:dis_mechanism:location-based}, we estimate the processing cost of a subscription bolt with keyword set $W_i$ as follows:
\begin{equation}
\label{equ:keyword_set_cost}
C(W_i) = N(W_i) \times p(W_i)
\end{equation}
where $N(W_i)$ is the number subscriptions whose keywords overlap with $W_i$ and $p(W_i)$ is the probability that a random incoming message contains at least one keyword from $W_i$.
Note that $p(W_i)$ can be estimated from historical message workloads.
We then define \textit{variance} to measure the workload balance as follows:
\begin{equation}
Var = \frac{1}{N_{sb}} \times \sum_{i = 0}^{N_{sb}-1}{ (C(W_i) - C_{\mu} )^2 }
\end{equation}
where $C_{\mu}$ is the average cost over all keyword sets.

We are now ready to present the keyword partition problem: given a keyword vocabulary, denoted as $\mathcal{V}$, where keywords are ordered by their frequencies, we aim to partition $\mathcal{V}$ into $N_{sb}$ subsets, i.e., $W_0, W_1,...,W_{N_{sb}-1}$, each covering a number of consecutive keywords, such that the following conditions are satisfied: (1). $\forall{0 \leq i < j \leq N_{sb}-1}, W_i \cap W_j = \emptyset$, (2). $\cup_{0 \leq i \leq N_{sb}-1}{ W_i} = \mathcal{V}$, (3). $Var$ is minimized.

Since the search space is exponentially large, we resort to an efficient heuristic algorithm to solve this problem, which is demonstrated in Algorithm~\ref{alg:keyword_partition}.
The idea is to recursively partition the keyword set into two subsets by carefully selecting a splitting keyword to balance the cost between two subsets.
The parameter $maxdepth$ here (Line~\ref{alg:keyword_partition:maxdepth}) is used to control the number of partitions we need.
Once we get keyword partitions, a subscription $s$ will be allocated to the $i^{th}$ bolt as long as $s.\psi \cap W_i \neq \emptyset$.
Similarly, a message $m$ will also be distributed to the $i^{th}$ bolt if $m.\psi \cap W_i \neq \emptyset$.

\begin{algorithm}[tb]
\SetVline 
\SetFuncSty{textsf}
\SetArgSty{textsf}
\small
\caption{\textsf{KeywordPartition}($W$, $depth$)}
\label{alg:keyword_partition}
\Input
{
    $W:$ current keyword set\\
    $depth: $ the depth of recursion
}
\Output{A keyword partition}

\If{$depth == maxdepth$}
{
	\State{\Return}
}\label{alg:keyword_partition:maxdepth}
\State{Find a splitting keyword $w$ from $W$ which leads to minimum $|C(W_1) - C(W_2)|$, where $W_1$ contains all the words on the left of $w$ (inclusive) while $W_2$ contains all the keywords on the right of $w$ (exclusive)}\label{alg:keyword_partition:find_split}
\State{\textsf{KeywordPartition}($W_1$, $depth + 1$)}
\State{\textsf{KeywordPartition}($W_2$, $depth + 1$)}
\end{algorithm}

\begin{example}
Figure~\ref{fig:distributed:keyword-based} shows an example of keyword-based method.
Assume there are 10 keywords in the vocabulary, and there are four subscription bolts.
Based on the keyword-based distribution method, subscription $s$ is allocated to bolt 1, bolt 2 and bolt 3 while message $m$ is distributed to bolt 3 and bolt 4.
\end{example}

\begin{figure}[t]
\centering
\includegraphics[width=0.8\columnwidth]{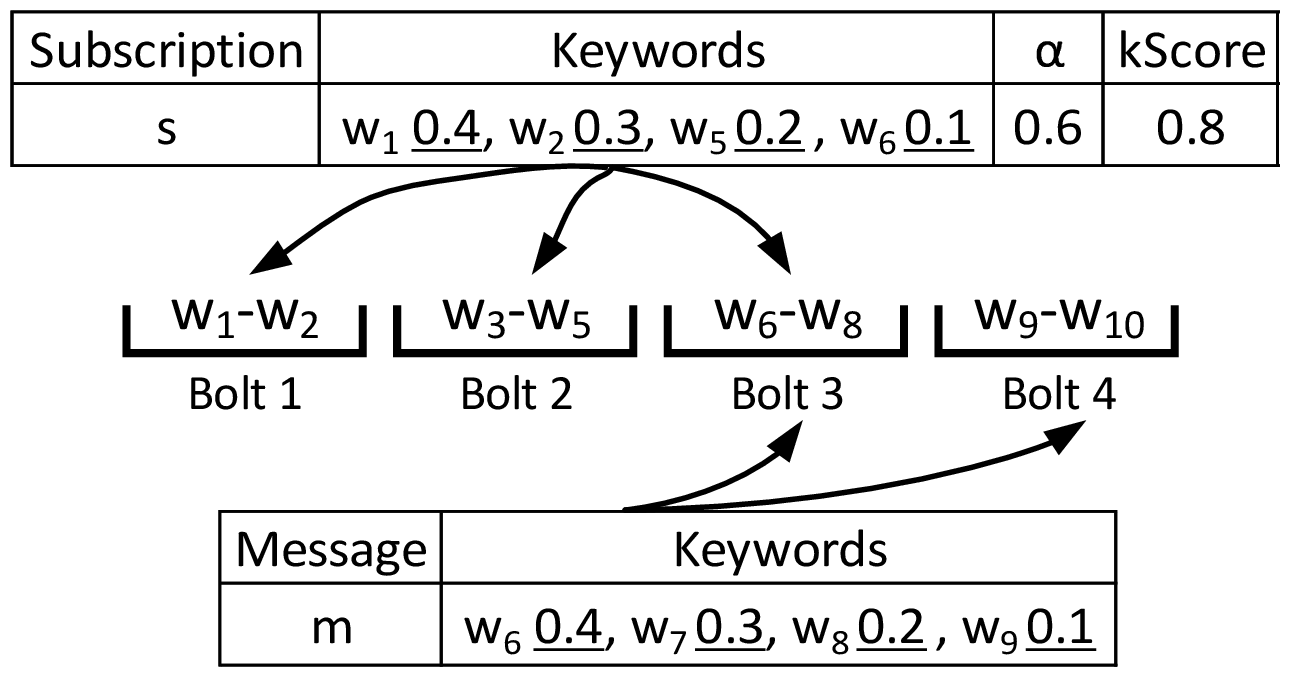}
\caption{\small{Keyword-based example}}
\label{fig:distributed:keyword-based}
\end{figure}

\noindent \textbf{Analysis.}
The time complexity of Algorithm~\ref{alg:keyword_partition} is $O(|\mathcal{V}| \times maxdepth + |\mathcal{S}|)$ if we use a linear scan to find the splitting keyword (Line~\ref{alg:keyword_partition:find_split}).
Each subscription is distributed to at most $|s.\psi|$ bolts and each message is distributed to at most $|m.\psi|$ bolts.
Thus, the replication ratio is at most $|s.\psi|$ while the average communication cost is at most $\min(|m.\psi|, N_{sb})$.

\subsubsection{Prefix-based Method}
Compared to both hashing-based and location-based methods, keyword-based method can reduce the communication cost significantly, especially when we have a large number of subscription bolts (i.e., large $N_{sb}$), since each message is only distributed to keyword-overlapping bolts.
However, the pitfall is that we have to duplicate subscriptions among different subscription bolts to ensure the correctness of our algorithms, which often leads to poor throughput as shown in the experimental part.
To further reduce the subscription replications as well as improve throughput, we propose a light-weighted prefix-based method, which only distributes subscriptions based on their keyword prefixes, rather than all the keywords.
However, we cannot employ the same prefix as defined in Definition~\ref{def:loc_aware_prefix}, because the value of $\ssimub(s.\rho,m.\rho)$ is not available beforehand.
To overcome this issue, we define a new textual similarity threshold $\tsimlb(s.\psi)$ as follows:
\begin{equation}
\label{equ:tsimlb_loose}
\tsimlb(s.\psi) = \frac{\kscore(s)}{1 - s.\alpha} - \frac{s.\alpha}{1 - s.\alpha} \cdot 1.0
\end{equation}
where we always assume $\ssimub(s.\rho,m.\rho)$ to be $1.0$ regardless of the actual location of the incoming message.
Then, we propose a \textit{loose prefix}, denoted as $\pref(s)$, which is defined as follows.
\begin{definition}[Loose Prefix]
\label{def:loose_prefix}
Given a subscription $s$ and a textual similarity threshold $\tsimlb(s.\psi)$,
we use $\pref(s) = s.\psi[1:p]$ to denote the loose prefix of $s$,
where $p = \argmin_{i}{ \{ \wtsum(s.\psi[i+1]) < \tsimlb(s.\psi) \} }$.
\end{definition}
We remark that this definition is similar to Definition~\ref{def:loc_aware_prefix}, except that we replace $\tsimlb(s.\psi,m.\psi)$ with $\tsimlb(s.\psi)$, such that the loose prefix $\pref(s)$ is irrelevant to the message.
The following lemma guarantees the correctness of our distribution mechanism based on loose prefix.
\begin{lemma}
\label{lemma:distirbute_prefix}
Our algorithm is correct as long as we distribute each subscription $s$ (resp. message $m$) only to the subscription bolts whose keyword sets intersect with $\pref(s)$ (resp. $m.\psi$).
\end{lemma}
\begin{proof}
Based on the distribution method, it is immediate to conclude that if $\pref(s) \cap m.\psi \neq \emptyset$ (and thus $m$ might be a candidate of $s$), $m$ must be distributed to at least one subscription bolt where $s$ also resides.
On the other hand, if $\pref(s) \cap m.\psi = \emptyset$, $m$ must not be a candidate of $s$.
This can be easily proved based on a similar deduction from the proof of Lemma~\ref{lemma:prune_s}.
Thus, there is no need to distribute the subscription $s$ to the subscription bolts whose keyword sets overlap with $s.\psi - \pref(s)$.
\end{proof}

\begin{example}
Following the same example in Figure~\ref{fig:distributed:keyword-based}, Figure~\ref{fig:distributed:prefix-based} depicts an example of prefix-based method.
From Equation~\ref{equ:tsimlb_loose}, we get $\tsimlb(s.\psi) = \frac{0.8}{1-0.6} - \frac{0.6}{1-0.6} = 0.5$, and therefore $\pref(s) = \{w_1,w_2\}$.
Thus, subscription $s$ is only allocated to bolt 1 while message $m$ is still distributed to bolt 3 and bolt 4.
It is obvious that the computation between $s$ and $m$ is avoided since $m$ cannot be the \topk results of $s$.
\end{example}

\begin{figure}[t]
\centering
\includegraphics[width=0.7\columnwidth]{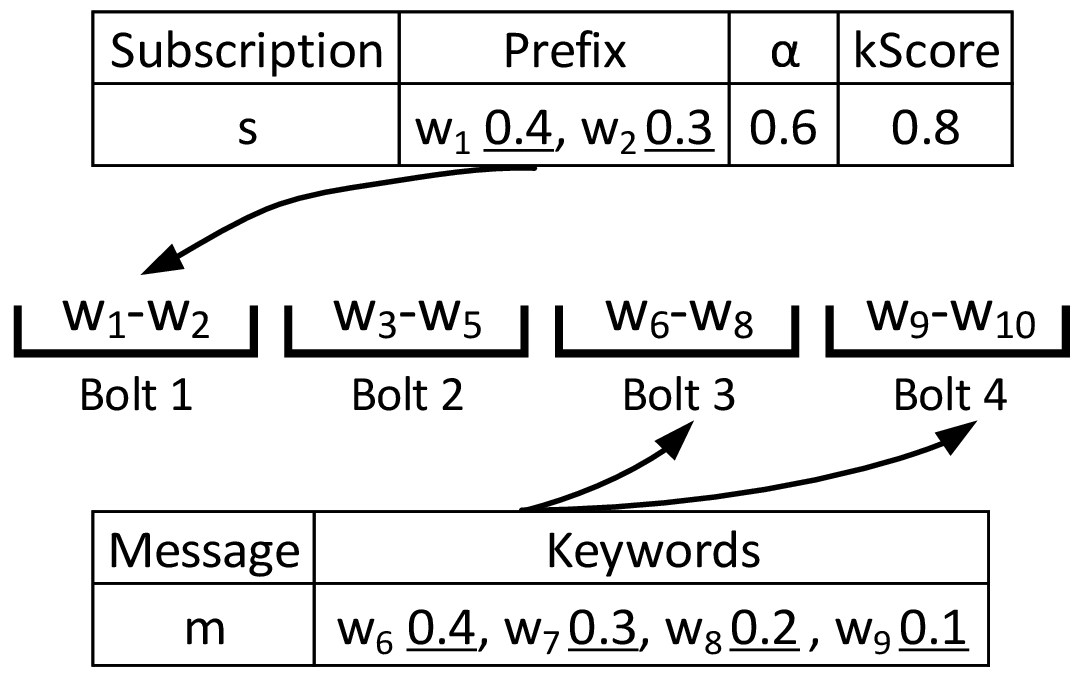}
\caption{\small{Prefix-based example}}
\label{fig:distributed:prefix-based}
\end{figure}

\noindent \textbf{Analysis.}
The replication ratio is determined by $|\pref(s)|$ based on Lemma~\ref{lemma:distirbute_prefix}, which is usually smaller than $|s.\psi|$.
Meanwhile, the average communication cost is bounded by $\min(|m.\psi|, N_{sb})$, which is the same as keyword-based method.
As shown in the experiments, the prefix-based method not only can reduce the replication ratio but also can improve the system throughput with a large margin compared to keyword-based method.
Note that Algorithm~\ref{alg:keyword_partition} is also used here to get keyword partitions, except that the value of $N(W_i)$ in Equation~\ref{equ:keyword_set_cost} needs to be recomputed since we only use the loose prefix of subscription.
Besides, a subscription may need to be reallocated among subscription bolts when $\pref(s)$ changes due to the updating of $\kscore(s)$.
We employ a lazy reallocation strategy, where the reallocation is triggered only when $\pref(s)$ needs to cover more keywords.

\noindent\textbf{Remark.}
We remark that all the distribution mechanisms discussed above are light-weighted indexes employed in the distribution bolts in order to facilitate the distribution of subscriptions and messages.
This is different from the subscription index built in each subscription bolt, which aims to accelerate \topk dissemination regarding each incoming message.

\noindent \textbf{Discussions.}
We summary all the distribution mechanisms discussed above in Table~\ref{tb:distribution_mechanisms_summary}, where we report the replication ratio of subscriptions and average communication cost w.r.t. each message, both of which are the main factors of system throughput.
From the table, we notice that
when $|m.\psi| < N_{sb}$, both keyword-based and prefix-based methods are likely to benefit a lot from the large reduction in communication cost, and are expected to have higher throughput;
however, when $|m.\psi| > N_{sb}$, both keyword-based and prefix-based methods may not perform very well because they cannot reduce communication cost by a large margin, while suffered from the extra cost triggered by replicates.
In Section~\ref{sec:exp:distributed}, we have conducted detailed experiments to evaluate the performances of different algorithms under different settings, which further verify our analysis above.
Furthermore, workload balance is also a main factor contributing to the system throughput.
The hashing-based method can achieve best balance due to the nature of uniform hashing.
However, for all the other three methods, we observe that the workload balance can also be well maintained, since we take into consideration both the subscription workload and message workload during distribution.

\begin{table}[t]
  	\caption{Summary of distribution mechanisms (worst case)}\label{tb:distribution_mechanisms_summary}
  	\centering
    \begin{tabular}{|c|c|c|c|}
      \hline
      \textbf{Method} & \textbf{Replica. ratio} & \textbf{Comm. cost} \\ \hline \hline
      Hashing-based & $1$ & $N_{sb}$ \\ \hline
      Location-based & $1$ & $N_{sb}$ \\ \hline
      Keyword-based & $|s.\psi|$ & $\min(|m.\psi|, N_{sb})$ \\ \hline
      Prefix-based & $|\pref(s)|$ & $\min(|m.\psi|, N_{sb})$ \\ \hline
    \end{tabular}
\end{table}

\subsection{Maintenance}
\label{sec:distributed:maintenance}
\storm provides a graphical interface to monitor the workload of each spout/bolt\footnote{http://storm.apache.org/releases/0.10.0/STORM-UI-REST-API.html}.
We fork a background process to access the interface and monitor the workload of each component automatically and periodically.
When some component becomes overloaded, we may simply increase its parallelism, while decreasing its parallelism if it becomes idle.
The subscription bolts discussed in the above distribution mechanisms can be easily further partitioned or merged according to the system workload.
For example, for the location-based method, we can simply split a leaf node if it becomes overloaded.

\section{Experiments}
\label{sec:exp}

\subsection{Centralized Evaluations}
\label{sec:exp:single}

In this section, we conduct extensive experiments to verify the efficiency and effectiveness of \ours in a single machine.
All experiments are implemented in \verb!C++!, and conducted on a PC with 3.4GHz Intel Xeon 2 cores CPU and 32GB memory running Red Hat Linux.
Following the typical setting of \pubsub systems (e.g.,~\cite{DBLP:conf/icde/HuLLFT15,DBLP:conf/icde/ChenCCT15}), we assume indexes fit in main memory to support real-time response.

\begin{table}[t]
  	\caption{Datasets Statistics}\label{tb:statistics}
  	\centering
    \begin{tabular}{|c|c|c|c|}
      \hline
      \textbf{Datasets} & \textbf{TWEETS} & \textbf{GN} & \textbf{YELP} \\ \hline \hline
      \# of msg & 12.7M & 2.2M & 1.6M \\ \hline
      Vocabulary size & 1.7M & 208K & 85K \\ \hline
      Avg. \# of msg keywords & 9 & 7 & 37 \\ \hline
      Size in GB & 2.26 & 0.3 & 1.04 \\ \hline
    \end{tabular}
\end{table}

\begin{figure}[t]
\centering
	\begin{minipage}[b]{\linewidth}
		\centering
		\includegraphics[width=0.96\columnwidth]{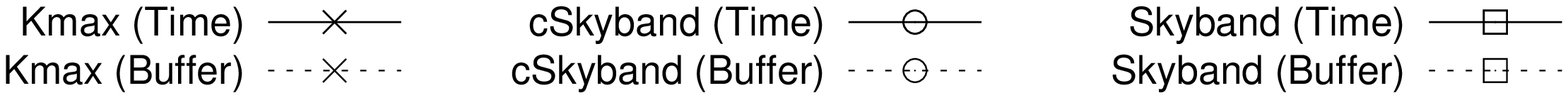}%
	 \end{minipage}	
	 \vfill
	\begin{minipage}[t]{\linewidth}	
		\centering
		\subfigure[\kmax]{
		    \label{fig:exp2:tuning_tweets_kmax_tmp} 
		    \includegraphics[width=0.48\columnwidth]{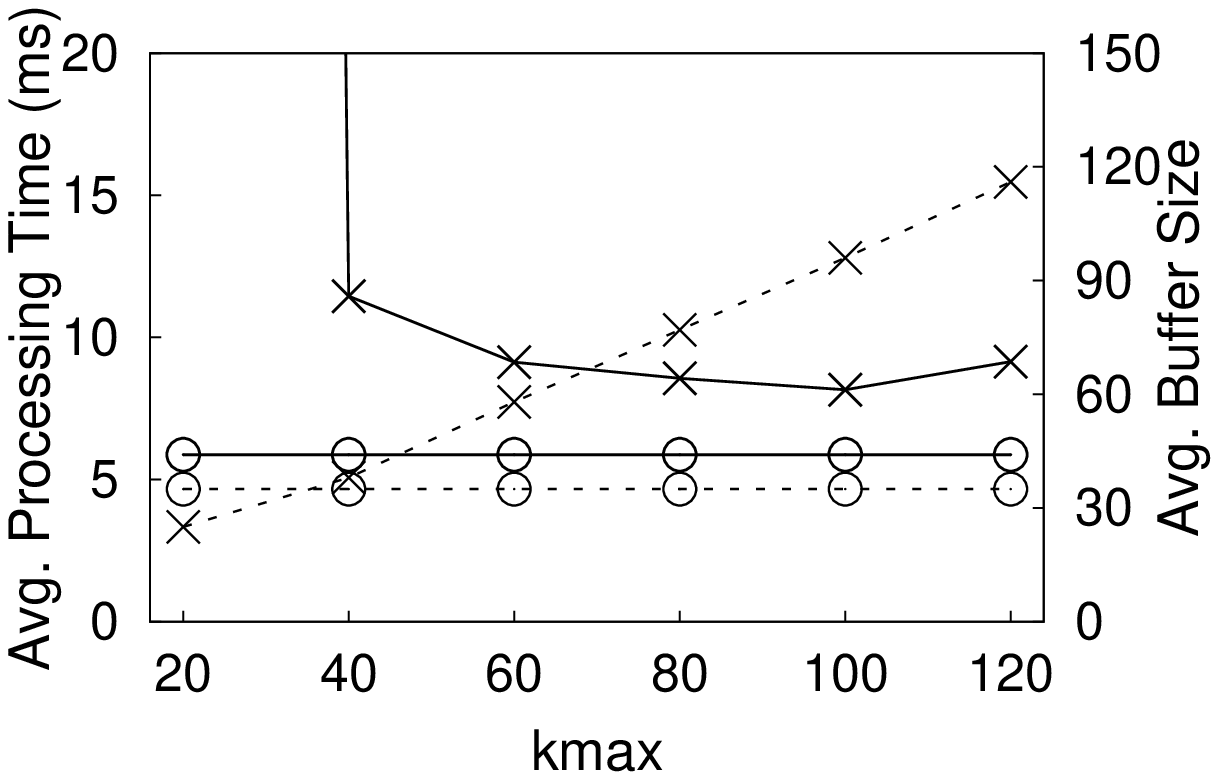}}
		\hfill
		\subfigure[\naive]{
		    \label{fig:exp2:tuning_tweets_skyband_tmp} 
		    \includegraphics[width=0.48\columnwidth]{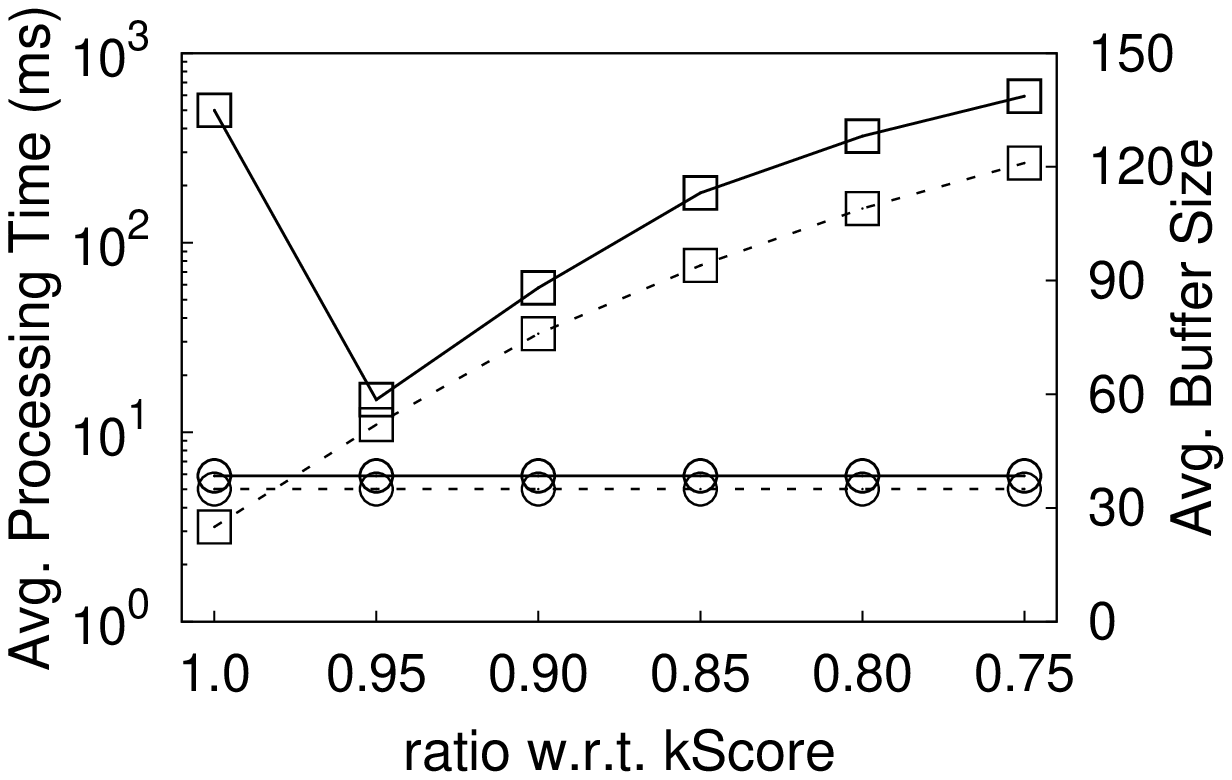}}
		\caption{Tuning baseline algorithms}	
		\label{fig:exp2:tuning}
	\end{minipage}%
\end{figure}

\begin{figure}[t]
\centering
	\begin{minipage}[t]{0.5\linewidth}	
		\centering
		\includegraphics[width=0.96\columnwidth]{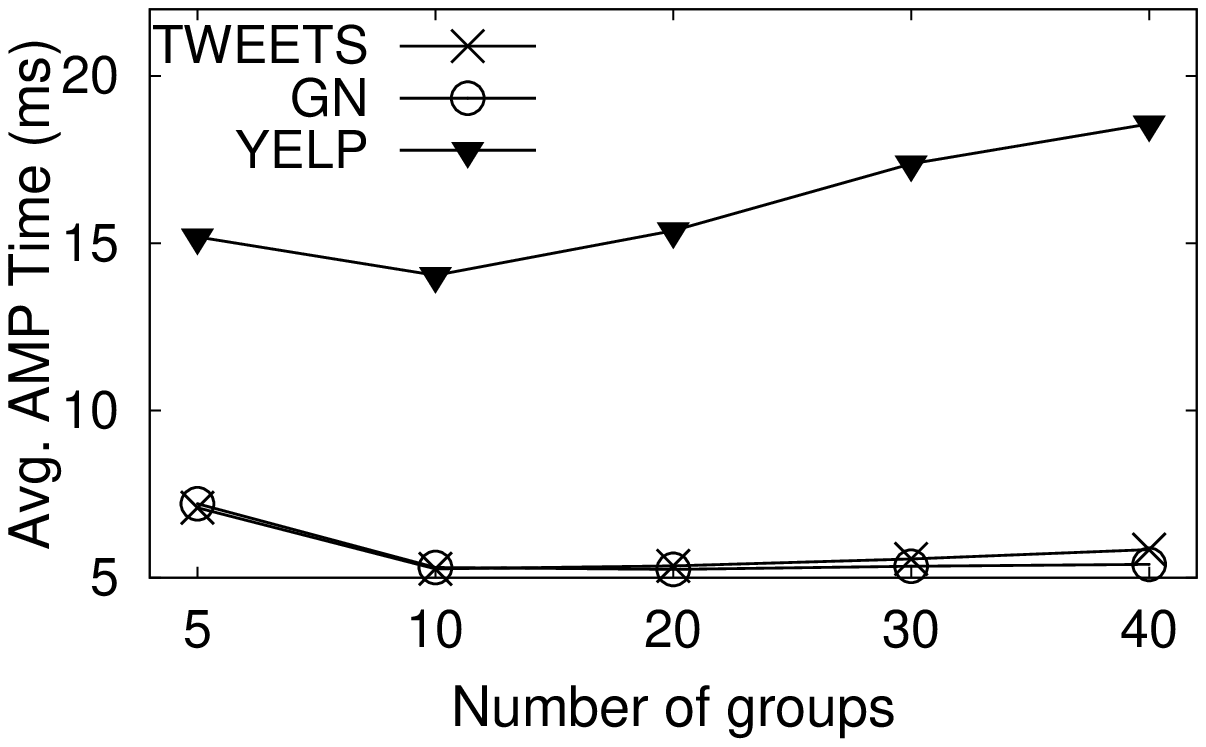}
		\caption{\small{Vary \# groups}}
		\label{fig:exp1:tuning_group}	
	\end{minipage}%
	\hfill
	\begin{minipage}[t]{0.5\linewidth}	
		\centering
		\includegraphics[width=0.96\columnwidth]{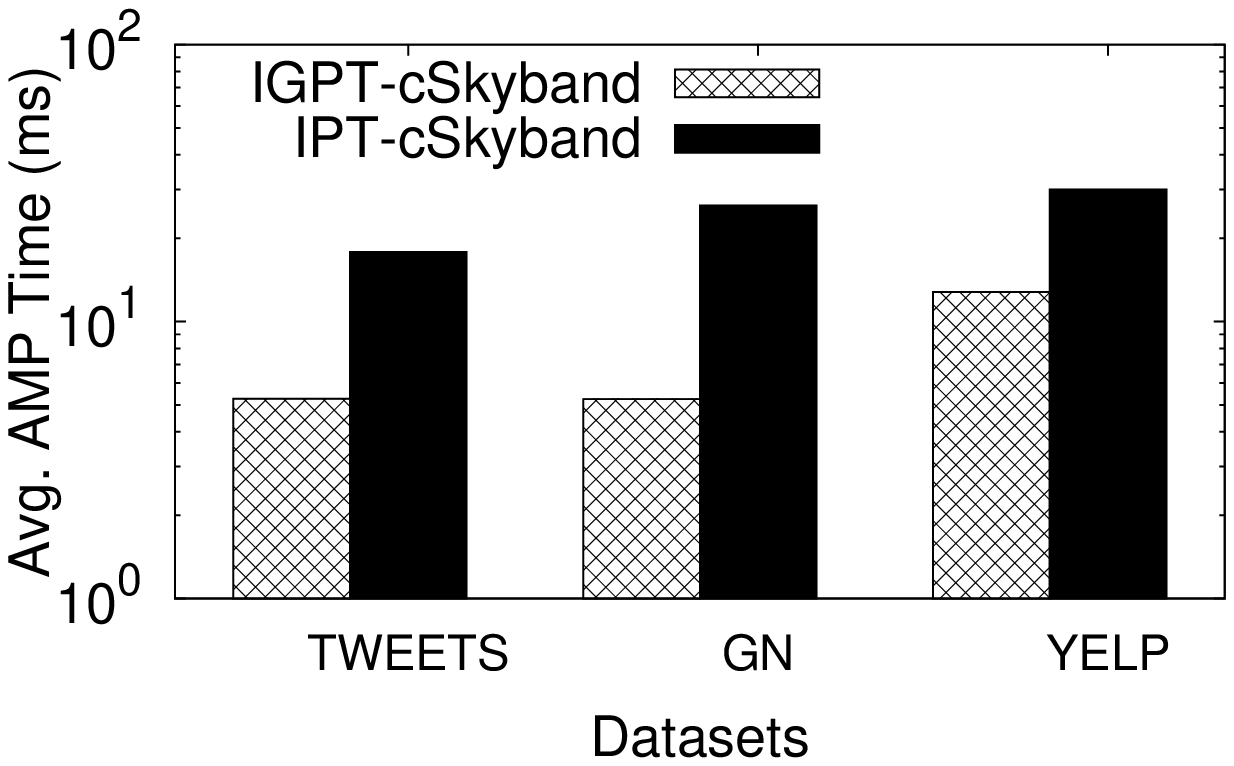}
		\caption{\small{Pruning methods}}
		\label{fig:exp1:compare_pruning_techniques}	
	\end{minipage}%
\end{figure}

\subsubsection{Experimental Setup}
\label{sec:exp:single:setup}
As this is the first work to study \topk \sk publish\slash subscribe over sliding window,
we extend previous work~\cite{DBLP:conf/icde/ChenCCT15} to sliding window.
We implement and compare the following algorithms.
For \moduleone:
\begin{itemize}
\item \textbf{\ciq.} The subscription index proposed in~\cite{DBLP:conf/icde/ChenCCT15}
\footnote{The time decay function and related index in \ciq are removed to adapt to our problem.}.
\item \textbf{\igpt.} The subscription index proposed in our paper, which combines both \textbf{I}ndividual and \textbf{G}roup \textbf{P}runing \textbf{T}echniques.
\end{itemize}
For \moduletwo:
\begin{itemize}
\item \textbf{\naive.} The \skyband algorithm proposed in~\cite{DBLP:conf/sigmod/MouratidisBP06}.
\item \textbf{\kmax.} The \kmax algorithm proposed in~\cite{DBLP:conf/icde/YiYYXC03}.
\item \textbf{\csky.} The cost-based \skyband algorithm proposed in our paper.
\end{itemize}
Note that our \ours algorithm is the combination of \igpt and \csky.
We apply \irtree~\cite{DBLP:journals/pvldb/CongJW09} to index messages.

\noindent \textbf{Datasets.}
Three datasets are deployed for experimental evaluations.
\tweets is a real-life dataset collected from Twitter~\cite{DBLP:conf/kdd/LiWWF13}, containing $12.7M$ tweets with geo-locations from 2008 to 2011.
\gn is obtained from the US Board on Geographic Names\footnote{http://geonames.usgs.gov} in which each message is associated with a geo-location and a short text description.
\yelp is obtained from Yelp\footnote{http://www.yelp.com/}, which contains user reviews and check-ins for thousands of businesses.
The statistics of three datasets are summarized in Table~\ref{tb:statistics}.

\noindent \textbf{Subscription workload.}
We generate \topk subscriptions based on the above datasets.
For each dataset, $1M$ \gt messages are randomly selected.
For each selected message, we randomly pick $j$ keywords as subscription keywords with $1 \leq$ $j \leq 5$.
The weight of each keyword is computed according to \textit{tf-idf} weighting scheme\footnote{https://en.wikipedia.org/wiki/Tf–idf}.
The subscription location is the same as message location.
For each subscription, the preference parameter $\alpha$ is randomly selected between $0$ and $1$,
while the default value of $k$, i.e., the number of \topk results, is set to 20.

\noindent \textbf{Message workload.}
Our simulation starts when the sliding window with default size of $1M$ is full and continuously runs for $100,000$ arriving messages over the sliding window.

We report the average processing time, including average arriving message processing time (i.e.,~\textbf{\amp}) and average expired message processing time (i.e.,~\textbf{\emp}), as well as the index size.
By default, the number of $\alpha$-partition groups is set to 10.
The maximum number of subscriptions that can be stored in each cell is set to 1000.

\subsubsection{Experimental Tuning}
\label{sec:exp:single:tuning}

\noindent \textbf{Tuning \kmax and \naive.}
In the first set of experiments, we tune the performance of both \kmax and \naive techniques in Figure~\ref{fig:exp2:tuning} on \tweets dataset, where \igpt algorithm is employed for message dissemination.
For better understanding, we evaluate average processing time (denoted as solid line) and average buffer size (denoted as dotted line) in the same figure.
We also show the results of our \igptcsky algorithm under default settings which remains unchanged.
For \kmax algorithm (Figure~\ref{fig:exp2:tuning_tweets_kmax_tmp}), we vary \kmax from $20$ to $120$.
It is noticed that a small \kmax leads to high \topk re-evaluation cost while a large \kmax results in high message dissemination cost and buffer maintenance cost.
We set 60 as the default \kmax value since it strikes a good trade-off between performance and buffer size (i.e., memory cost).
For \naive algorithm, we vary the threshold score $s.\theta$ from $1.0\times$\kscore~to $0.75\times$\kscore~where the smaller value leads to larger buffer size.
It is noticed that when the ratio is $1.0$ which is the original setting in~\cite{DBLP:conf/sigmod/MouratidisBP06}, the performance of \naive is poor due to the frequent re-evaluations.
For comparison fairness, $s.\theta$ is set to its sweet point $0.95 \times$\kscore~for \naive in the following experiments.
It is worth mentioning that our \csky always outperforms \naive under all settings because \csky can tune a best threshold for each individual subscription based on the cost model while there is no sensible way to tune \naive for millions of subscriptions.
The similar trends are also observed on other datasets.

\begin{figure}[t]
	\centering
	\begin{minipage}[t]{\linewidth}	
		\centering
		\subfigure[Arriving msg processing]{
	    	\label{fig:exp3:compare_disseminate_amp} 
	    	\includegraphics[width=0.47\columnwidth]{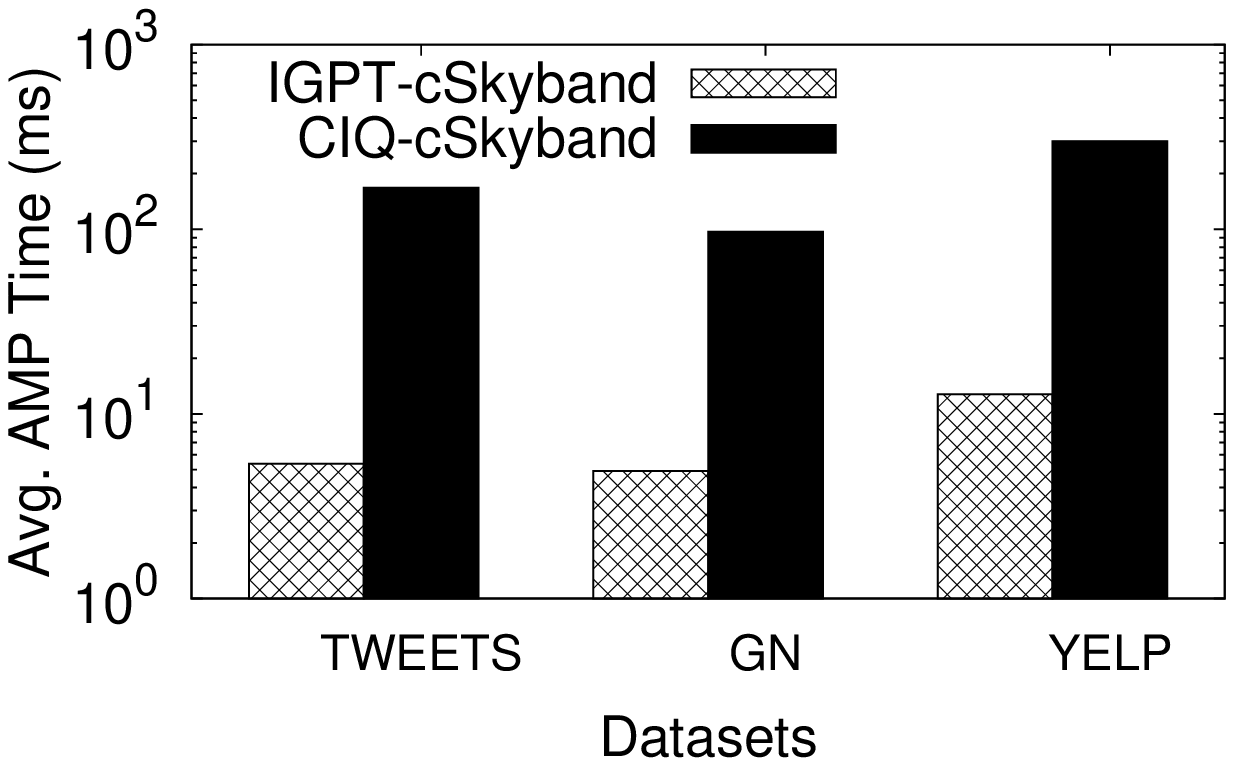}}
	    \hfill
		\subfigure[Memory cost]{
		    \label{fig:exp3:compare_disseminate_mem} 
		    \includegraphics[width=0.47\columnwidth]{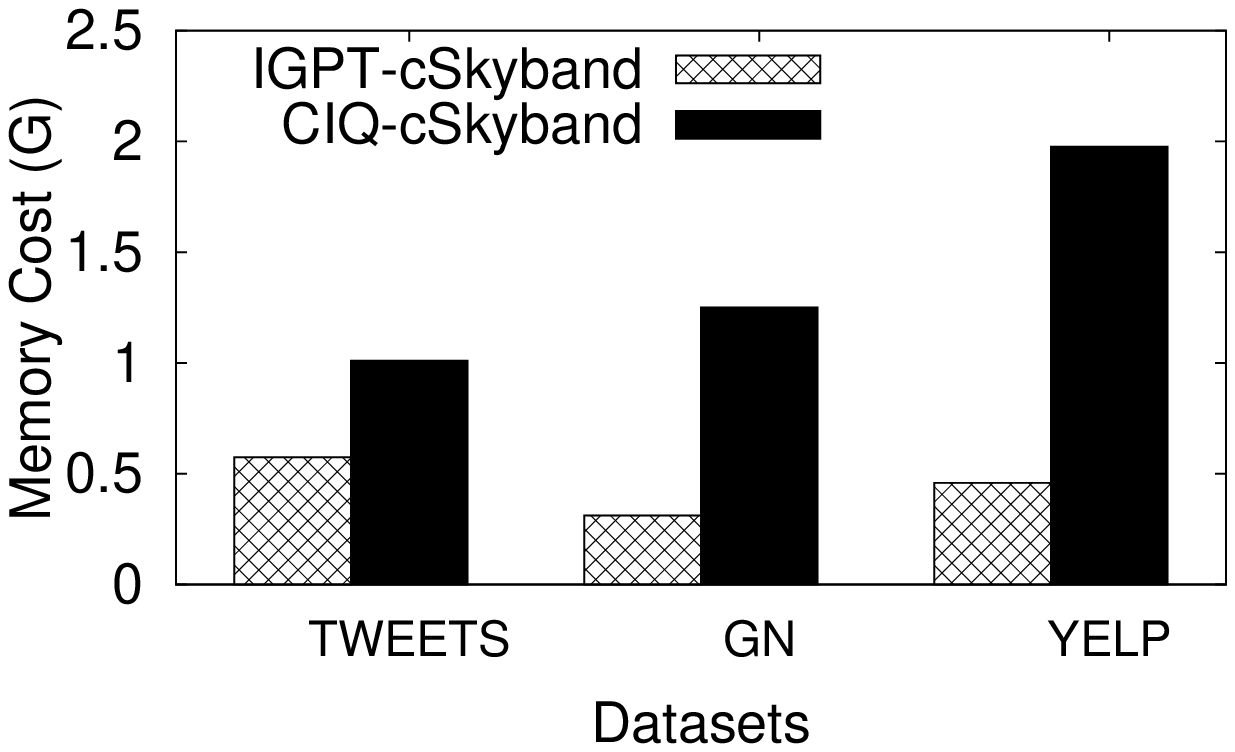}}
		\caption{Compare dissemination algorithms}	
		\label{fig:exp3:compare_disseminate}
	\end{minipage}%
\end{figure}

\begin{figure}[t]
	\centering
	\begin{minipage}[t]{\linewidth}	
		\centering
		\includegraphics[width=0.7\columnwidth]{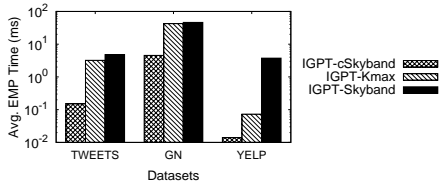}
		\caption{Compare re-evaluation algorithms}	
		\label{fig:exp4:compare_refill_emp}	
	\end{minipage}
\end{figure}

\noindent \textbf{Vary the number of groups in $\alpha$-partition.}
Figure~\ref{fig:exp1:tuning_group} reports the \amp time of \ours algorithm against three datasets where the number of groups varying from 5 to 40. It is shown that we can achieve a good trade-off between the group filtering effectiveness and group checking costs
when the number of groups is set to 10, which is used as default value in the following experiments.

\noindent \textbf{Effect of pruning techniques.}
In this experiment, we compare the \amp time of different pruning techniques employed in our \moduleone.
Specifically, we compare \igpt with \ipt, which only employs individual pruning technique in Figure~\ref{fig:exp1:compare_pruning_techniques}.
We observe that \igpt algorithm can achieve at least three times improvement compared with \ipt over all the datasets, which verifies the efficiency of our group pruning techniques.
This is mainly because the group pruning technique can skip the whole group without the need to check individual subscription, and can terminate early within a group.
In the following experiments, we always use \igpt as our dissemination algorithm.

\begin{table}[t]
  \caption{Average buffer size of different re-evaluation algorithms (Unit: 1)}
  \label{tb:exp4:buffer_size}
  \centering
    \begin{tabular}{|c|c|c|c|c}
      \hline
      \textbf{Algorithm} & \tweets & \gn & \yelp  \\ \hline \hline
      \igptcsky & \textbf{35} & \textbf{33} & \textbf{28}  \\ \hline
      \igptkmax & 58  & 58 & 59 \\ \hline
      \igptnaive & 52  & 58 & 56 \\ \hline
    \end{tabular}
\end{table}

\subsubsection{Performance Evaluation}

%

\begin{figure*}[t]
	\centering
	\begin{minipage}[b]{1.0\linewidth}
		\centering
		\includegraphics[width=0.45\columnwidth]{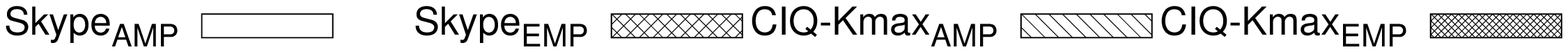}%
	\end{minipage}	
	\vfill
	\begin{minipage}[t]{0.495\linewidth}	
		\centering
		\subfigure[\tweets]{
		    \label{fig:exp6:tweets_knum_tmp} 
		    \includegraphics[width=0.48\columnwidth]{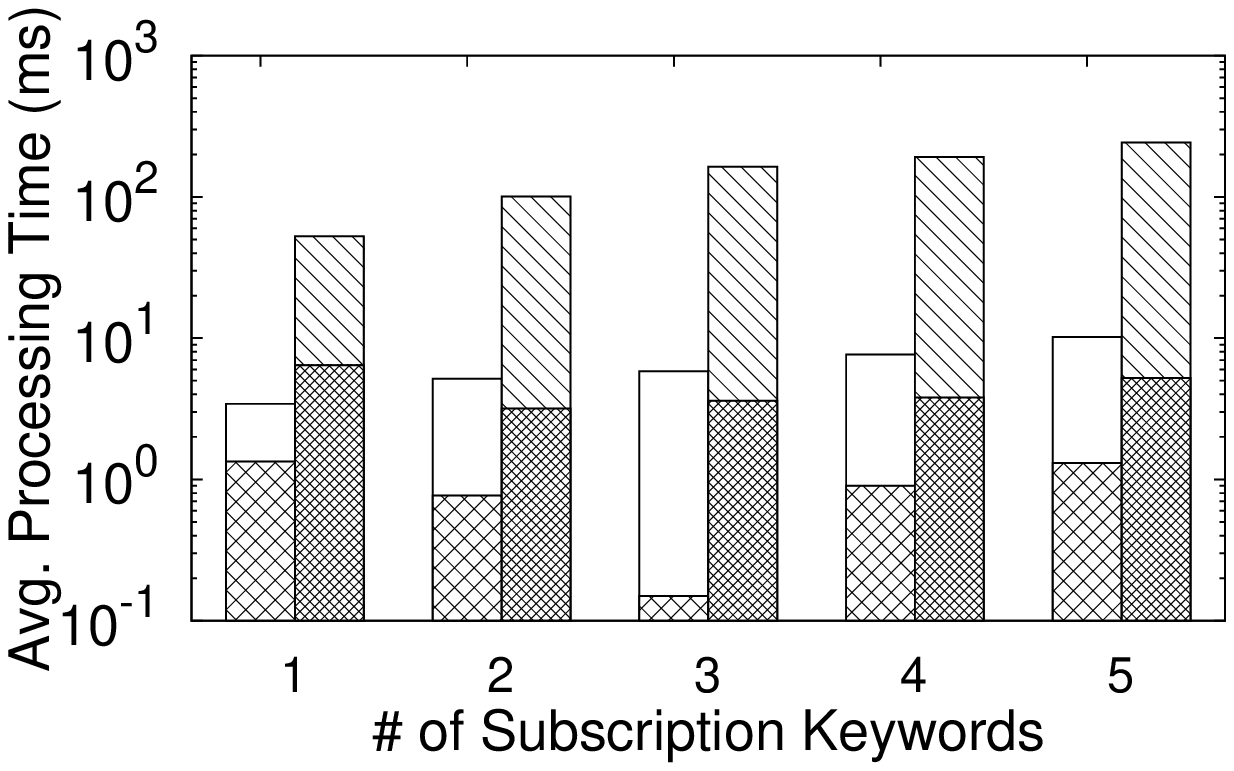}}
		\hfill
		\subfigure[\gn]{
		    \label{fig:exp6:gn_knum_tmp} 
		    \includegraphics[width=0.48\columnwidth]{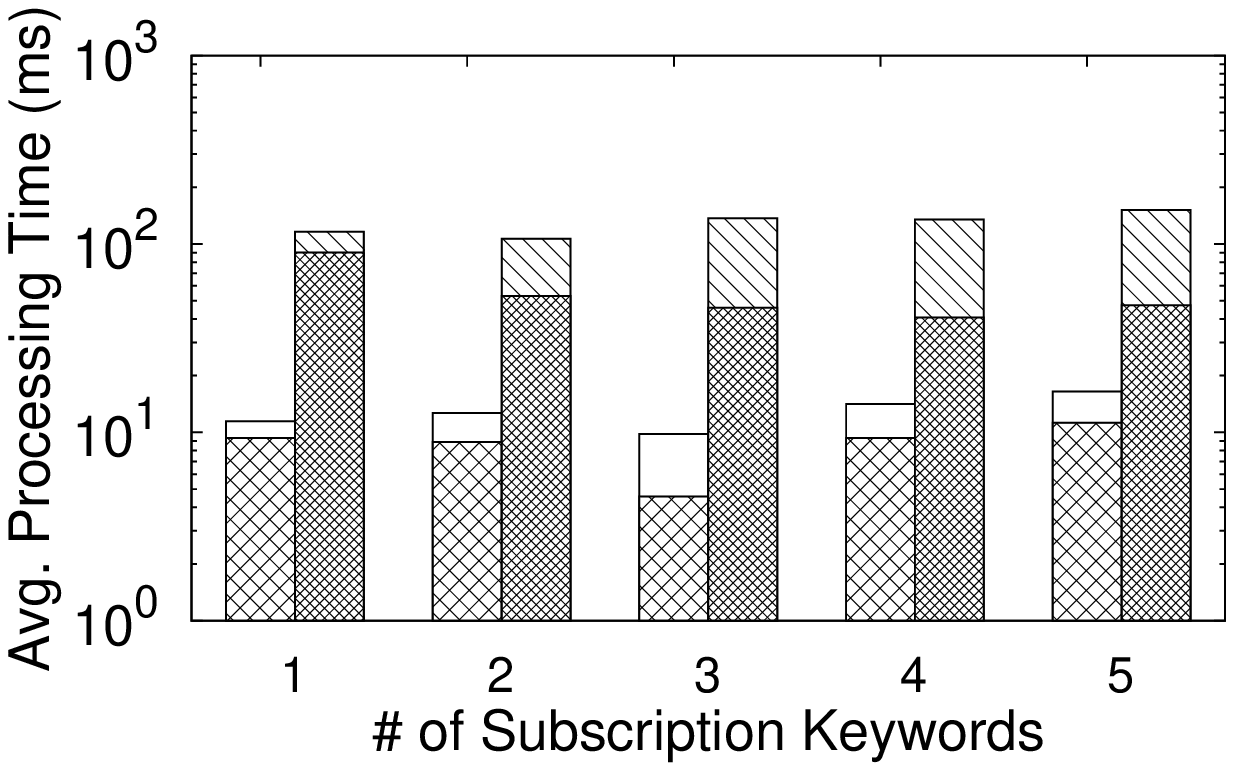}}
		\caption{Effect of number of subscription keywords}	
		\label{fig:exp6:knum_tmp}
	\end{minipage}
	\hfill
	\begin{minipage}[t]{0.495\linewidth}	
		\centering
		\subfigure[\tweets]{
		    \label{fig:exp9:tweets_topk_tmp} 
		    \includegraphics[width=0.48\columnwidth]{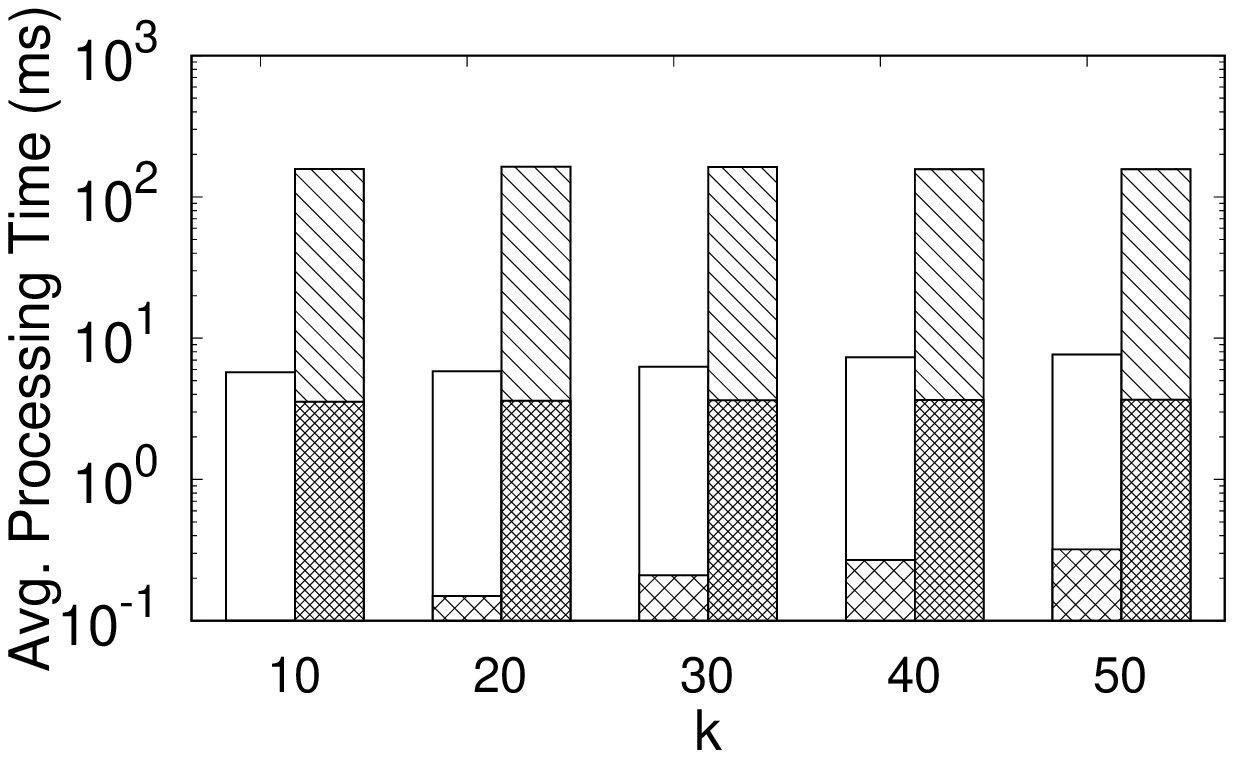}}
		\hfill
		\subfigure[\gn]{
		    \label{fig:exp9:gn_topk_tmp} 
		    \includegraphics[width=0.48\columnwidth]{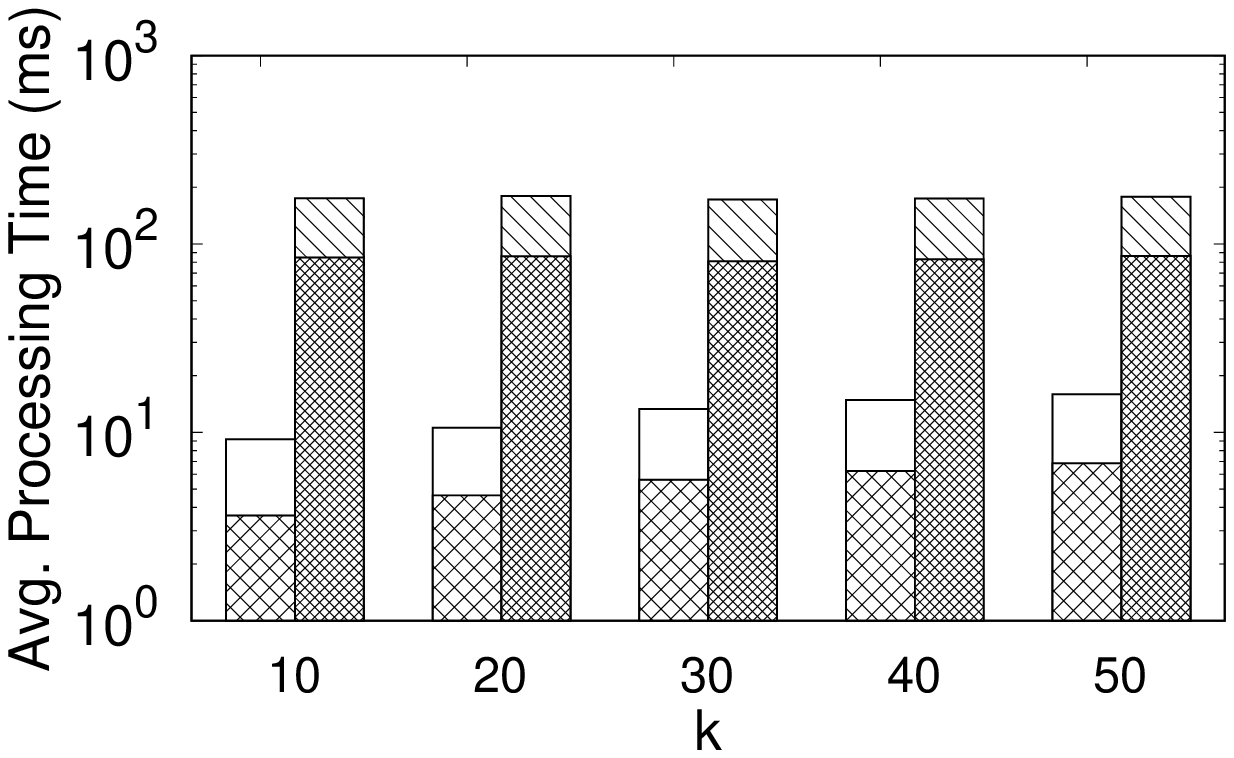}}
		\caption{Effect of number of \topk results}	
		\label{fig:exp9:topk_tmp}
	\end{minipage}%
\end{figure*}

\begin{figure*}[t]
\centering
	\begin{minipage}[b]{1.0\linewidth}
		\centering
		\includegraphics[width=0.45\columnwidth]{exp6_draw_title.eps}%
	 \end{minipage}	
	 \vfill
	\begin{minipage}[t]{0.495\linewidth}	
		\centering
		\subfigure[\tweets]{
		    \label{fig:exp8:tweets_qnum_tmp} 
		    \includegraphics[width=0.48\columnwidth]{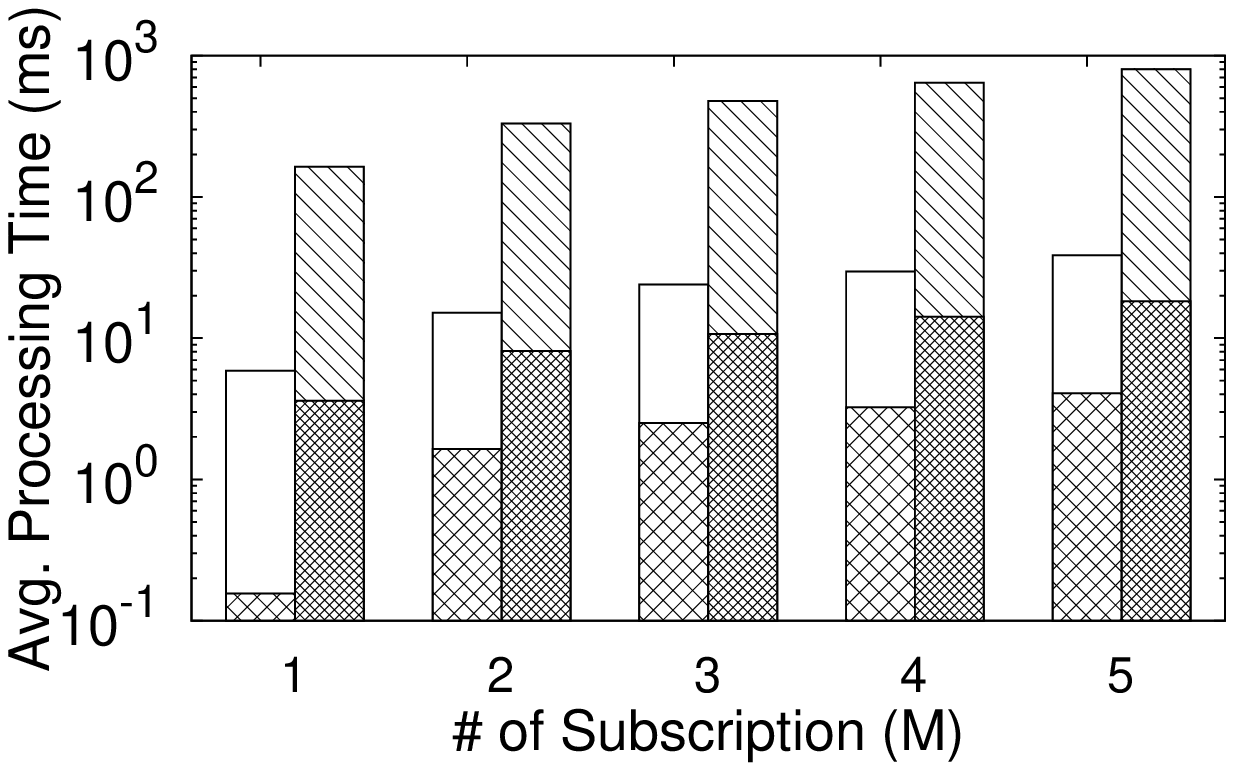}}
		\hfill
		\subfigure[\gn]{
		    \label{fig:exp8:gn_qnum_tmp} 
		    \includegraphics[width=0.48\columnwidth]{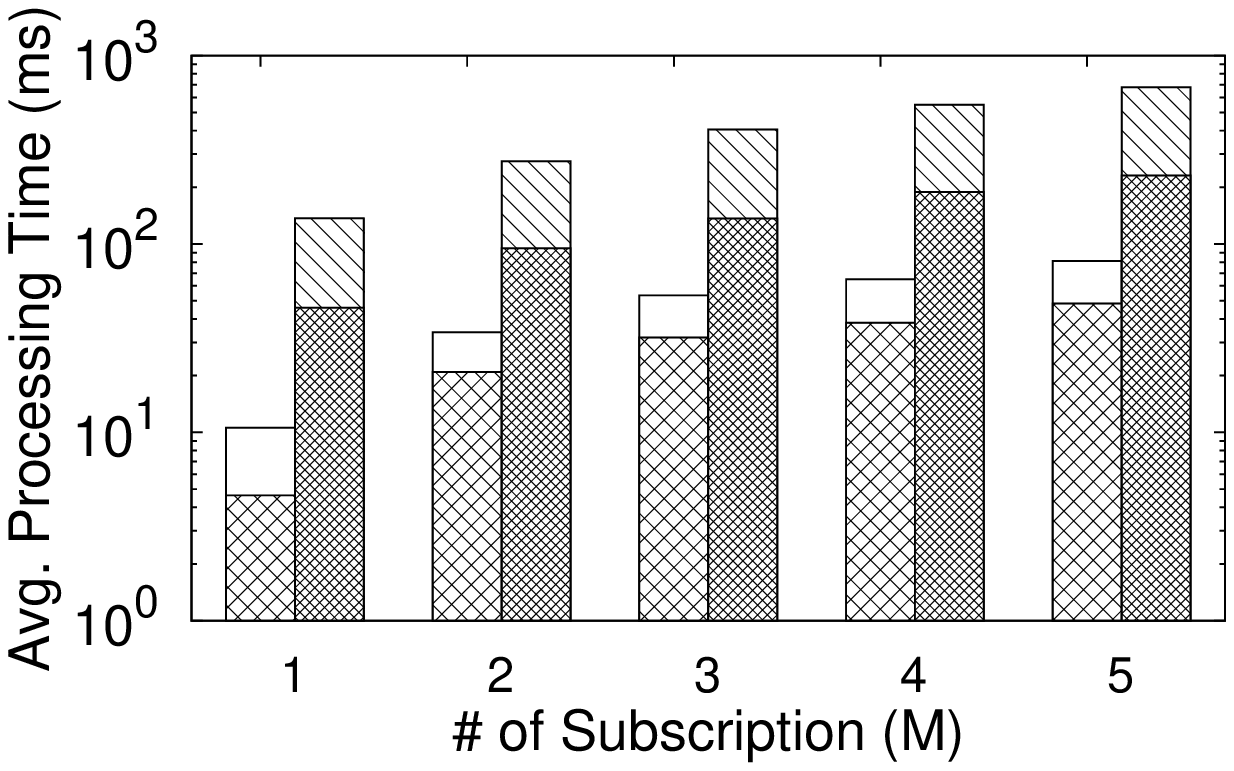}}
		\caption{Effect of number of subscriptions}	
		\label{fig:exp8:qnum_tmp}
	\end{minipage}
	\hfill
	\begin{minipage}[t]{0.495\linewidth}	
		\centering
		\subfigure[\tweets]{
		    \label{fig:exp7:tweets_window_tmp} 
		    \includegraphics[width=0.48\columnwidth]{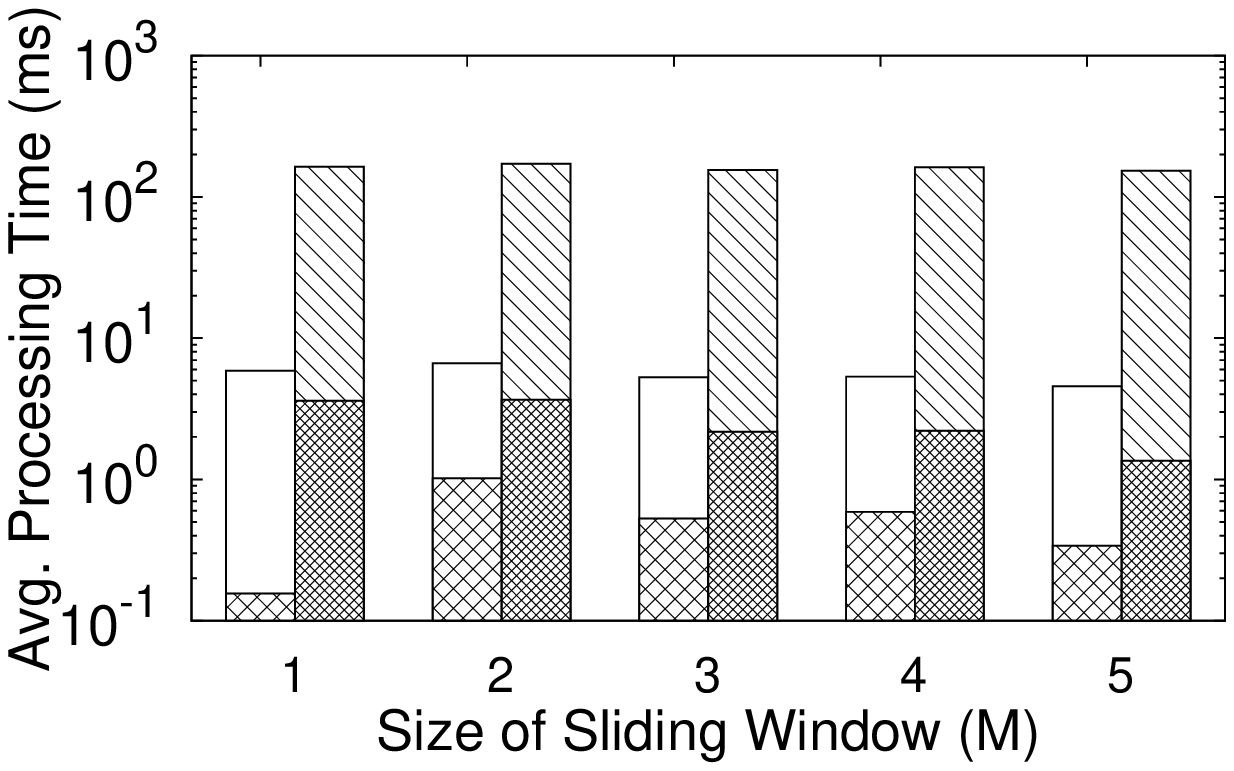}}
		\hfill
		\subfigure[\gn]{
		    \label{fig:exp7:gn_window_tmp} 
		    \includegraphics[width=0.48\columnwidth]{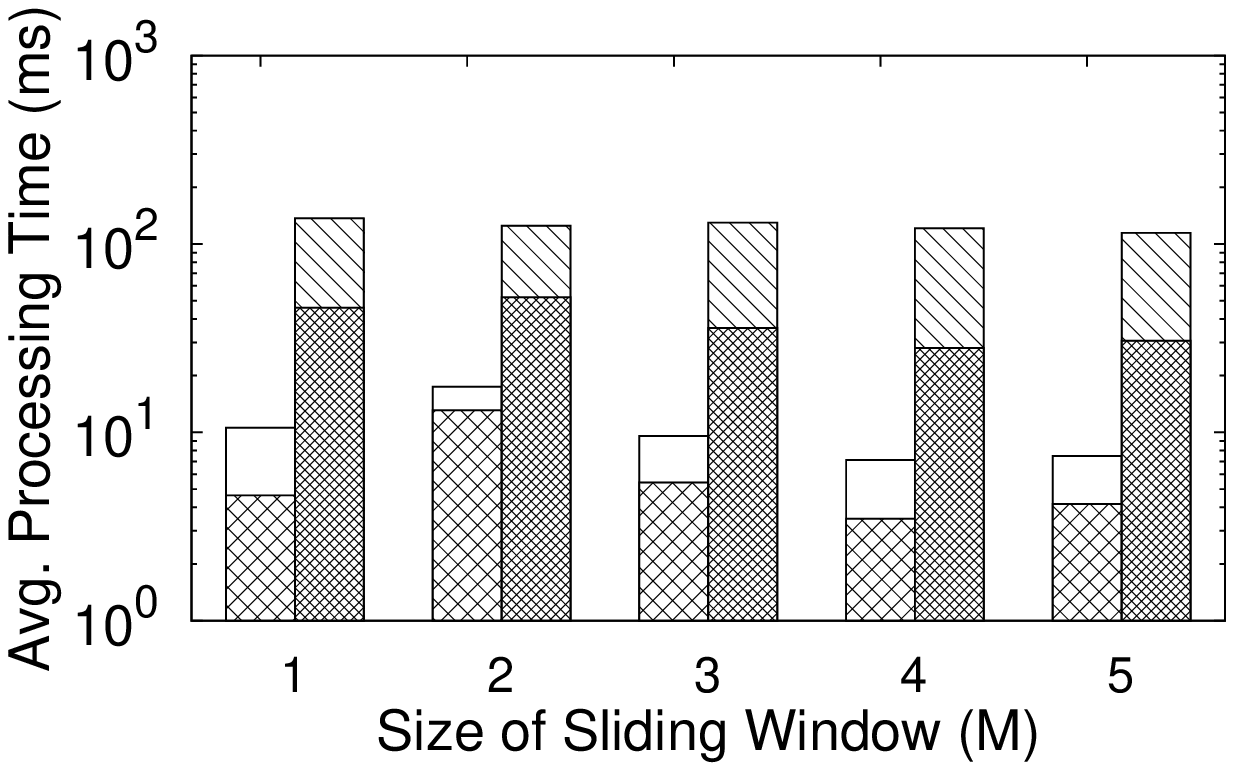}}
		\caption{Effect of sliding window size}	
		\label{fig:exp7:window_tmp}
	\end{minipage}%
\end{figure*}

\begin{figure}[t]
	\centering
	\begin{minipage}[b]{1.0\linewidth}
		\centering
		\includegraphics[width=0.9\columnwidth]{exp6_draw_title.eps}%
	\end{minipage}	
	\vfill
	\begin{minipage}[t]{0.48\linewidth}	
		\centering
		\includegraphics[width=\columnwidth]{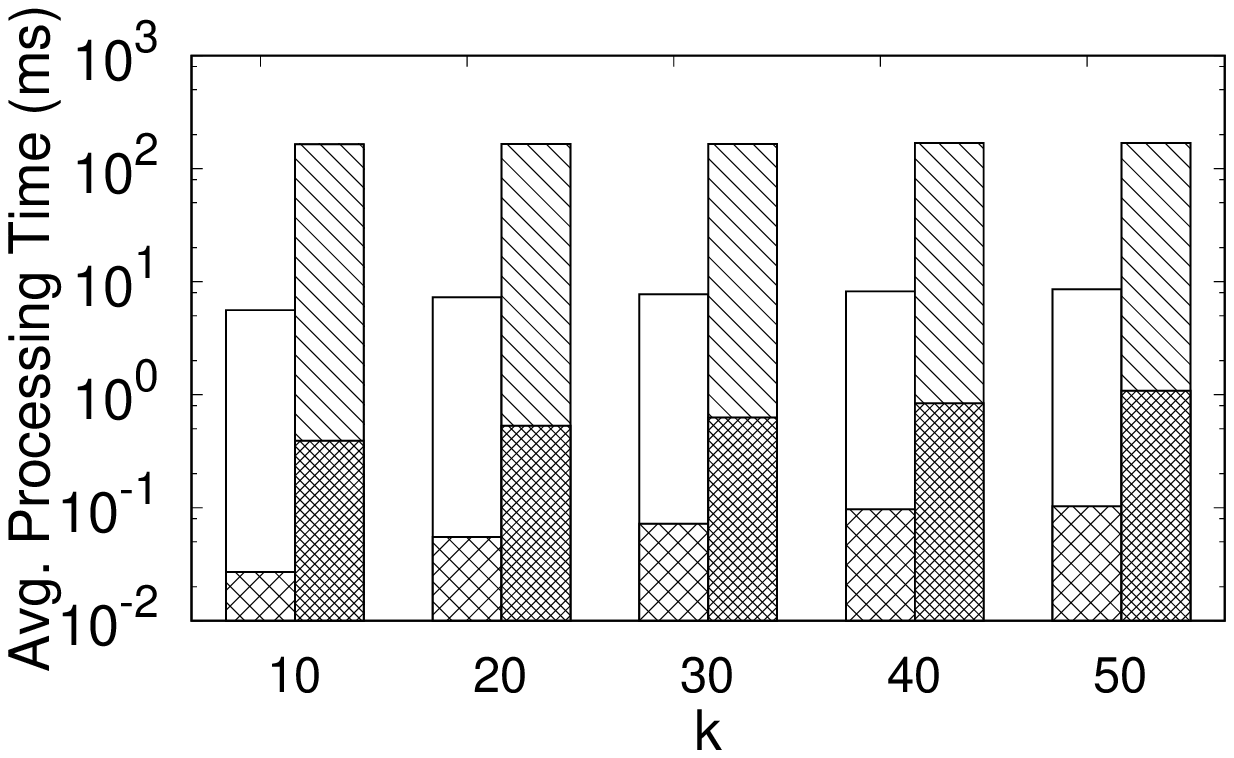}
		\caption{\small{Zipf distribution workload}}
		\label{fig:exp:r1:zipf_distribution}
	\end{minipage}%
	\hfill
	\begin{minipage}[t]{0.48\linewidth}	
		\centering
		\includegraphics[width=\columnwidth]{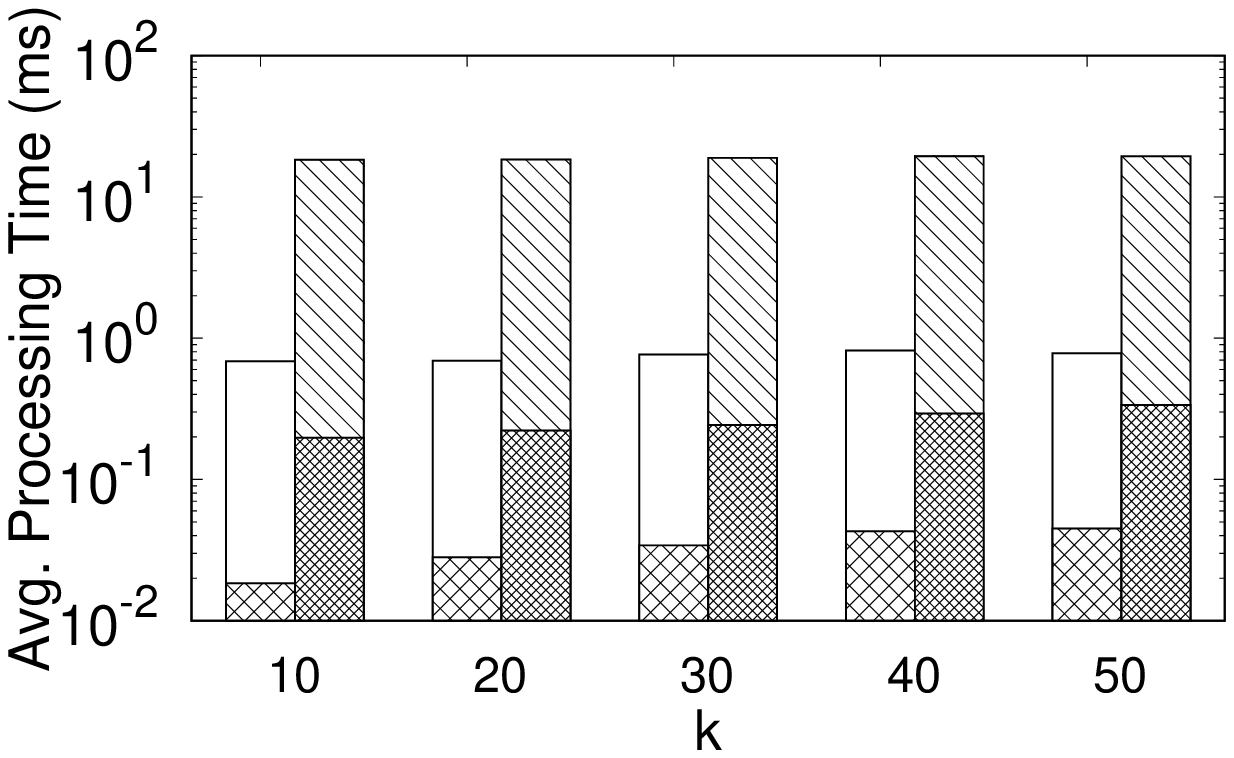}
		\caption{\small{Performance of time-based sliding window}}
		\label{fig:exp:r1:time_sliding_window}
	\end{minipage}
\end{figure}

%

\noindent \textbf{Compare message dissemination algorithms.}
In this experiment, we compare the performance of different dissemination algorithms.
Specifically, we compare \ciq and \igpt with \csky being \topk re-evaluation algorithm.
As shown in Figure~\ref{fig:exp3:compare_disseminate_amp}, our algorithm can achieve about 10 times faster than \ciq algorithm, due to the benefit of individual pruning technique and group pruning technique.
On the other hand, as shown in Figure~\ref{fig:exp3:compare_disseminate_mem},
even if we need to maintain some additional statistics,
the memory cost of our subscription index is much smaller than that of \ciq,
since our algorithm only indexes each subscription into single cell, rather than multiple cells.

\noindent \textbf{Compare \topk re-evaluation algorithms.}
In this experiment, we compare the performance of different \topk re-evaluation strategies combined with our \igpt algorithm.
Specifically, we compare \kmax algorithm~\cite{DBLP:conf/icde/YiYYXC03}, \skyband algorithm~\cite{DBLP:conf/sigmod/MouratidisBP06} and our cost-based \skyband algorithm, which are denoted as \igptkmax, \igptnaive, \igptcsky respectively.
The average \emp time is reported in Figure~\ref{fig:exp4:compare_refill_emp}.
We observe that our \csky algorithm can achieve about 4-20 times improvement compared to the second best algorithm.
This is mainly due to the adaptiveness of our cost model which can tune a best threshold for each subscription.
Table~\ref{tb:exp4:buffer_size} demonstrates the average buffer size of each algorithm.
Our algorithm maintains much fewer number of messages than other competitors due to the advantage of our cost model.
In the following experiments, we only compare our \ours algorithm (i.e., \igptcsky) with \ciqkmax, which performs best among all the baselines.


\noindent \textbf{Effect of number of subscription keywords.}
We assess the effect of number of subscription keywords in Figure~\ref{fig:exp6:knum_tmp}.
We notice that the \amp time increases as we vary the number of keywords from 1 to 5.
This is obvious since more candidates will be encountered during traversing posting lists when the number of keywords is large.
As to the \emp time, we observe that the selectivity is low and fewer messages are \textit{relevant} at initial, thus leading to high cost.
With increasing number of keywords, the selectivity increases, thus reducing the number of re-evaluations accordingly.
Finally, when the number of keyword reaches 4 or 5, a message is less likely to have a high score w.r.t. a subscription due to the smaller weight assigned to each subscription keyword on average, resulting in the increase of \emp time.
The overall processing time increases decently for a large number of keywords.

\noindent \textbf{Effect of number of \topk results.}
In this set of experiments, we analyse the effect of number of \topk results, i.e., $k$, in Figure~\ref{fig:exp9:topk_tmp}.
For \amp time, as we increment $k$ from 10 to 50, the average \kscore~of subscriptions decreases; therefore, an arriving message is more likely to influence more subscriptions, leading to high \amp time in our algorithm.
Meanwhile, a large $k$ usually results in high \emp time, because the subscriptions with low selectivity are more likely to expire and incur \topk re-evaluations.
Besides, the \skyband maintenance cost also increases for large $k$.
Overall, the average processing time increases slowly w.r.t. $k$.

\noindent \textbf{Effect of number of subscriptions.}
We evaluate the scalability of our system in Figure~\ref{fig:exp8:qnum_tmp}, where we vary the number of subscriptions from $1M$ to $5M$.
As shown in the figure, our algorithm scales very well with increasing number of subscriptions, thus making it practical to support real-life applications with fast response.

\noindent \textbf{Effect of sliding window size.}
We turn to evaluate the effect of sliding window size $|\mathcal{W}|$ in this set of experiments.
The results are demonstrated in Figure~\ref{fig:exp7:window_tmp}, where we vary $|\mathcal{W}|$ from $1M$ to $5M$.
It is observed that, when we increase $|\mathcal{W}|$, the \amp time decreases, which is due to the fact that a large sliding window usually leads to better \topk results with higher \kscore.
Thus, a new message will affect less subscriptions, resulting in lower \amp cost.
Regarding the \emp time, it fluctuates around a value due to the competitive results of fewer number of re-evaluations and high query cost against the message index as we increase $|\mathcal{W}|$.


\noindent \textbf{Performance over \zipf distribution.}
We evaluate the performance over a different subscription workload where the keywords are sampled from a \zipf distribution (Figure~\ref{fig:exp:r1:zipf_distribution}).
It is observed that the \zipf distribution workload has similar performance compared to original workload, where keywords are randomly sampled from messages.

\noindent \textbf{Performance over time-based sliding window.}
We verify the performance over time-based sliding window against \tweets dataset.
We set the sliding window size as 4 months and feed the initial sliding window with tweets from January to April 2010.
The tweets from recent one month are used to estimate window size and probability discussed in Section~\ref{sec:refill:discussions}.
The \amp and \emp time w.r.t. \textit{each timestamp} (i.e., $1$ sec) are reported by continuously feeding the sliding window with tweets collected in May 2010.
The results are shown in Figure~\ref{fig:exp:r1:time_sliding_window} where we vary the number of \topk results.
It is noticed that \ours can still achieve an order of magnitude improvement compared to \ciqkmax, which verifies the efficiency of our extensions.

\subsection{Distributed Evaluations}
\label{sec:exp:distributed}

In this section, we verify the performance of our distributed \pubsub system, i.e.,~\oursdis.
All the experiments are conducted on a 10-node (one nimbus and nine supervisors) cluster running Storm 0.10.0\footnote{http://storm.apache.org/2015/11/05/storm0100-released.html}, with a single node Zookeeper server\footnote{http://zookeeper.apache.org/doc/r3.4.8/} deployed for coordination.
Each node in the cluser is a Debian 6.0.10 server that has 3.4GHz Intel Xeon 8 cores CPU, 16GB memory, and gigabit ethernet interconnect.
Each supervisor node is configured to run at most 3 workers at the same time, and each worker can run multiple spouts/bolts concurrently.

We use the same \tweets, \gn and \yelp datasets as above and generate subscription workload and message workload accordingly.
The default number of subscriptions is $5M$, and the default size of sliding window is $1M$.
The message workload is fed to the system continuously for one hour.
The number of subscription spouts, message spouts, distribution bolts, subscription bolts, message bolts and aggregation bolts is set to 1, 1, 1, 32, 3 and 1 respectively by default.
It is noticed that the number of subscription bolts is the largest compared to other components, since subscription bolts are the main bottleneck of our system.
The parameters in the subscription index are tuned in a similar way to Section~\ref{sec:exp:single:tuning}.
We use the \textbf{throughput}, i.e., the average number of messages processed in one second, and the \textbf{communication cost}, i.e., the average number of tuples transmitted between distribution bolts and subscription bolts to process one message, as the main measurements.
Note that, since the real communication cost heavily depends on the hardware, we use the number of tuples transmitted as the measure of communication cost, which is hardware-independent.
Besides, we only consider the communication cost between distribution bolts and subscription bolts because it dominates all the other communication costs.
All the measurements are computed after system initialization, which usually takes about 2 minutes.

\begin{figure}[t]
	\centering
	\begin{minipage}[b]{\linewidth}
		\centering
		\includegraphics[width=0.95\columnwidth]{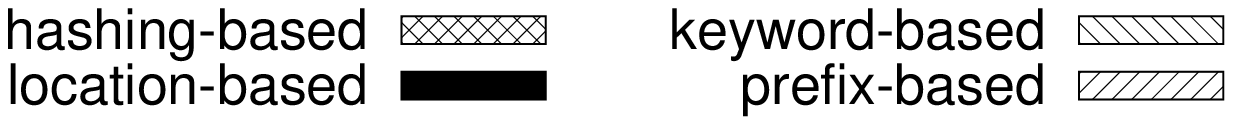}%
	 \end{minipage}	
	 \vfill
	\begin{minipage}[t]{\linewidth}	
		\centering
		\subfigure[Throughput]{
	    	\label{fig:journal:exp1:dis_methods_throughput} 
	    	\includegraphics[width=0.48\columnwidth]{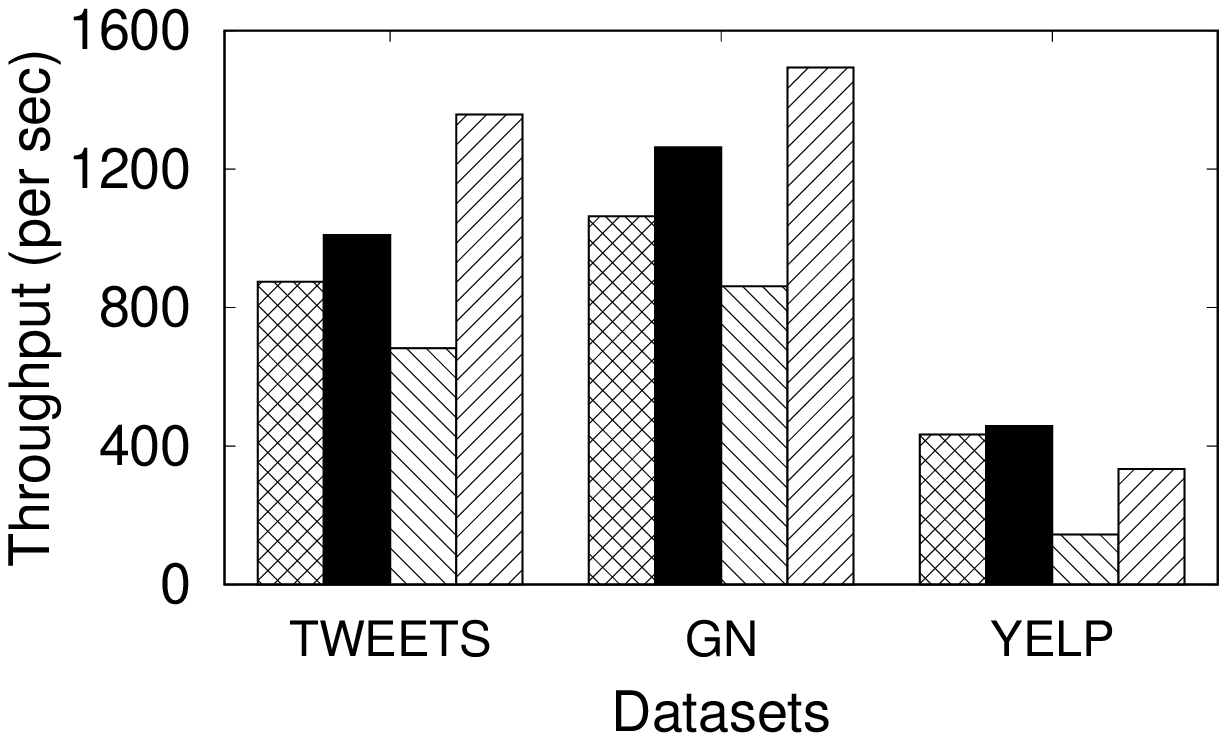}}
	    \hfill
		\subfigure[Avg. communication cost]{
		    \label{fig:journal:exp1:dis_methods_comm} 
		    \includegraphics[width=0.48\columnwidth]{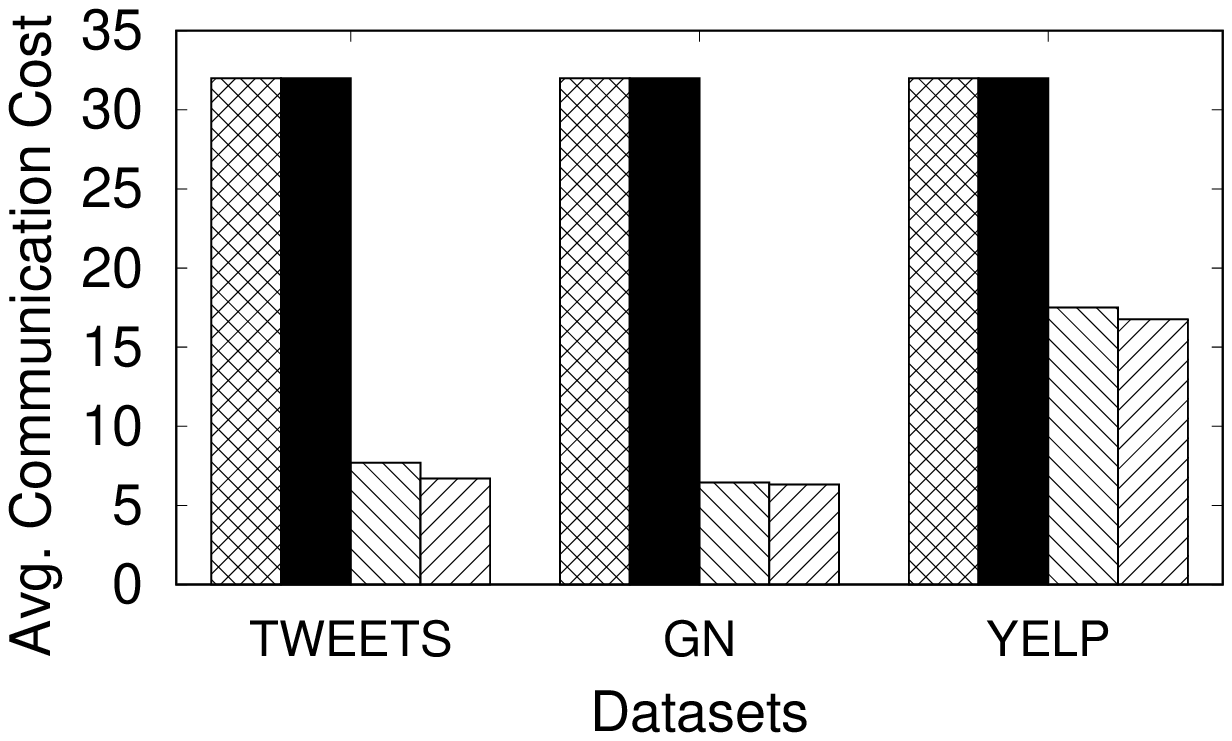}}
		\caption{Compare different distribution mechanisms}	
		\label{fig:journal:exp1:dis_methods}
	\end{minipage}%
\end{figure}

\begin{table}[t]
  	\caption{Replication ratio}\label{tb:exp_dis:replicate_ratio}
  	\centering
    \begin{tabular}{|c|c|c|c|}
      \hline
      \textbf{Methods} & \textbf{TWEETS} & \textbf{GN} & \textbf{YELP}\\ \hline \hline
      \dmhashing & $1$ & $1$ & $1$ \\ \hline
      \dmlocation & $1$ & $1$ & $1$ \\ \hline
      \dmkeyword & $3.2$ & $3.3$ & $3.3$ \\ \hline
      \dmprefix & $1.9$ & $1.9$ & $2.1$\\ \hline
    \end{tabular}
\end{table}

\begin{figure*}[t]
	\centering
	\begin{minipage}[t]{\linewidth}	
		\centering
		\subfigure[TWEETS]{
	    	\label{fig:journal:exp2:tweets_bolt_num_throughput} 
	    	\includegraphics[width=0.27\columnwidth]{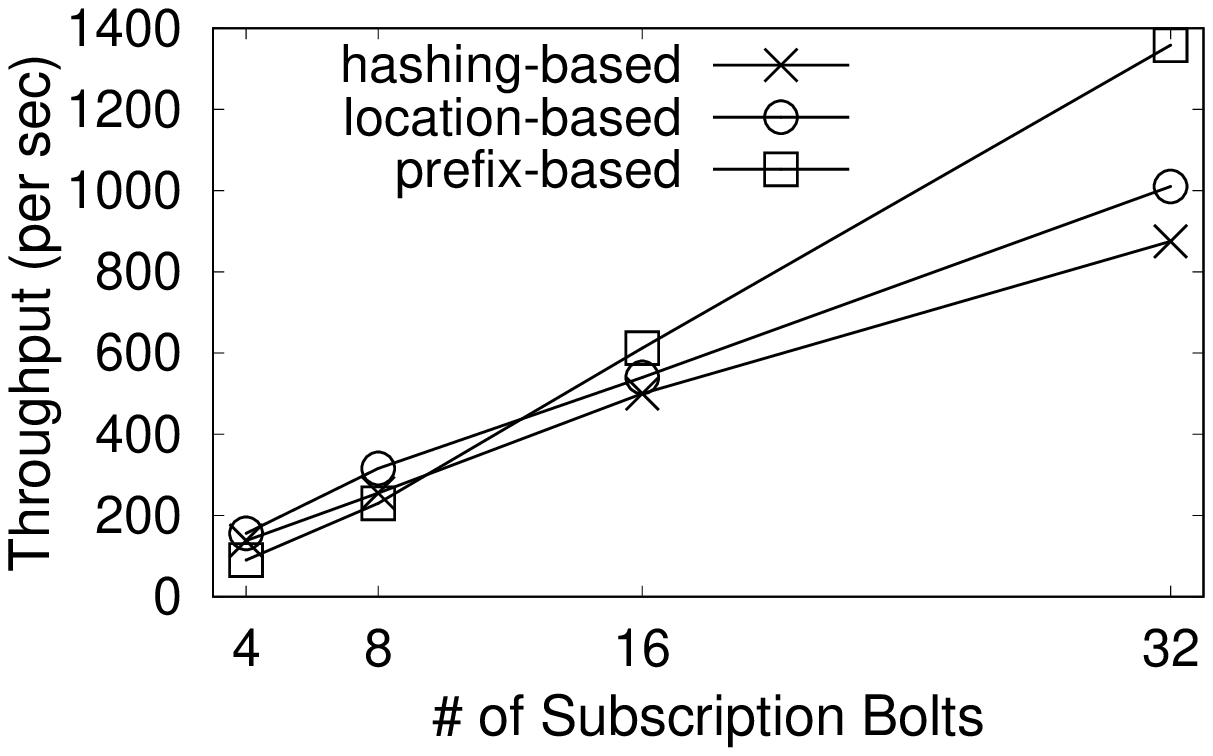}}
	    \hfill
	    \subfigure[GN]{
	    	\label{fig:journal:exp2:gn_bolt_num_throughput} 
	    	\includegraphics[width=0.27\columnwidth]{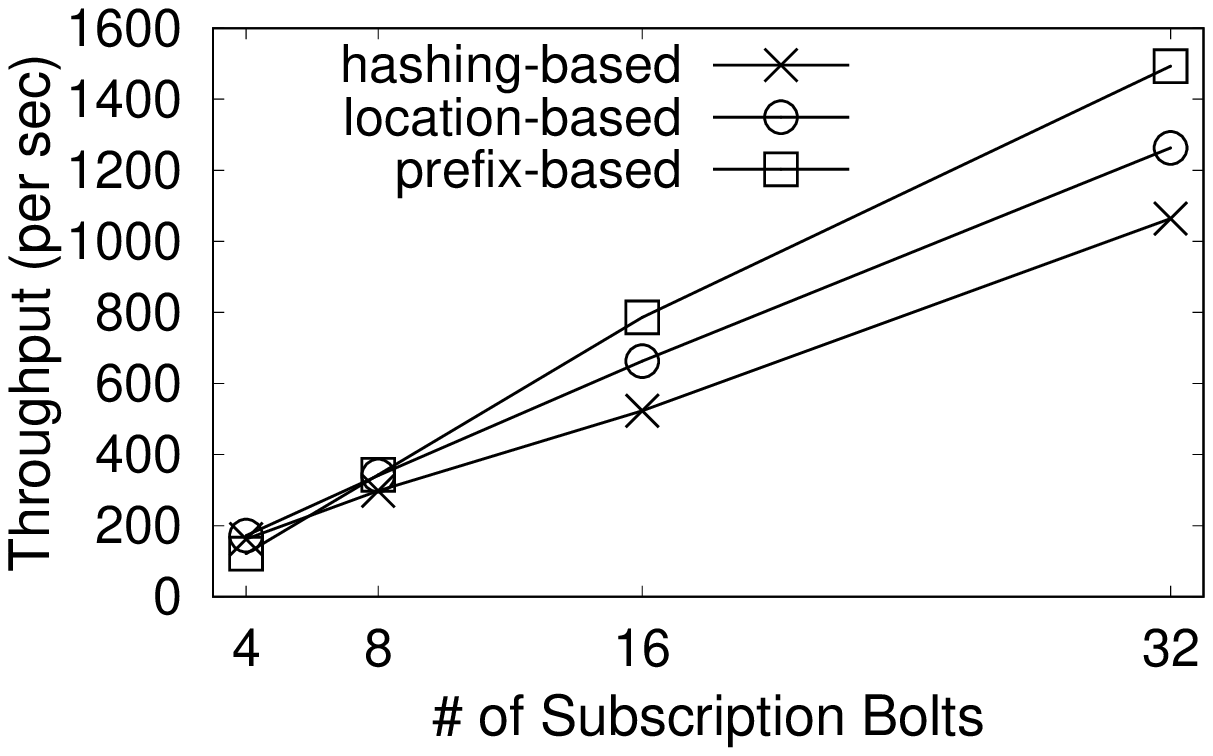}}
	    \hfill
	    \subfigure[YELP]{
	    	\label{fig:journal:exp2:yelp_bolt_num_throughput} 
	    	\includegraphics[width=0.27\columnwidth]{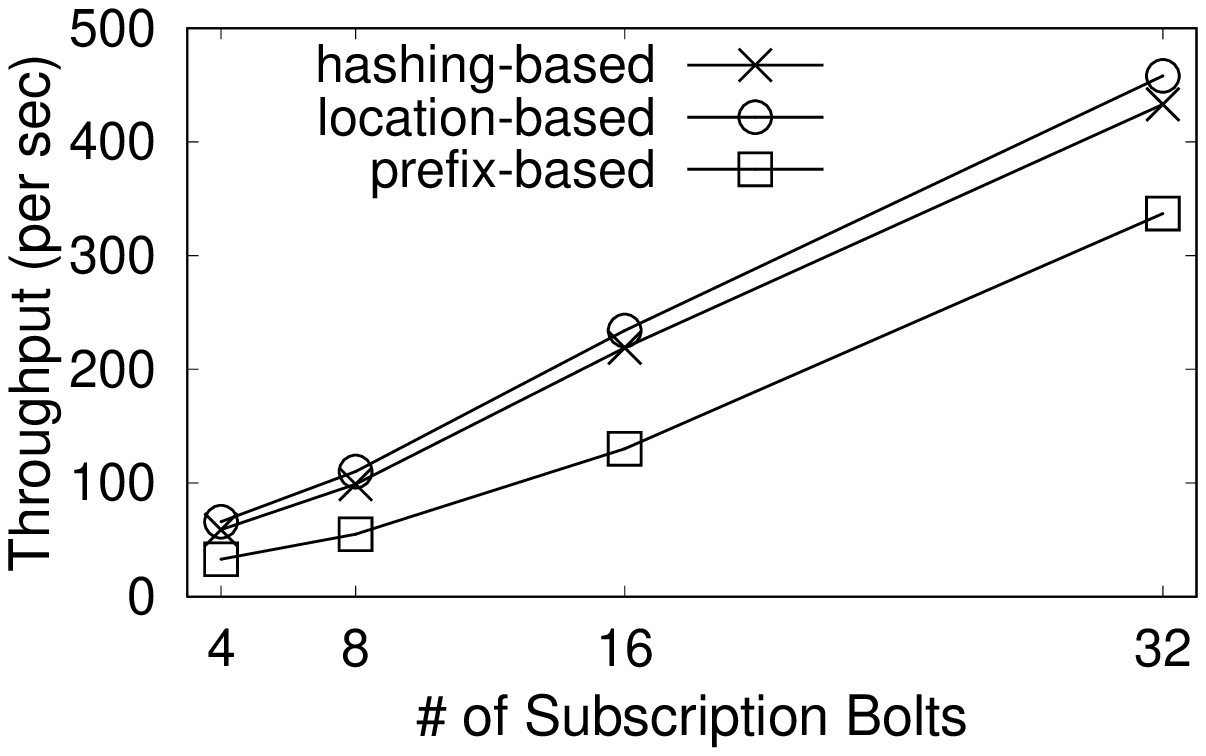}}
	
		\subfigure[TWEETS]{
		    \label{fig:journal:exp2:tweets_bolt_num_comm} 
		    \includegraphics[width=0.27\columnwidth]{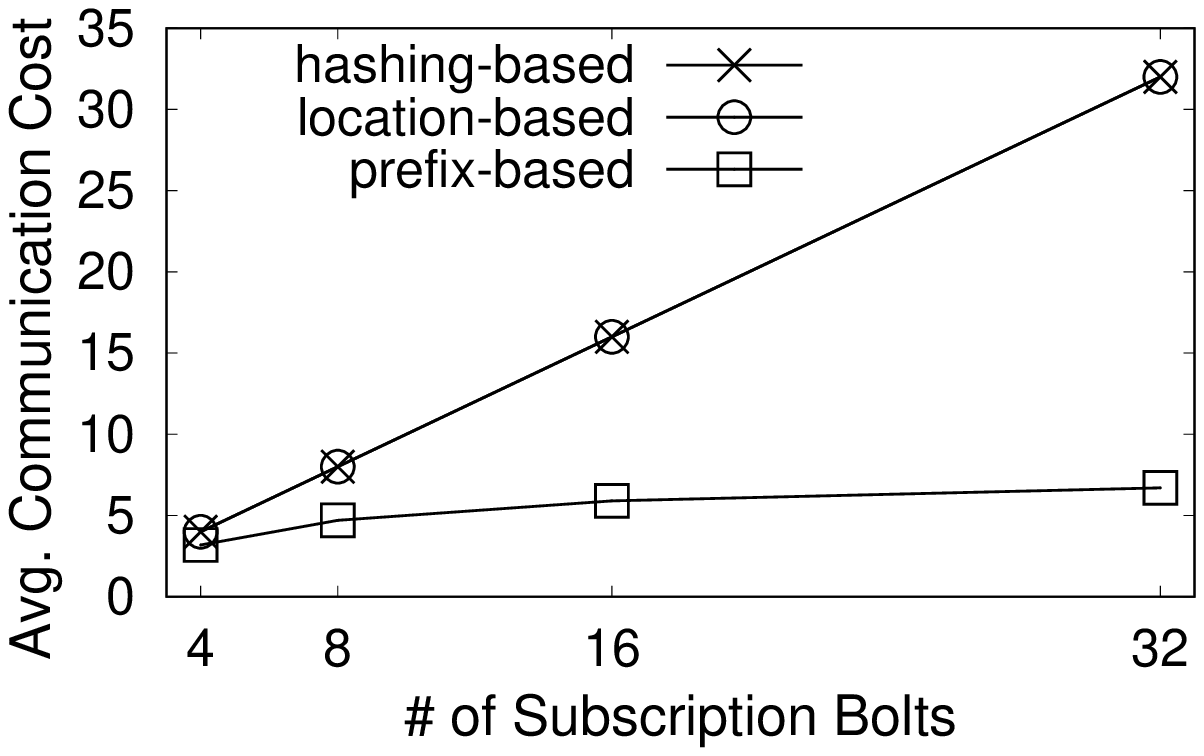}}
		\hfill
		\subfigure[GN]{
		    \label{fig:journal:exp2:gn_bolt_num_comm} 
		    \includegraphics[width=0.27\columnwidth]{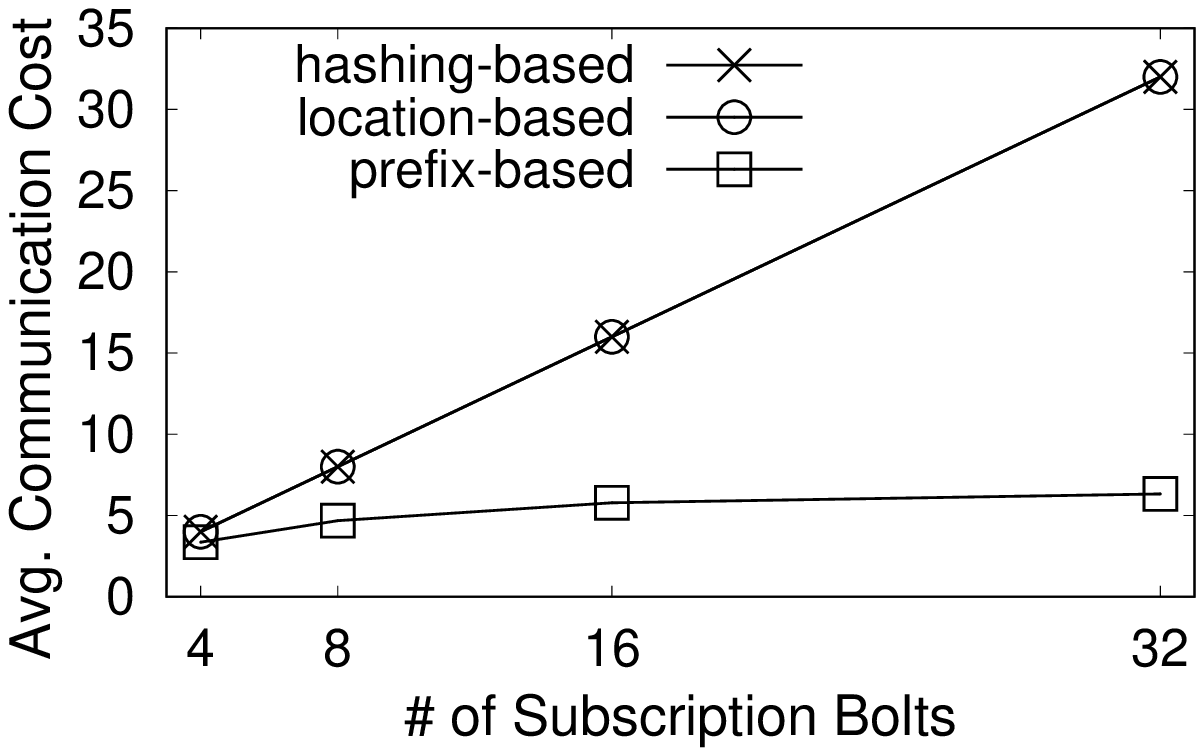}}
		\hfill
		\subfigure[YELP]{
		    \label{fig:journal:exp2:yelp_bolt_num_comm} 
		    \includegraphics[width=0.27\columnwidth]{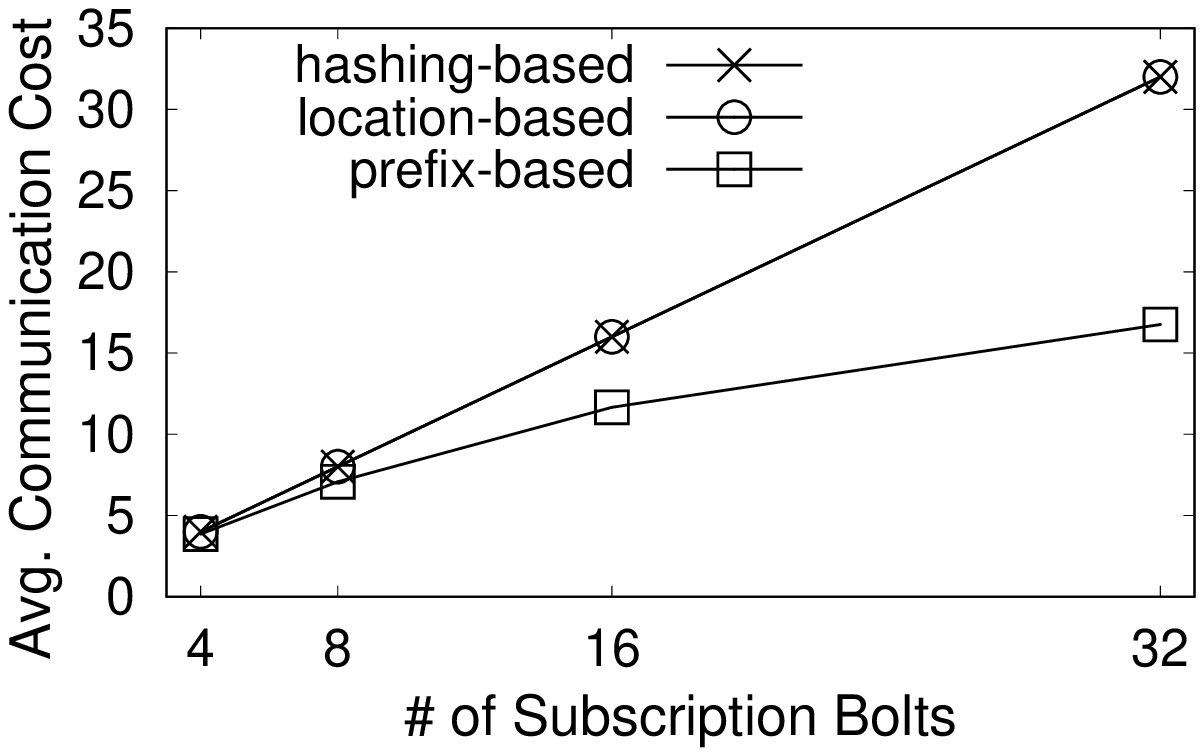}}
		\vspace{-3mm}
		\caption{Effect of number of subscription bolts}	
		\label{fig:journal:exp2:bolt_num}
	\end{minipage}%
\end{figure*}

\begin{figure*}[t]
	\centering
	\begin{minipage}[t]{\linewidth}	
		\subfigure[TWEETS]{
	    	\label{fig:journal:exp3:tweets_topk_throughput} 
	    	\includegraphics[width=0.27\columnwidth]{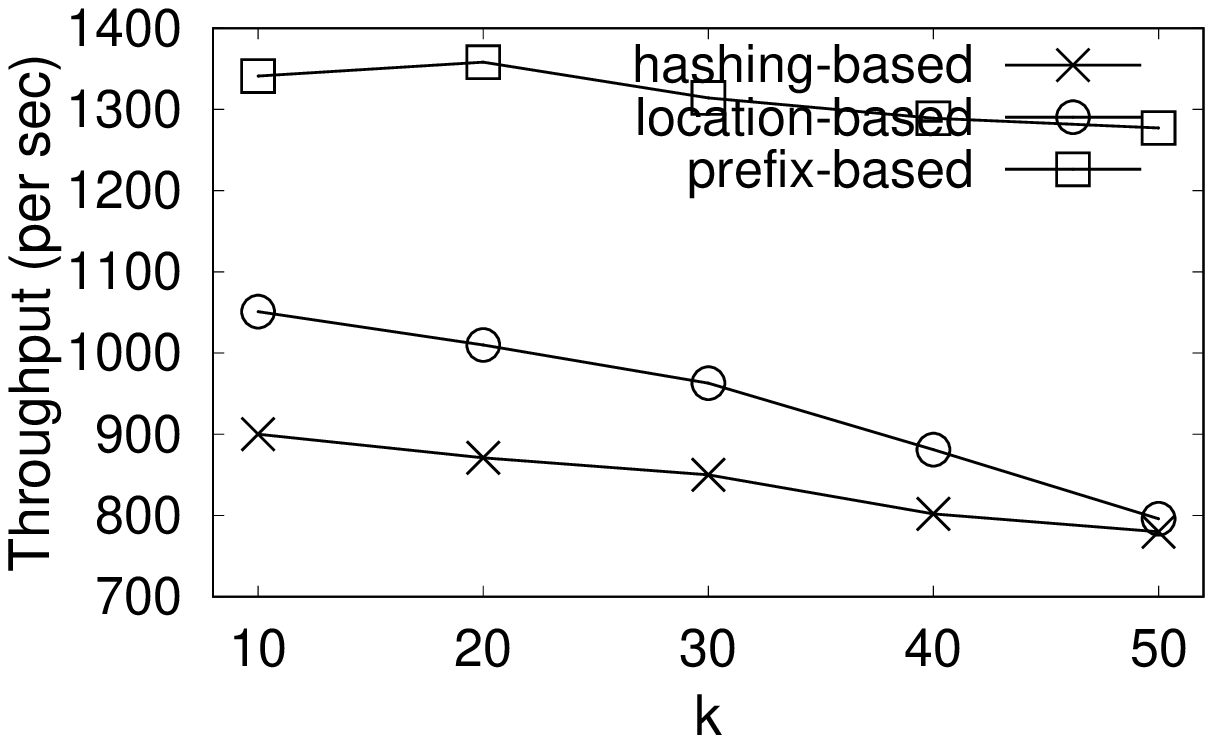}}
	    \hfill
		\subfigure[GN]{
		    \label{fig:journal:exp3:gn_topk_throughput} 
		    \includegraphics[width=0.27\columnwidth]{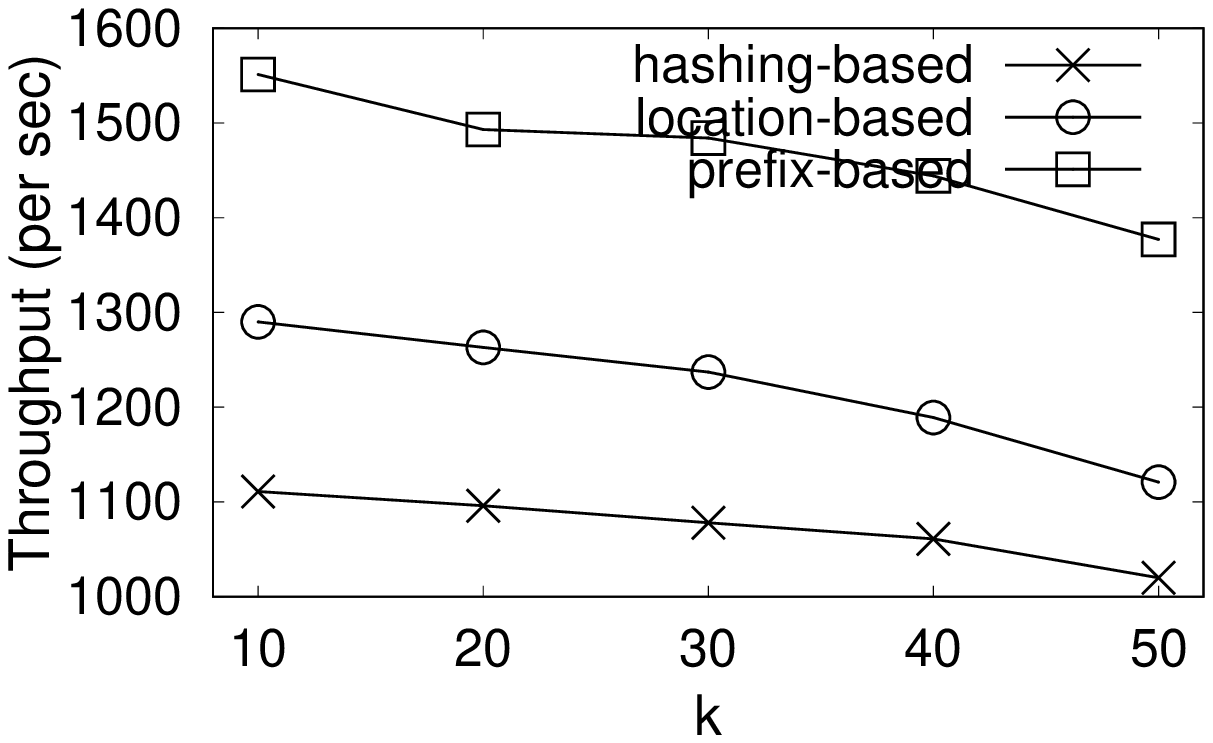}}
		\hfill
		\subfigure[YELP]{
		    \label{fig:journal:exp3:yelp_topk_throughput} 
		    \includegraphics[width=0.27\columnwidth]{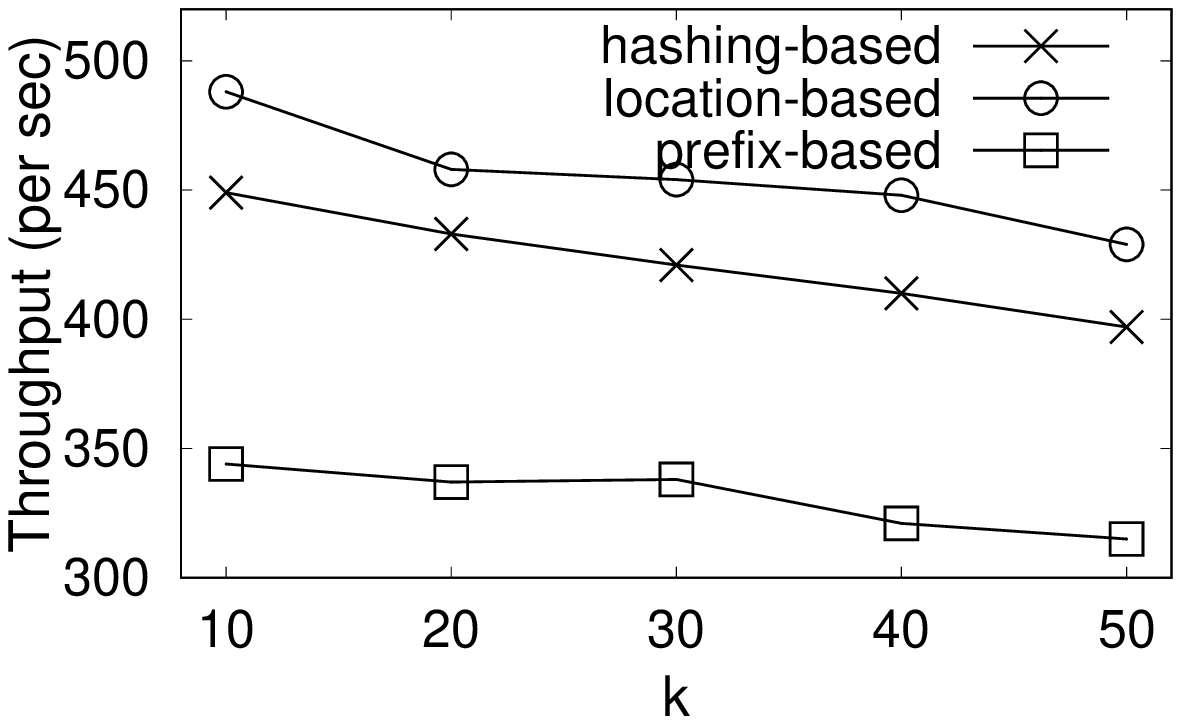}}
		\vspace{-3mm}		
		\caption{Effect of number of \topk results}	
		\label{fig:journal:exp3:topk}
	\end{minipage}%
\end{figure*}

\begin{figure*}[t]
	\centering
	\begin{minipage}[t]{\linewidth}	
		\centering
		\subfigure[TWEETS]{
	    	\label{fig:journal:exp4:tweets_scale_throughput} 
	    	\includegraphics[width=0.27\columnwidth]{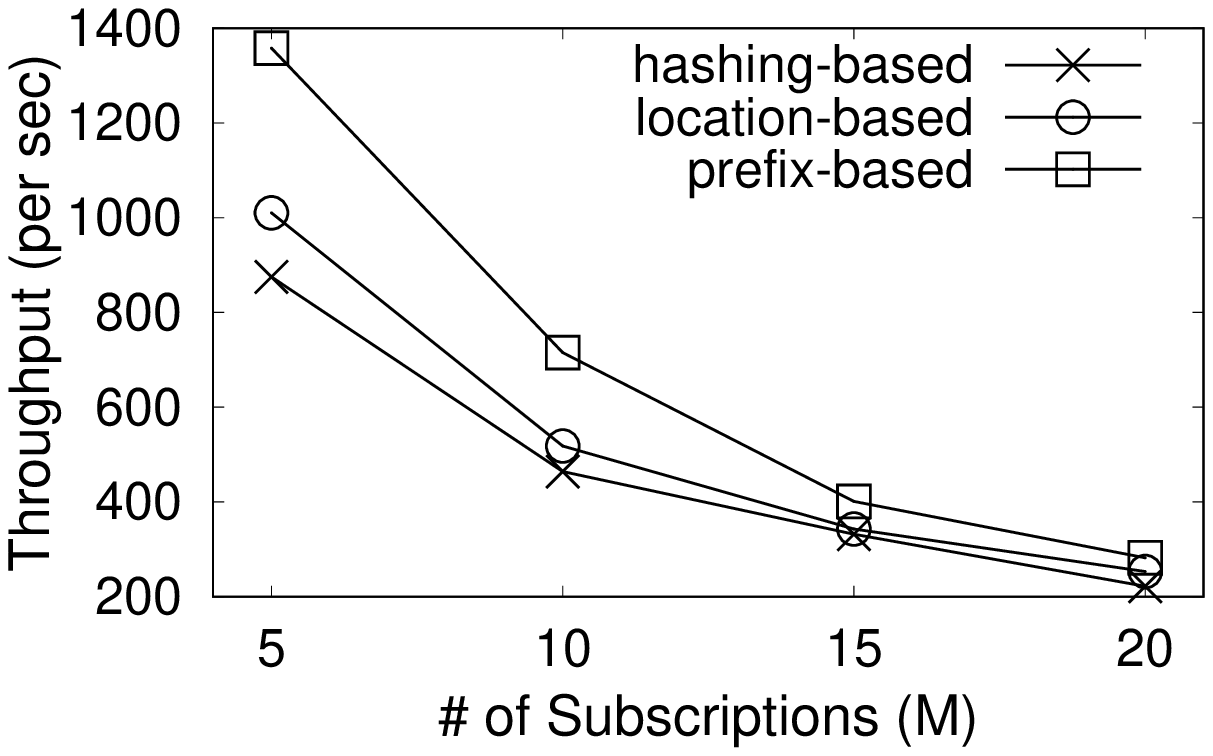}}
	    \hfill
		\subfigure[GN]{
		    \label{fig:journal:exp4:gn_scale_throughput} 
		    \includegraphics[width=0.27\columnwidth]{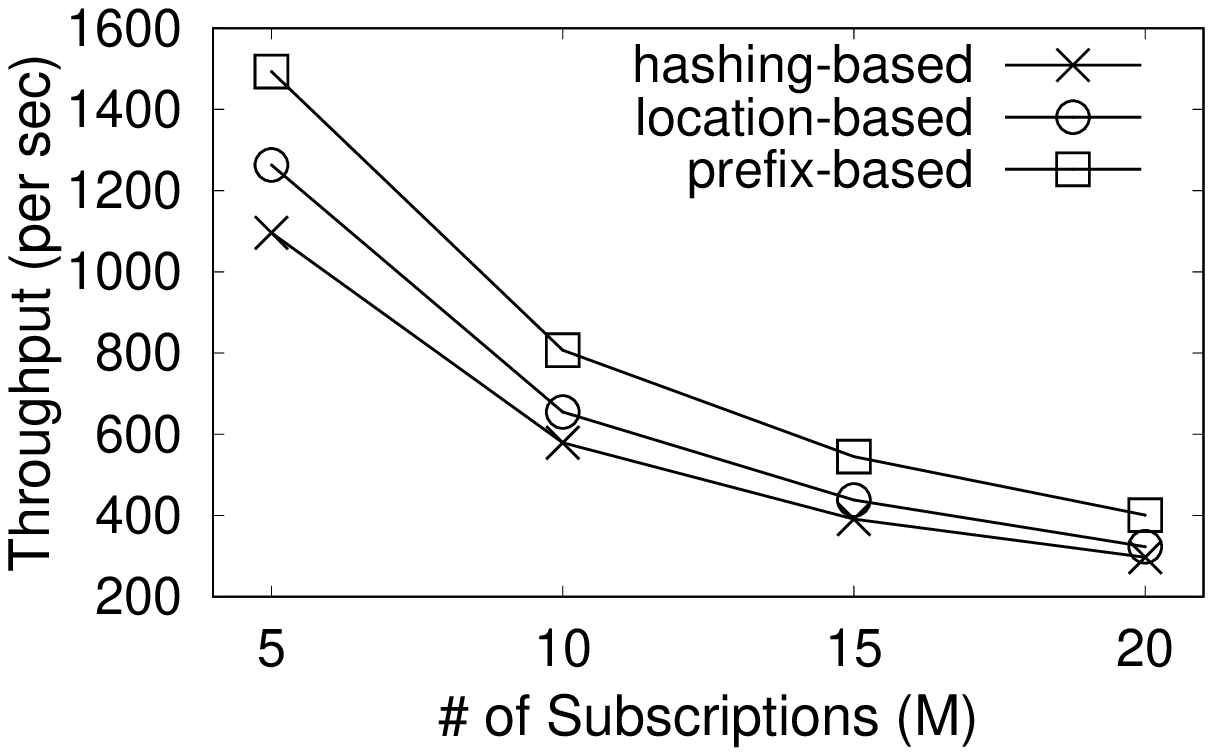}}
		\hfill
		\subfigure[YELP]{
		    \label{fig:journal:exp4:yelp_scale_throughput} 
		    \includegraphics[width=0.27\columnwidth]{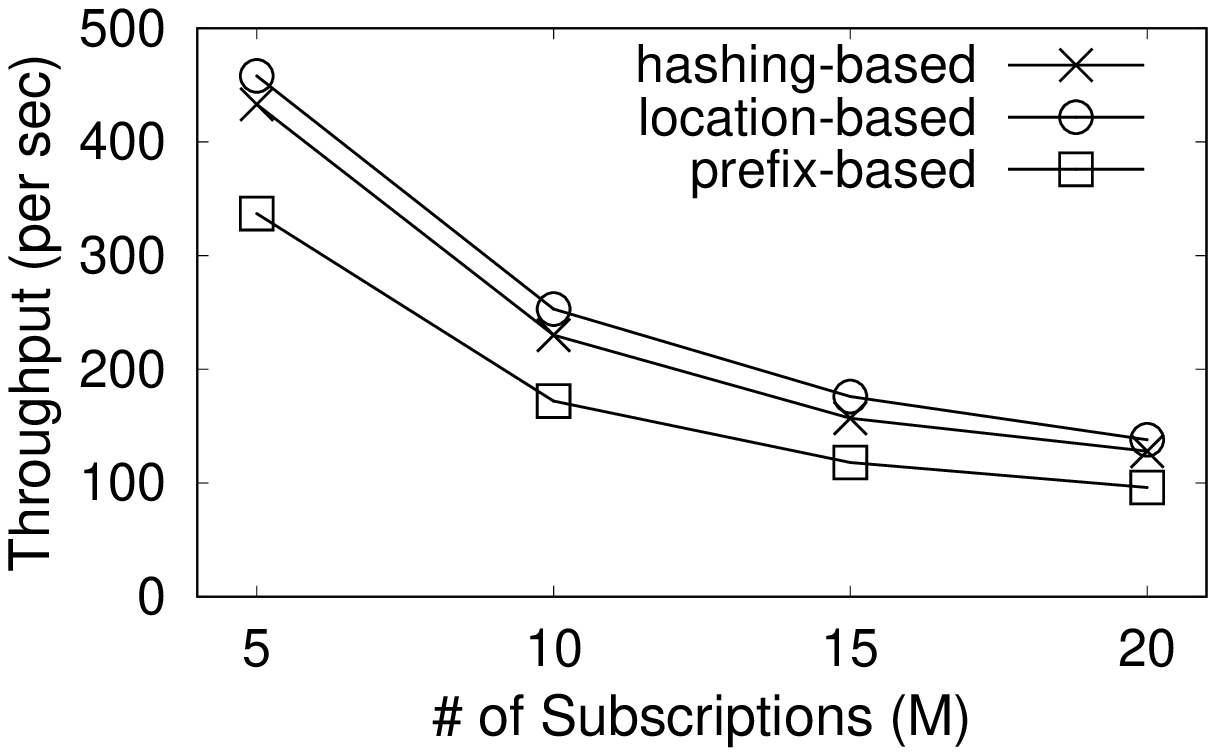}}
		\vspace{-3mm}		
		\caption{Effect of number of subscriptions}	
		\label{fig:journal:exp4:scale}
	\end{minipage}%
\end{figure*}

\begin{figure*}[t]
	\centering
	\begin{minipage}[b]{\linewidth}
		\centering
		\includegraphics[width=0.7\columnwidth]{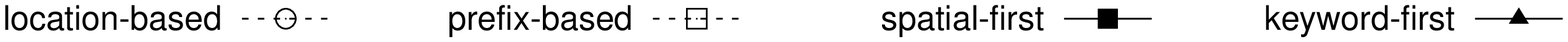}%
	 \end{minipage}	
	 \vfill
	\begin{minipage}[t]{1.0\linewidth}	
		\centering
		\subfigure[TWEETS]{
	    	\label{fig:journal:exp5:tweets_hybrid_throughput} 
	    	\includegraphics[width=0.27\columnwidth]{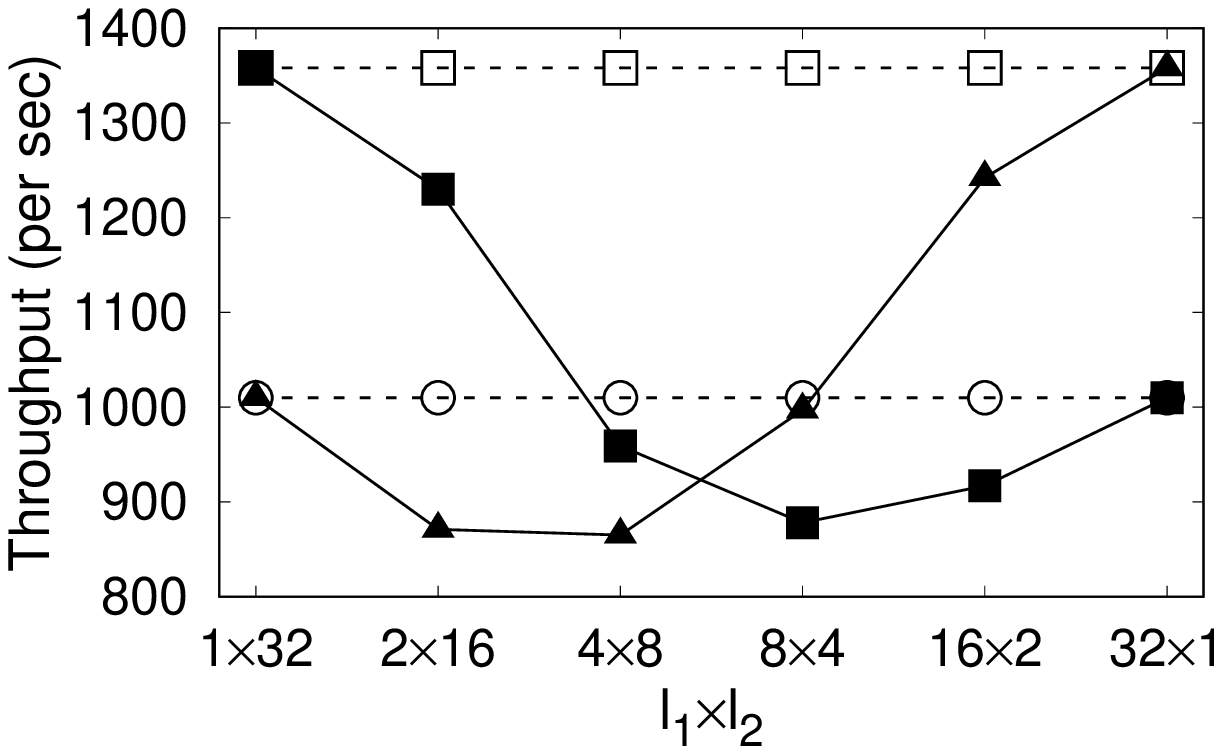}}
	    \hfill
	    \subfigure[GN]{
	    	\label{fig:journal:exp5:gn_hybrid_throughput} 
	    	\includegraphics[width=0.27\columnwidth]{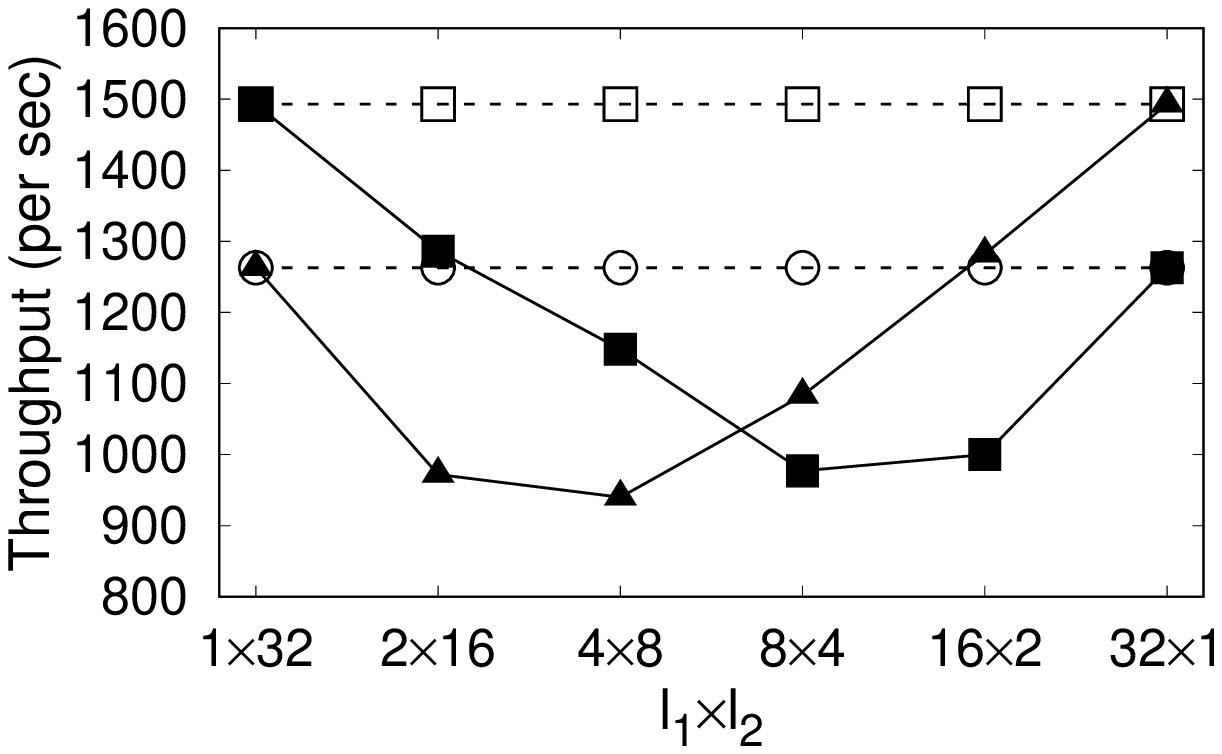}}
	    \hfill
		\subfigure[YELP]{
	    	\label{fig:journal:exp5:yelp_hybrid_throughput} 
	    	\includegraphics[width=0.27\columnwidth]{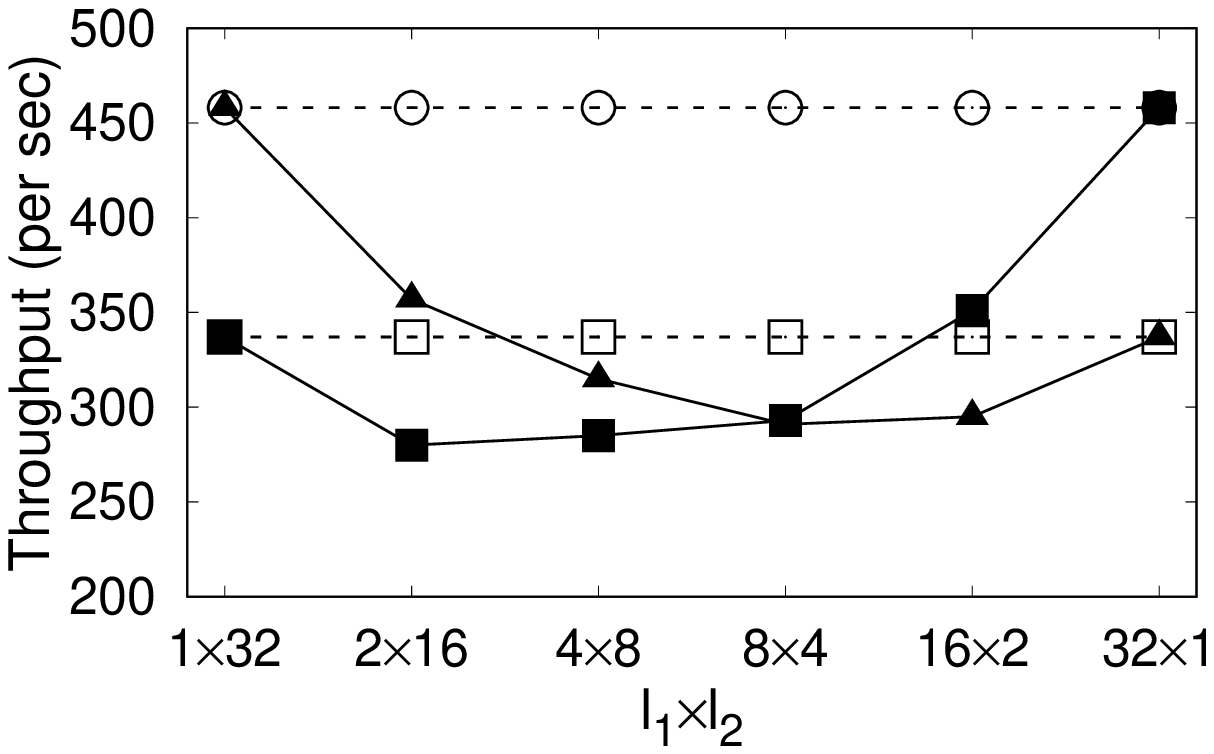}}
		\vspace{-3mm}		
		\caption{Throughput of hybrid methods}	
		\label{fig:journal:exp5:hybrid_throughput}
	\end{minipage}%
\end{figure*}

\begin{figure*}[t]
	\centering
	\begin{minipage}[b]{\linewidth}
		\centering
		\includegraphics[width=0.7\columnwidth]{journal_exp5_title.eps}%
	 \end{minipage}	
	 \vfill
	\begin{minipage}[t]{1.0\linewidth}	
		\centering
		\subfigure[TWEETS]{
	    	\label{fig:journal:exp5:tweets_hybrid_comm} 
	    	\includegraphics[width=0.27\columnwidth]{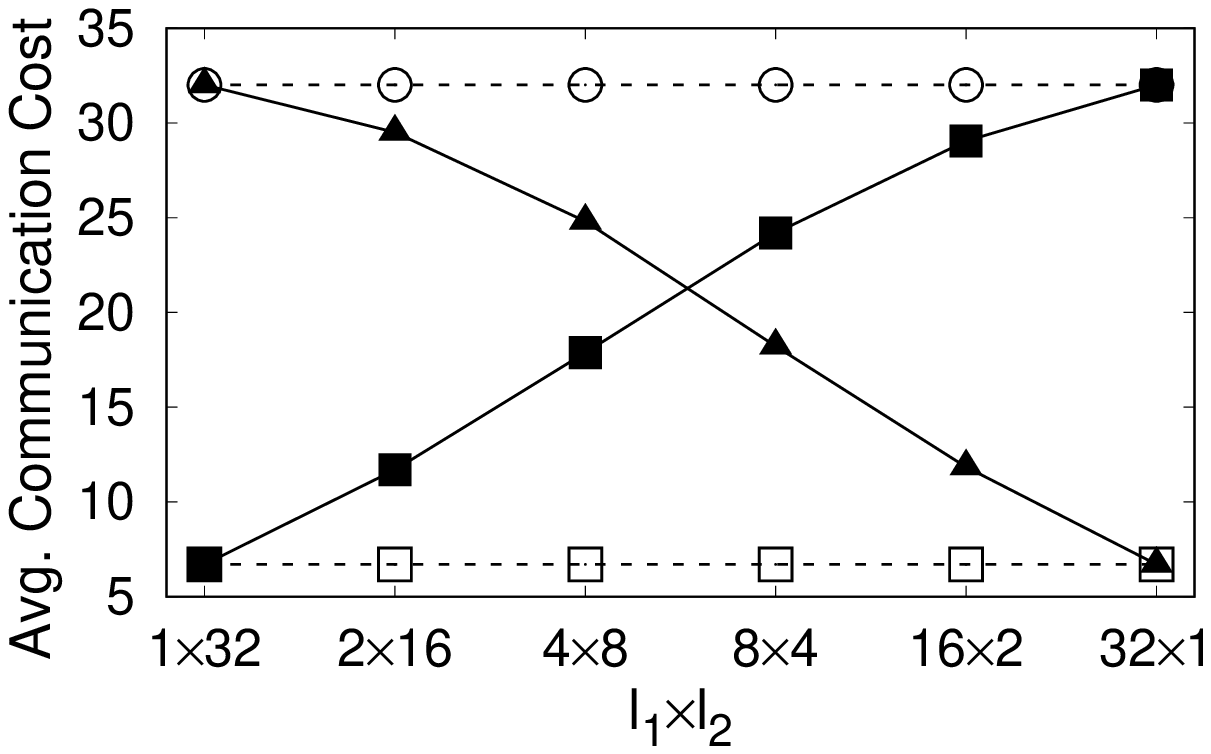}}
	    \hfill
	    \subfigure[GN]{
	    	\label{fig:journal:exp5:gn_hybrid_comm} 
	    	\includegraphics[width=0.27\columnwidth]{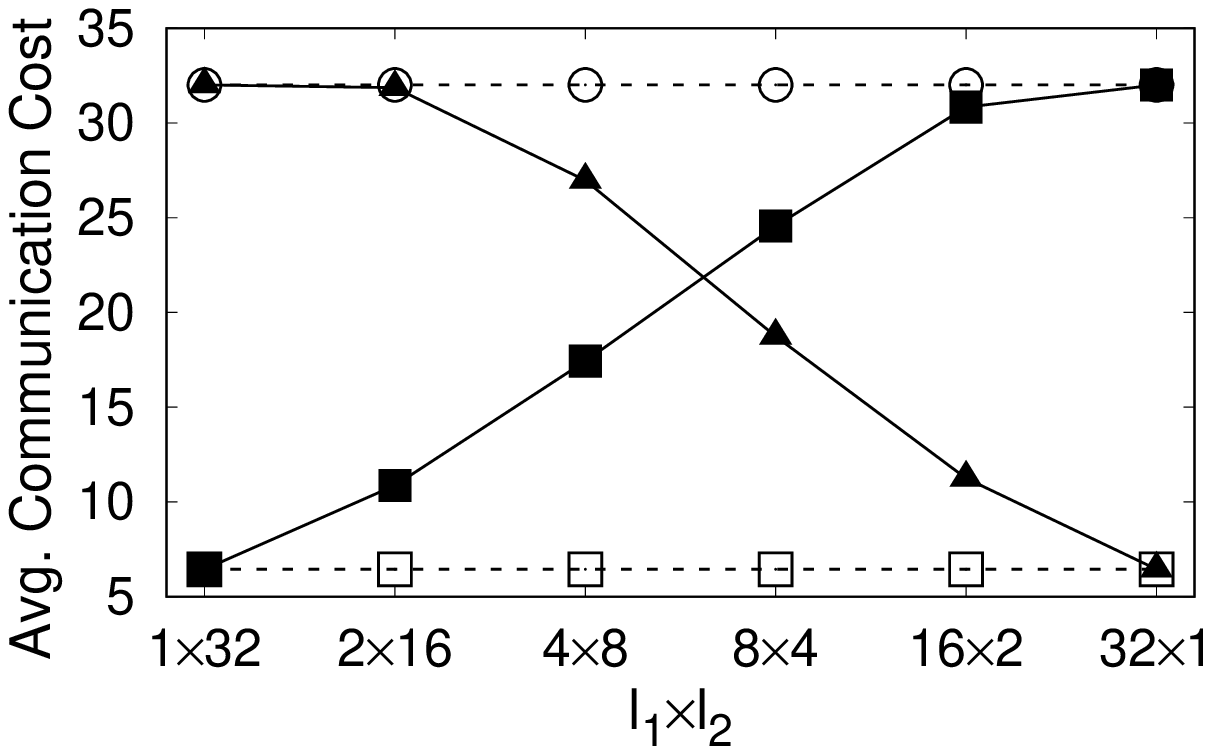}}
	    \hfill
		\subfigure[YELP]{
		    \label{fig:journal:exp5:yelp_hybrid_comm} 
		    \includegraphics[width=0.27\columnwidth]{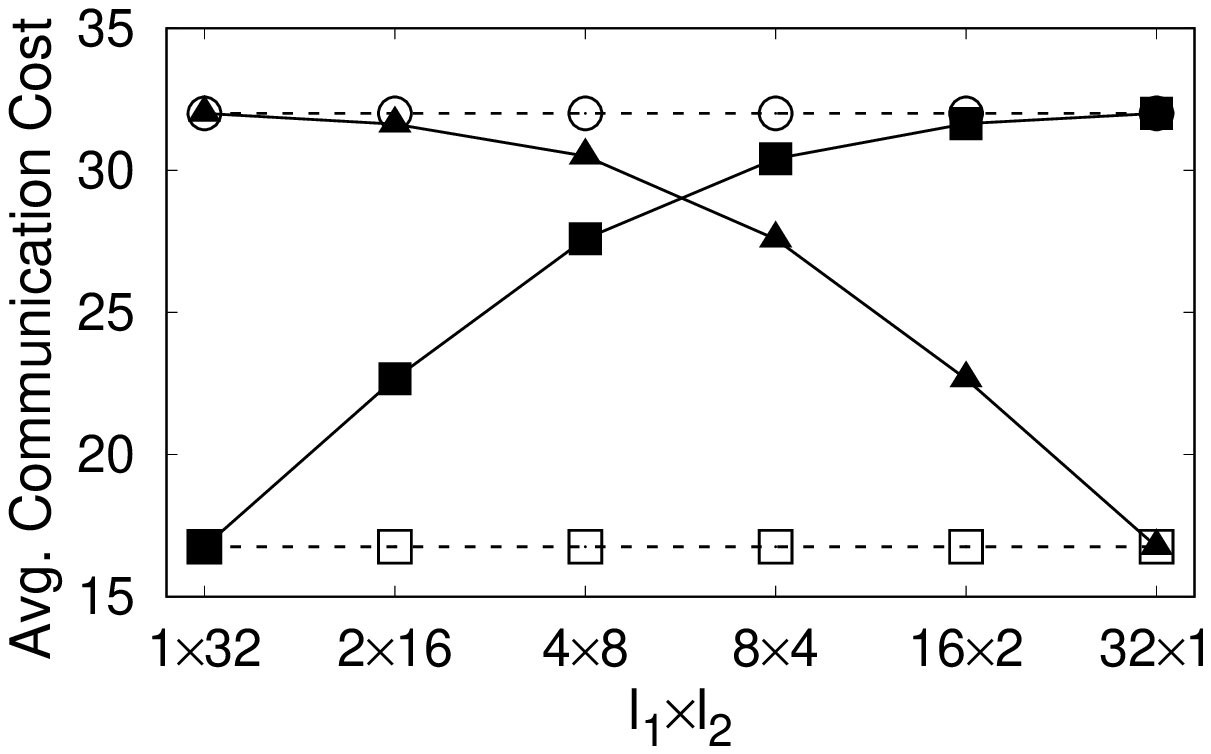}}
		\vspace{-3mm}
		\caption{Communication cost of hybrid methods}	
		\label{fig:journal:exp5:hybrid_comm}
	\end{minipage}%
\end{figure*}

\noindent \textbf{Compare different distribution mechanisms.}
In this set of experiments, we compare the performance of our proposed distribution mechanisms over \tweets, \gn and \yelp datasets.
We denote the four distribution mechanisms as \dmhashing, \dmlocation, \dmkeyword and \dmprefix respectively.
The results are depicted in Figure~\ref{fig:journal:exp1:dis_methods}.
Regarding \tweets dataset in Figure~\ref{fig:journal:exp1:dis_methods_throughput}, we observe that \dmlocation method achieves about $15\%$ throughput improvement than \dmhashing method.
This is mainly due to the better spatial similarity bound computed in \dmlocation method.
It is also noticed that \dmkeyword method performs the worst among all the methods because of the extra overhead incurred by duplicate subscription allocation, which is shown in Table~\ref{tb:exp_dis:replicate_ratio}.
However, \dmprefix method has the best performance, about $56\%$ faster than \dmhashing, at slightly extra cost of replication ratio (i.e., $1.9$).
Similar trend is also observed on \gn dataset.
We remark that the distributed system is about $26 \sim 49$ times faster than its centralized version, which demonstrates the superiority of distributed processing.
On the other hand, the communication cost of both \dmkeyword and \dmprefix methods in \tweets dataset (Figure~\ref{fig:journal:exp1:dis_methods_comm}) achieve about $80\%$ reduction compared to \dmhashing and \dmlocation methods, which demonstrates the benefit of former two methods, especially when the communication cost becomes system bottleneck.
Similar trends can also be observed from \gn dataset.
However, as to \yelp dataset, we notice that even though \dmprefix method is much faster than \dmkeyword method, it is slower than \dmhashing and \dmlocation methods.
This is mainly because the number of keywords in \yelp message is much larger than \tweets and \gn, and thus the reduction in communication cost (Figure~\ref{fig:journal:exp1:dis_methods_comm}) is not significant enough to pay off the cost contributed by duplicate allocation.
In the following experiments, we omit \dmkeyword method since \dmprefix method has much better overall performance in terms of both throughput and replication ratio.


\noindent \textbf{Effect of number of subscription bolts.}
We evaluate the effect of number of subscription bolts $N_{sb}$ in Figure~\ref{fig:journal:exp2:bolt_num} by varying $N_{sb}$ from $4$ to $32$.
In terms of throughput, all the three algorithms can scale very well with increasing number of subscription bolts.
We observe that when $N_{sb}$ is small (i.e., $4$ or $8$), \dmprefix method performs worse than \dmhashing and \dmlocation methods due to the duplicate subscriptions in \tweets and \gn (Figure~\ref{fig:journal:exp2:tweets_bolt_num_throughput} and Figure~\ref{fig:journal:exp2:gn_bolt_num_throughput}).
However, when $N_{sb}$ reaches $16$ or $32$, the benefit of communication cost reduction has become significant (Figure~\ref{fig:journal:exp2:tweets_bolt_num_comm} and Figure~\ref{fig:journal:exp2:gn_bolt_num_comm}), thus contributing to the high throughput of \dmprefix method.
For \yelp dataset, \dmprefix method processes less messages per second compared to \dmhashing and \dmlocation methods while managed to reduce communication by about $50\%$.

\noindent \textbf{Effect of number of \topk results.}
We turn to evaluate the effect of number of \topk results in Figure~\ref{fig:journal:exp3:topk}.
We do not show the communication cost because it is irrelevant to the number of \topk results.
As shown in Figure~\ref{fig:journal:exp3:topk}, the throughputs in all the datasets decrease when we vary $k$ from $10$ to $50$.
This is because large $k$ usually yields high processing cost in subscription index.
However, the influence of $k$ is not very significant as the throughputs decrease slowly and linearly.

\noindent \textbf{Effect of number of subscriptions.}
In Figure~\ref{fig:journal:exp4:scale}, we vary the number of subscriptions from $5M$ to $20M$.
It is obvious that the throughputs drop in all the datasets when we increase the number of subscriptions.
However, the decreasing trend indicates that the average processing time of an incoming message is still linearly proportional to the number of subscriptions, considering the factor that the average processing time is inversely proportional to the throughput.

\noindent \textbf{Hybrid methods.}
In this set of experiments, we compare our methods with two possible hybrid methods: \sfirst method and \kfirst method.
Specifically, \sfirst method employs two-level partition scheme.
On the first level, it employs \dmlocation method while on the second level it employs \dmprefix method.
\kfirst method is similar to \sfirst method except that it employs \dmprefix method first and then \dmlocation method.
We compare the throughput and communication cost of these two hybrid methods while changing the number of partitions on each level.
The results of \dmlocation and \dmprefix methods are also shown in dotted line for comparison purpose.
Regarding throughput in Figure~\ref{fig:journal:exp5:hybrid_throughput} where $l_1$ is the number of partitions on the first level and $l_2$ is the number of partitions on the second level, all the datasets exhibit similar trends.
Specifically, for \sfirst method, when $l_1$ increases and $l_2$ decreases, the benefit of keyword partition becomes less while the benefit of spatial partition is still not significant, which leads to decreasing throughput.
As we further increase $l_1$ (e.g., $16 \times 2$), the spatial partition becomes dominant and thus improves the throughput which finally reaches the same throughput as \dmlocation method at $32 \times 1$.
The trends of \kfirst method can be explained similarly.
The communication costs of both \sfirst and \kfirst methods in Figure~\ref{fig:journal:exp5:hybrid_comm} are between those of \dmlocation and \dmprefix due to its hybrid nature.
Overall, the hybrid methods do not exhibit large improvement compared to \dmlocation and \dmprefix methods.

\section{Conclusion}
\label{sec:conclusion}
The popularity of streaming \gt data offers great opportunity for applications such as information dissemination and location-based campaigns.
In this paper, we study a novel problem of continuous \topk \sk \pubsub over sliding window.
To maintain \topk results for a large number of subscriptions over a fast stream simultaneously and continuously, we propose a novel indexing structure, which employs both individual pruning technique and group pruning technique, to process a new message instantly on its arrival.
In addition, to handle the re-evaluations incurred by expired messages from the sliding window, we develop a novel cost-based \skyband model with theoretical analysis to judiciously maintain a partial \skyband for each subscription.
Furthermore, a distributed stream processing system called \oursdis is developed, which extends \ours to \storm to exploit the benefit of parallel processing.
The experiments demonstrate that both \ours and its distributed version \oursdis can achieve high throughput performance over \gt stream.

\vspace{-3mm}

\bibliographystyle{spmpsci}      
\bibliography{ref}   

\end{document}